\documentclass[a4paper,onecolumn,11pt]{quantumarticle}
\pdfoutput=1
\usepackage{amsfonts,amsmath,amsthm,amssymb,amscd}
\usepackage[utf8]{inputenc}
\usepackage[english]{babel}
\usepackage[T1]{fontenc}
\usepackage[colorlinks=true,
            linkcolor=black,
            citecolor=black,
            urlcolor=black]{hyperref}
\usepackage{cite}
\usepackage{graphicx,tikz,floatflt}
\usepackage{blkarray}
\usepackage{comment}
\usepackage{mathtools}

\DeclarePairedDelimiter\floor{\lfloor}{\rfloor}

\newcommand{\ket}[1]{\ensuremath |{#1}\rangle}
\DeclareMathOperator{\Tr}{Tr}
\usepackage{authblk}

\newtheorem{theorem}{Theorem}[section]

\newtheorem{lemma}[theorem]{Lemma}
\newtheorem{proposition}[theorem]{Proposition}
\newtheorem{corollary}[theorem]{Corollary}

\theoremstyle{definition}
\newtheorem{definition}[theorem]{Definition}
\newtheorem{example}[theorem]{Example}

\theoremstyle{remark}
\newtheorem{remark}[theorem]{Remark}

\title{Quantum Lego Power-up: Designing Transversal Gates with Tensor Networks}

\author[1,2]{ChunJun Cao}
\author[3]{Brad Lackey}

\affil[1]{Department of Physics, Virginia Tech, Blacksburg, VA, USA 24061}
\affil[2]{Virginia Tech Center for Quantum Information Science and Engineering, Blacksburg, VA 24061, USA}
\affil[3]{Microsoft Quantum, Redmond, WA, USA}

\date{}
\begin{document}
 
\maketitle

\begin{abstract}
Transversal gates are the simplest form of fault-tolerant gates and are relatively easy to implement in practice. Yet designing codes that support useful transversal operations---especially non-Clifford or addressable gates---remains difficult within the stabilizer formalism or CSS constructions alone. We show that these limitations can be overcome using tensor-network frameworks such as the quantum lego formalism, where transversal gates naturally appear as global or localized symmetries. Within the quantum lego formalism, small codes carrying desirable symmetries can be ``glued'' into larger ones, with operator-flow rules guiding how logical symmetries are preserved. This approach enables the systematic construction of codes with \emph{addressable} transversal single- and multi-qubit gates targeting specific logical qubits regardless of whether the gate is Clifford or not. As a proof of principle, we build new finite-rate code families that support strongly transversal $T$, $CCZ$, $SH$, and Gottesman's $K_3$ gates---structures that are challenging to realize with conventional methods. We further construct holographic and fractal-like codes that admit addressable transversal inter-, meso-, and intra-block $T$, $CS$, and $C^\ell Z$ gates. As a corollary, we demonstrate that the heterogeneous holographic Steane-Reed-Muller black hole code also supports fully addressable transversal inter- and intra-block $CZ$ gates, significantly lowering the overhead for universal fault-tolerant computation.
\end{abstract}

\newpage
\tableofcontents

\newpage
\section{Introduction}

The stabilizer and tensor network (TN) formalisms are two influential and yet complementary frameworks in quantum information. In quantum coding theory, the correspondence between classical codes and quantum stabilizer codes such as the CSS construction have given the stabilizer formalism a clear advantage for code design, error simulation, and decoding. By contrast, tensor networks has been conspicuously absent in quantum code design and quantum coding theory where it is mostly relegated to tackling computational problems, such as decoding, especially for systems subjected to non-Pauli errors.  

Recent progress, however, highlights the broader potential of tensor networks in designing and characterizing quantum error correcting codes (QECCs) \cite{Ferris_2014}. Developments such as holographic codes~\cite{Pastawski_2015,Harris_2018} and tensor network codes~\cite{tnc,Farrelly_2022} show that TNs can provide new graphical intuition for constructing quantum codes. More generally, the quantum lego (QL) framework~\cite{QL,QL2} proposes a quantum-first approach where codes are built directly from smaller quantum code building blocks, rather than by ``quantizing'' classical codes. These methods extend naturally beyond Pauli stabilizer codes~\cite{Shen:2023xmh,Kukliansky:2024rvz,Cao:2022ybv}, enabling new tools such as efficient evaluations of weight enumerator polynomials~\cite{TNenum,QL2} and automated code discovery \cite{Su_2025,Mauron:2023wnl}. This work introduces a new application of tensor network methods: the construction of transversal logical gates, with a special focus on non-Cliffordness and addressability, a task that has remained challenging within stabilizer or CSS-based frameworks.

Transversal gates are depth-one circuits that provide one of the simplest and most practical forms of fault-tolerant logic. They are easy to implement on quantum hardware and incur significantly lower overhead than surgery-based techniques. Although transversal gates alone cannot yield universality~\cite{Eastin_2009}, identifying codes with transversal non-Clifford or addressable gates is central to scalable fault-tolerant architectures, whether combined with magic state distillation, code switching, pieceable fault tolerance, or heterogeneous concatenation. Yet discovering codes with non-trivial transversal gates is difficult. For example, quantum codes supporting transversal $T$ gates impose strong constraints on classical codes ~\cite{Rengaswamy_2020,Rengaswamy_20202,Camps_Moreno_2024}, which poses a serious challenge for classical coding theory. Far less is known for gates such as $SH$, $C^\ell P(\varphi)$, or higher-level Clifford hierarchy operations, but recent work have made progress in identifying codes with transversal $C^\ell Z$ gates~\cite{Nguyen:2024qwg,he2025quantumcodesaddressabletransversal,He:2025zbj,Zhu25,Guemard:2025asu,Golowich:2024ogv}. The difficulty grows further for codes with multiple logical qubits, such as general Low-Density Parity-Check (LDPC) codes, where logical operations must be targeted to selected qubits rather than applied indiscriminately~\cite{Fu:2025lbb,Guyot:2025lyj}.  

Tensor networks offer a natural language for addressing these challenges. When codes are expressed as TNs, transversal gates correspond to the unitary symmetries, and code design reduces to engineering networks with appropriate global or local symmetries. The QL framework realizes this idea by gluing together “atomic” codes that already support desirable gates. Using operator flow, these symmetries are then extended to the TN while preserving both transversality and addressability. In this way, the TN representation provides a graphical blueprint in which check operators, logical gates, and fault-tolerance properties emerge transparently. Recent applications have already produced families of codes with diverse geometries, including those with sparse parity checks reminiscent of quantum LDPC codes~\cite{QL,Cao:2025oep}.  

In this work, we extend QL to systematically construct codes with exotic transversal gates and fully addressable multi-qubit operations. By generating TNs with localized symmetries, we show how to realize Clifford and non-Clifford transversal gates, as well as addressable gates across multiple blocks. To enable this, we generalize the gluing procedure to localize compatible symmetries, and develop gadgets that propagate these symmetries through the network. Using these tools, we construct new code families and identify addressable transversal multi-qubit gates in holographic and iterated fractal codes.

The remainder of the paper is organized as follows. In Sec 2, we review the QL framework and explain how transversal single or multiqudit gates are mapped to symmetries of the system. In Sec. 3, we then introduce how these symmetries can propagate under generalized traces which are necessary for producing codes that produce (addressable) transversal gates that are non-Pauli. We also develop useful gadgets and local moves that will be useful for producing codes with addressable gates. In Sec.4.1, we apply this formalism and general guidelines to produce codes with global transversal gates, which can be relevant for quantum resource theory and state distillation. We construct a family of $[[3L+2,L,3]]$ and $[[L^2+4L, L^2,3]]$ codes with strongly transversal $SH$ or $K_3$ gates. We also show how finite rate codes with transversal $CCZ$ gates, $T$ gates, and their mixtures, can be constructed using Pauli deformed traces. We then generalize these constructions and prove that by gluing  Maximum-Distance-Separable (MDS) codes into a network that resembles Euclidean lattices, one can produce codes with strongly transversal gates and asymptotically unit encoding rate. 
In Sec. 4.2, we show how operator pushing can generate gates with more addressability, especially interblock and intrablock multi-qubit gates. We focus on two classes of codes as proof-of-principle examples, the holographic codes and iterated fractal codes. For the former, we consider finite rate holographic codes which can encode bulk qubits at different radii in the dual hyperbolic geometry, and black hole codes which encodes many logical qubits but they all live on the horizon. For the finite rate holographic codes, we show that interblock and intrablock multiqubit gates are addressable on bulk qubits in the same layer, i.e. roughly the same distance away from the center of the bulk. Gates that act on disjoint bulk qubits are also parallelizable. For the black hole codes, as the bulk qubits already lie in the same layer, all logical qubits can be acted upon by addressable gates. With suitable atomic legos with distance 3 or greater when shortened, these transversal gates are also fault-tolerant. For iterated fractal codes, all logical qubits are fully addressable. By iterating a partial concatenation procedure, one can show that these codes satisfy $[[n, O(n^{\alpha}), O(n^{\beta})]]$ scaling where $\alpha,\beta$ depend on the parameters of the atomic lego blocks and the branching ratio in the fractal tensor network. Similarly, fault tolerance can be achieved by choosing suitable lego blocks. Codes that permit simultaneous fully addressable $T, CS$ and  $CCZ$ gates can also be built with $\alpha=\beta\approx 0.224$. Although we mostly focus on qubit codes in our applications, the formalism applies equally well for qudit codes.

Other technical results and detailed examples are deferred to the Appendices.  

\section{Quantum Lego, symmetries, and transversal gates}
\label{sec:2}
Quantum Lego \cite{QL,QL2,Shen:2023xmh,TNenum,Cao:2025oep} is a tensor-network formalism tailored to the design and analysis of quantum error-correcting codes. The framework has two main components: (i) identifying and characterizing small “lego blocks,” namely codes or states with useful structures such as transversal gates, stabilizers, or favorable distances; and (ii) developing systematic moves for combining these blocks, allowing one to predict the properties of the resulting code—such as distance, rate, sparsity, and available fault-tolerant logical gates.

A quantum lego block is a tensor $\mathbf{V}$ which can be expanded into its coefficients $V_{i_1i_2\dots n_r}$ in a given basis. In its conventional graphical representation, a dangling leg is assigned to each of its free index (Fig.~\ref{fig:qlego}a). As shown in \cite{QL}, the specific role that the tensor plays can vary depending on how it is interpreted, where by contracting $V_{i_1 i_2\dots}$ with different vector or matrix bases, we can obtain an object that ranges from a state, to linear maps that can represent encoding isometries, projections, or co-isometries. Therefore, the tensor in QL should be treated as a fundamental object, not unlike a stem cell, that can differentiate into states, maps, or codes depending on how its indices or legs are treated. For notational purposes, we call a leg logical or physical depending on whether it is associated with an input or output degree of freedom in the map. For example, suppose all the legs have bond dimension 2 then the tensor corresponds to an object over qubits. If all $r$ legs are chosen as physical legs, then it represents a $r$-qubit state $|V\rangle$. If $k<r$ legs are used as logical legs and $r-k$ as physical legs, then it becomes a linear map $V$ whose exact interpretation, i.e., encoding isometry, projection, non-isometry, depend on the specific structure of the tensor. If $V$ is a map, then we often call $|V\rangle$ its Choi state, to be consistent with the notation in channel-state duality. 

\begin{figure}
    \centering
    \includegraphics[width=0.8\linewidth]{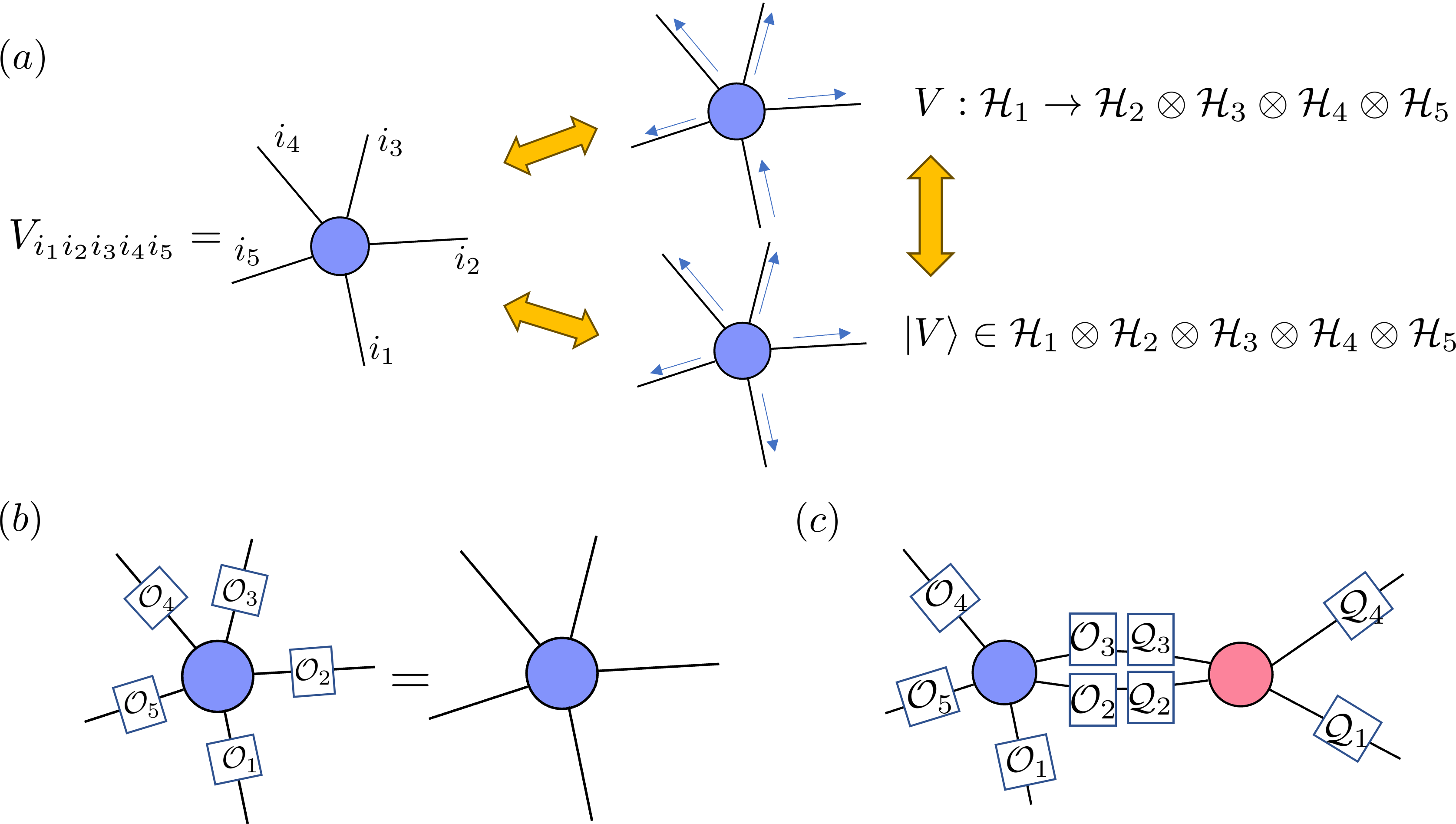}
    \caption{(a) A tensor as a quantum lego block. A tensor can represent a state or a map depending on how the legs are assigned. In-going arrows mark input/logical legs while out-going arrows mark output/physical legs. (b) Unitary symmetry of a tensor. (c) Symmetries of the tensor network is generated by operator matching.}
    \label{fig:qlego}
\end{figure}

Contraction of tensors allows us to generate bigger codes from smaller components. Pictorially, one connects the legs of the tensors corresponding to the indices that are contracted. A familiar example in quantum error correction is code concatenation, where the legs associated to the logical qubits of the inner code are contracted with legs that represent physical qubits of the outer code in a tree-like fashion \cite{Ferris_2014}. Of course, more general tensor contractions are also permitted which go beyond code concatenation \cite{Farrelly_2022,QL}.

In principle, any tensor can be contracted to create bigger codes or states but a lego block is most useful when it has a large number of symmetries. Namely, symmetries allow us to efficiently capture the process and outcome of tensor contractions. For instance, when the tensor is derived from a stabilizer code\footnote{or a generalization of stabilizer codes whose stabilizers need not be an Abelian subgroup of the Pauli group \cite{Shen:2023xmh}} the symmetries of the stabilizer group efficiently and completely characterize each lego block; moreover, the contractions using conjoining operations over check matrices can be efficiently computed---cubically in the number of qubits in the worst case. 
Additionally, and to the point of this paper, the symmetries of the tensors actually map to transversal gates of the code.

Transversal gates and addressable logical operations are best understood as unitary symmetries of the encoding tensor $\mathbf{V}$, or equivalently, that of the state \begin{equation}
    |V\rangle = \sum_{i_j} V_{i_1 i_2,\dots}|i_1,i_2,\dots\rangle.
\end{equation}

\begin{definition}
   Consider a state $|V\rangle$ over $n$ qudits (of possibly different local dimensions) and a unitary $$U= \bigotimes_{I\in \mathcal B} U_I$$ such that $U |V\rangle = |V\rangle$ where $\mathcal B$ is a partition of $\{1,2,\dots, n\}$, then we say that $U$ is a \emph{unitary symmetry} of the tensor $\mathbf{V}$ (relative to $\mathcal{B}$). If $\mathcal B$ is the finest partition, where $I$ runs over single qudits, then we call $U$ is a \emph{unitary product symmetry} \cite{QL}.
\end{definition}

Graphically, a unitary (product) symmetry is shown in Fig.~\ref{fig:qlego}b where $\mathcal{O}_i$ are unitary operators. If $|V\rangle$ is a stabilizer state, then clearly all its stabilizers are unitary symmetries. However, it can also admit non-Pauli symmetries. For instance, the Bell state satisfies $U\otimes U^*|\Phi^+\rangle = |\Phi^+\rangle$ for any unitary $U$ (here $(.)^*$ represents the entry-wise conjugate). Its unitary product symmetry group is $\{U\otimes U^* \::\: U \in \mathcal{U}(\mathbb{C}^2)\}$. Additionally, $\mathrm{SWAP}|\Phi^+\rangle = |\Phi^+\rangle$ and hence gives a unitary symmetry that is not a unitary product symmetry. Such symmetries can also give rise to new unitary product symmetries under trace, e.g. Fig.~\ref{fig:422swapHS}.

As pointed out in \cite{QL}, symmetries can be converted into logical operations by dualizing $|V\rangle$ into an encoding map. In principle, we can define an encoding map by choosing any subset of tensor legs to be logical legs and the remaining to be physical. Yet in practice we want our encoding maps to be isometric, which can be determined by its entanglement properties. 
\begin{lemma}
\label{lemma:isometry}
    Suppose the state representation of a tensor $|V\rangle \in \mathcal{H}_J \otimes \mathcal{H}_{J^c}$ has $\Tr_{J^c}[|V\rangle\langle V|]\propto I$. Then the dual map $V:\mathcal{H}_J \rightarrow \mathcal{H}_{J^c}$ is an isometry.
\end{lemma}
The proof is clear as the state has a full rank and flat singular value spectrum, and hence the corresponding linear map is an isometry. 

\begin{proposition}\label{prop:duallogical}
     Suppose $V:\mathcal{H}_J\rightarrow \mathcal{H}_{J^c}$ is an (encoding) isometry and $U=Q_J\otimes O_{J^c}$ is a unitary symmetry of its Choi state $|V\rangle$, then the code admits $Q^*$ as a logical operator with implementation $\bar{Q}^*=O$.
 \end{proposition}

See \cite{QL} for a proof. As all unitary logical gates are symmetries of the encoding tensor, i.e., the Choi state of an encoding map, it is sufficient to characterize the unitary symmetries of a contracted tensor network or its associated Choi state and then convert it back to an encoding map when convenient. We can treat transversal and addressable gates simply as special cases of above where the unitary symmetry have additional structure. For simplicity, we will focus on cases where bond dimensions are uniform with $\dim \mathcal{H}_J = q^k$ and $\dim \mathcal{H}_{J^c} = q^n$. 

\begin{definition}
    A logical gate $Q^*$ derived from the unitary symmetry $U=Q_{J}\otimes O_{J^c}$ is \emph{transversal} if $O_{J^c} = \bigotimes_{i=1}^{n} O_{i}$. 
If $O_{i}= O_0$ are all identical, then we say $\bar{Q}^*$ is \emph{strongly} (or \emph{bitwise}) transversal, otherwise we refer to it as \emph{weakly} transversal. If $\bar{Q}^*$ is strongly transversal and $Q^*=\bigotimes_{i=1}^k Q_i$ where  $Q_i=O_0$ for all $i$, i.e. if $\bar{O}_0^{\otimes k} = O_0^{\otimes n}$, then we say the gate is \emph{completely} transversal. If a weakly transversal gate has non-trivial support $B\subsetneq \{1,\dots,n\}$, then we say the gate is \emph{localized} transversal. If the logical operation $Q_J$ has non-trivial support $A\subsetneq \{1,\dots,k\}$, then we say the gate is $A$-addressable. 
\end{definition}
For example, in the usual $[[4,2,2]]$ code one has $\overline{CZ}=SS^{\dagger}SS^{\dagger}$ is weakly transversal. A Pauli gate ${\bar{Z}\otimes \bar{I}}=ZZII$ in the same code is also weakly transversal but, more specifically, localized transversal. It is also single-qubit addressable as it only acts on the first encoded qubit. On the other hand, in the $[[8,3,2]]$ cubic code $\overline{CCZ}=T^{\otimes 8}$ is strongly transversal, while the logical Hadamard in the $[[7,1,3]]$ Steane code $\bar{H}=H^{\otimes 7}$ is completely transversal. In addition to these conventional gates, other symmetries like arbitrary phase rotation are also of interest.

\subsection{Multi-qubit and addressable transversal gates}

The general formalism above of rephrasing symmetries as transversal gates has two straightforward but important applications, which to the best of our knowledge have not yet been discussed in literature: the construction of multi-qubit transversal gates and addressable transversal gates. In the examples thus far, we had the implicit assumption that each tensor leg, or each qudit, represents a single logical or physical degree of freedom. However, we can also think of each site or leg as an $\ell$-tuple of physical or logical qudits. The above definition then also applies to multiqubit gates. Two special groupings of such qudits are of particular interest to us, which will give rise to logical gates that are interblock or intrablock transversal (Fig.~\ref{fig:multicopy_tensor}ab).

In the first type of grouping (Fig.~\ref{fig:multicopy_tensor}a), each tensor $\mathbf{V} = \bigotimes_{i=1}^\ell \mathbf{V}_i$, where each $\mathbf{V}_i=\mathbf{V}_0$.  Each leg in the left tensor $\mathbf{V}$ consists of $\ell$ physical/logical qudits. If each $\mathbf{V}_0$ represents a single state or code block, the symmetries of $\mathbf{V}$ now corresponds to multi-qudit transversal gates, such as CNOT for $\ell=2$ or CCZ for $\ell=3$. Such a symmetry maps to the conventional \emph{interblock transversal multiqubit gate} as they act on $\ell$ code blocks if we use $V_0$ to represent the encoding map of a single code block\footnote{More generally, $\textbf{V}_i$ need not all be identical and still admitting the symmetries that corresponds to interblock transversal gates.}.

Explicit examples of such tensors include any $[[n,k]]$ CSS code: if $V$ is the encoding isometry of one block of the code, CNOT is completely transversal operating on two distinct blocks as $CNOT^{\otimes(n+k)}(|V\rangle\otimes|V\rangle) = |V\rangle\otimes |V\rangle$.
A more concrete example is the encoding tensor of the $[[12,1,2]]$ code of \cite{Vasmer_2019} with completely transversal CCZ gate: the Choi states of their encoding isometries have global symmetry $CCZ^{\otimes 13}|V_{[[12,1,3]]^{\otimes 3}}\rangle=|V_{[[12,1,3]]^{\otimes 3}}\rangle$.

The second type of grouping is shown in Fig.~\ref{fig:multicopy_tensor}b, where each leg of $\mathbf{V}$ corresponds to multiple qubits/qudits in the same code block. If we were to promote a leg to logical, then such a symmetry is nothing but a multi-qubit ``transversal gate'' that we shall refer to as an \emph{intrablock transversal gate}. Examples of such gates can be found in \cite{Quintavalle_2023}. Later, we will build up some more non-trivial examples of this using quantum lego. A simple example is given by the encoding tensor of the Steane code, which has eight legs each one representing a qubit. However, it is easy to pair up the legs into four pairs where $CZ^{\otimes 4}$ becomes a unitary symmetry of the tensor $\mathbf{V}$ with bond dimension 4. Explicitly, taking just the purple leg of the tensor in Fig.~\ref{fig:multicopy_tensor}d as logical input recovers the usual $[[7,1,3]]$ code. Taking both purple and green legs as inputs, then we arrive at a $[[6,2,2]]$ code where its checks are weight-4 $X$- or $Z$- operators acting on the blue or red colored plaquettes. Then $\overline{CZ}=CZ^{\otimes 3}$ is an intrablock transversal gate thanks to its unitary symmetry.

For $\ell>2$, there are logical gates that fit in between interblock and intrablock, where a gate can act $k_i$ logical qubits of the $i$th code block (hence $\sum_{i=1}^m k_i=\ell$ where $m<\ell$ is the number of distinct code blocks). Here we call such gates \emph{mesoblock gates}. For example, for codes with $k>1$, a mesoblock logical CCZ gate can act on two logical qubits in the first block and one from the second block. Codes with such transversal gates have also been identified in \cite{he2025quantumcodesaddressabletransversal}. Transversal mesoblock gates come from more complicated symmetries formed by a combination of both splits in Fig.~\ref{fig:multicopy_tensor}a and b. We will discuss them in the examples in Sec.~\ref{subsec:genholo}. 

In addition to symmetries that are global---acting on every qudit leg---we also have \textit{localized symmetries} where the unitary product symmetry only have support over a subregion. When such a symmetry only has support over a subset of logical legs, it corresponds to an \textit{addressable transversal gates} associate to that support. 

We have seen such examples from Pauli symmetries of the encoding tensor of the $[[4,2,2]]$ code \cite{QL} where it is easy to identify operators that only act on three of the six legs of the encoding tensor. When it only acts on two physical legs and one of the two logical legs, then it becomes an addressable Pauli gate. Such localized Pauli symmetries are commonplace for stabilizer codes. 

Localized symmetries can also be identified for multi-qubit gates, even though they are far less understood among our current tabulation of codes and logical gates. Consider the encoding tensor(s) $\mathbf{V}_{[[15,1,3]]}$ of the $[[15,1,3]]$ (punctured) Reed-Muller code with localized transversal logical $CZ$ gate. It is not hard to check that there is a configuration of $CZ^{\otimes 7}$ that act on two code blocks which produces a logical $CZ$ which is locally transversal. For example, $CZ$s acting on two of the faces of the tetrahedron implements such a gate (Fig.~\ref{fig:multicopy_tensor}c).This means the tensor $\mathbf{V}=\mathbf{V}_{[[15,1,3]]}\otimes \mathbf{V}_{[[15,1,3]]}$ with $16$ legs and bond dimension $4$ admits a unitary symmetry $CZ^{\otimes 8}\otimes I$. In fact, one can rewrite the 16-leg encoding tensor as a Clifford-stabilized state where each leg is mapped to a vertex on a hypercube, then $CZ^{\otimes 8}$ acting on the vertices each 3 dimensional subcube is a unitary symmetry of $\mathbf{V}$ (see Appendix~\ref{app:cz}).

\begin{figure}
    \centering
    \includegraphics[width=0.9\linewidth]{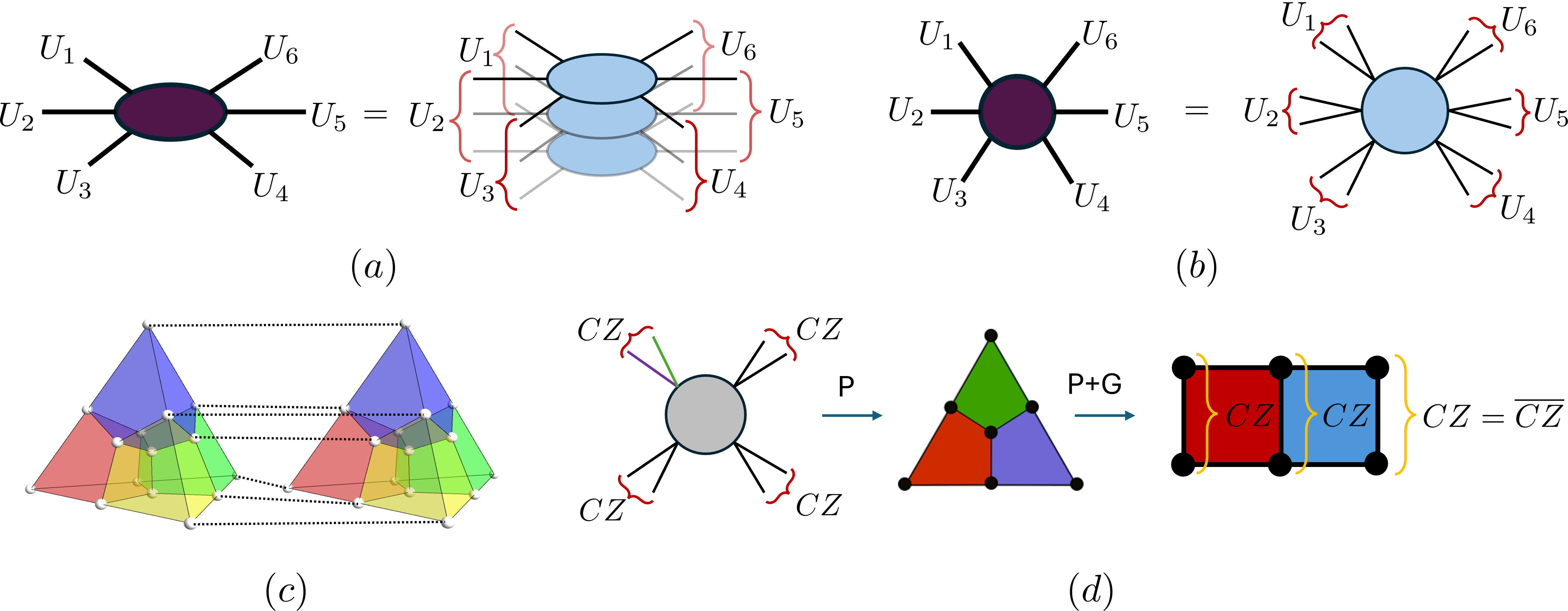}
    \caption{(a,b) Unitary symmetry of a tensor $\mathbf{V}$ (purple) can manifest as different types of multi-qubit FT gates based on the different groupings of the tensor legs. (c) $CZ$s acting on pairs of qubits connected by the dotted lines performs a locally transversal logical $CZ$ gate, but the $\mathbf{V}$ encoding tensor also admits unitary symmetry when $CZ$s act on all vertices of any colored cubes. (d) As an explicit example of (b), a unitary symmetry of a single Steane encoding tensor is $CZ^{\otimes 4}$. Taking just the purple (P) leg as logical produces the Steane code whereas taking both purple and green (P+G) produces a $[[6,2,2]]$ code with intrablock transversal $CZ$. The green leg can be chosen to be the physical qubit on any of the 3 vertices of the triangle.}
    \label{fig:multicopy_tensor}
\end{figure}

Instead take two legs from each $\mathbf{V}_{[[15,1,3]]}$ as logical legs, for which we get a $[[14,2,2]]$ code. This can be written as a bilayer code where each layer consists of 7 physical qubits\footnote{One choice for the second logical leg would be the physical qubit in the middle of the tetrahedron, though choosing any of the vertices will work also.} (Fig.~\ref{fig:1422_CZ}a). It is similar to two layers of the Steane code except its stabilizers consists of weight-8 $X$- and $Z$-operators acting on the vertices of two plaquettes with the same color. In addition, we have weight-4 $Z$-stabilizers acting on each face that is perpendicular to the layers or on the 4 vertices of any plaquette on a single layer. Note that the $Z$-stabilizers we give here are merely to provide geometric intuition; they are overcomplete when used as a generating set. This code has logcal $\bar{X}_i, \bar{Z}_i$ operators ($i=1,2$). We have $\bar{X}_i$ consists of a weight-7 $X$-operator supported entirely on the $i$th layer, and similarly for $\bar{Z}_i$. The localized $CZ$ symmetry now represents addressable interblock gates. Indeed, taking again two blocks of this code whose tensor is given by $\mathbf{V}$, an addressable logical $\overline{CZ}_{ij}$ acts the $i$th and $j$th logical qubit of first and second code block respectively. Then $\overline{CZ}_{ij}$ is implemented by having 7 parallel $CZ$s acting on the $i$th layer of the first code block and $j$th layer of the second code block. The localized symmetry now corresponds to addressable interblock transversal $CZ$ gates (Fig.~\ref{fig:1422_CZ}b).
\begin{figure}
    \centering
    \includegraphics[width=0.5\linewidth]{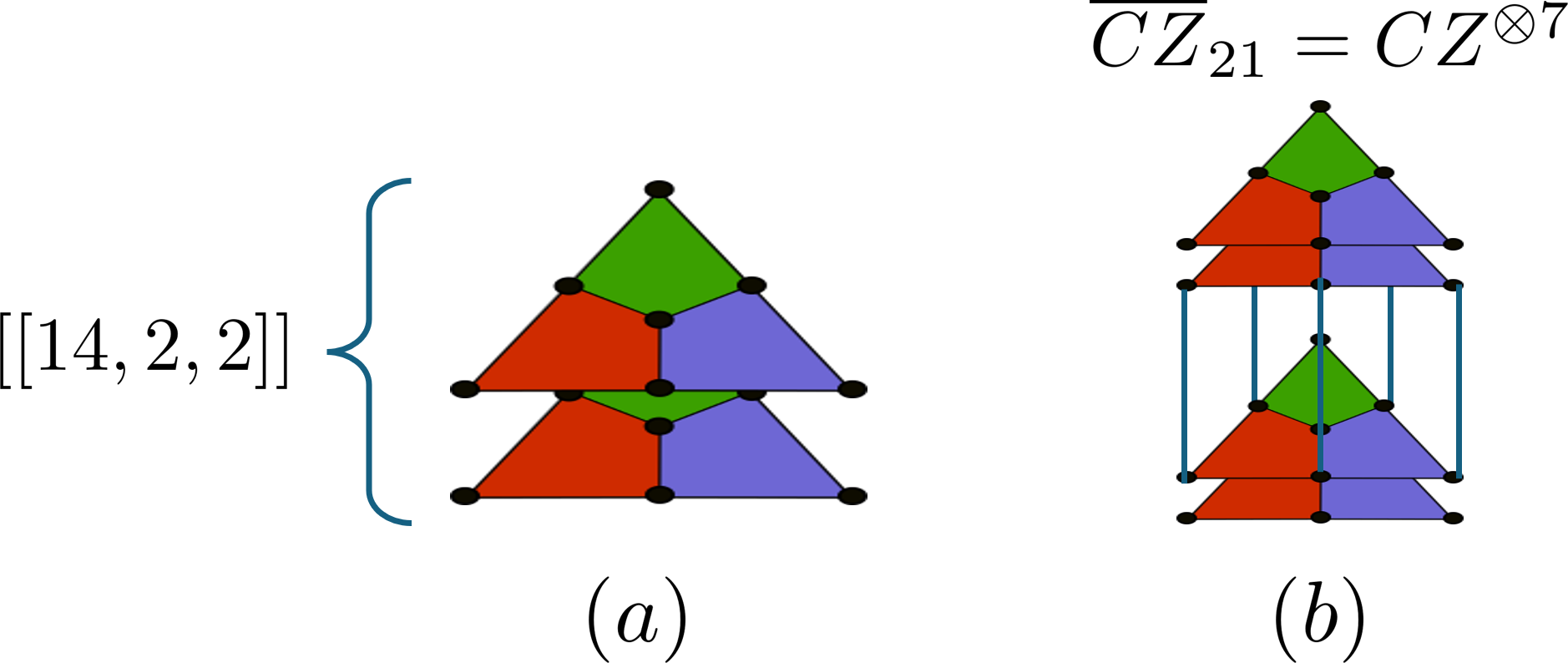}
    \caption{$[[14,2,2]]$ codes obtained from the encoding tensor of the punctured $[[15,1,3]]$ QRM encoding tensor admit addressable $CZ$ gates. Weight 8 X and Z checks act on the vertices of the plaquettes with the same color. Additional Z checks connect the two layers, which are not shown in the figure.}
    \label{fig:1422_CZ}
\end{figure}

Although we have been mostly focusing on smaller codes or tensors in the examples, the above observations of transversal gates still apply if $\mathbf{V}$ is a high rank tensor network.

\subsection{Addressable gates from global transversal ones}

Before turning to tensor network methods, we explore some simple ramifications of the existence of addressable Pauli symmetries (which every stabilizer code exhibits in its logical Pauli operators). Cleaning is a key concept, which we state here for general unitary operators as we will need this later.

\begin{definition}[Cleanability]
    Suppose $U$ acts on a single leg of a tensor $\ket{V}$ and $U^N=I$ for some integer $N$. We say a set of legs $A$ of $\ket{V}$ is \emph{cleanable with respect to $U$} if for any operator of the form $O(\mathbf{r})=\bigotimes_{i\in A} U^{r_i}$ (with $r_i\in \mathbb{Z}_N$) there exists a unitary symmetry $S$ of $\ket{V}$ such that  $SO(\mathbf{r})$ has only trivial support on $A$. 
    A code is $r$-cleanable with respect to $U$ if its encoding tensor is cleanable for all sets of legs $A$ that contains of all the logical legs and $r$ additional physical legs.
\end{definition}
In other words, a code is $r$-cleanable if any combinations of $U$ and $I$ on any subregion $B$ can be pushed to $B^c$ using a unitary stabilizer. The unitary stabilizer here need not be a Pauli operator. If $U$ is a multiqubit gate, this implies the interblock gates can be cleaned off of the thickened legs as in Fig.~\ref{fig:multicopy_tensor}a. 

For example, consider an $n$-qubit GHZ state $|\psi\rangle = \frac{1}{\sqrt{2}}(|00\dots\rangle +|11\dots\rangle)$. Let $C^\ell P(\varphi)$ be a multi-controlled phase gate with angle $\varphi$ where $\ell\geq 0$; and $\ell=0$ corresponds to a single qubit phase gate 
\begin{equation}
    P(\varphi) =\begin{pmatrix}
    1 & 0\\
    0 & e^{i\varphi}
\end{pmatrix}.
\end{equation}
We can check that $C^\ell P(\varphi)$  is 1-cleanable on $\ell$ GHZ states (or $k=0$ codes). Denote the gate that act on the $i$th $\ell$ tuple of qubits on $|\psi\rangle^{\otimes \ell}$ as $C^\ell_i P(\varphi)$. Because $S=C^{\ell}_iP(\varphi)C^{\ell}_jP(-\varphi)$ is a unitary symmetry  for any $i, j$, we have
\begin{equation}
    C^{\ell}_iP(\varphi)|\psi\rangle^{\otimes \ell} = C^{\ell}_jP(\varphi)|\psi\rangle^{\otimes \ell}. 
\end{equation} 
This is also shown diagrammatically in Fig.~\ref{fig:zxhcalc}.

The logical legs of a stabilizer code are cleanable with respect to any Pauli operator. This is a consequence of the existence of logical Pauli operators on stabilizer codes. Namely, suppose we are given a Pauli operator on a logical leg; associated to that is a logical Pauli operator whose representation matches the given Pauli operator on that logical leg, is trivial on other logical legs, and has some additional Pauli operators on the physical legs. This operator is a symmetry of the code's Choi state, which performs the desired cleaning.

In our language, the above recognition states that every stabilizer code admits transversal logical Pauli gates, and every logical Pauli operator is addressable. Here we catalog some simple results showing that given a (global) transversal $C^\ell P$ gate and addressable Pauli $X$-gates, we can construct addressable (controlled-)phase gates, albeit courser than the global transversal gate.

\begin{lemma}
\label{lemma:phasegate_lemma}
    We have that $XP(\varphi) = P(-\varphi)X$, and for $\ell \geq 1$ that 
    $$(I^{\otimes \ell} \otimes X)C^\ell P(\varphi) = C^{\ell}P(-\varphi)(C^{\ell-1}P(\varphi) \otimes X).$$
\end{lemma}
The proof of this lemma is straightforward and so left to the reader. 
Note that the multiqubit control phase $C^\ell P$ is symmetric in all its qubits, and so while we stated this result for ease of presentation one has analogous equalities where the $X$ operator can appear on any leg.

\begin{proposition}
    Let $\ket{V} \in \mathcal{H}_J \otimes \mathcal{H}_{J^c}$ be a tensor with a (global) phase symmetry
    $$\mathbf{P}\ket{V} = [P(-\varphi_1)\otimes\cdots\otimes P(-\varphi_k)\otimes P(\varphi_{k+1}) \otimes \cdots \otimes  P(\varphi_{k+n})] \ket{V} \propto \ket{V}.$$
    If $\ket{V}$ admits an addressable logical $X$-gate on leg $i\in J$, then it admits an $\{i\}$-addressable logical $P(2\varphi_i)$-gate whose support is contained in that of the logical $X$-gate. Moreover, if this logical $X$ gate is realized by a Pauli symmetry then this logical $P(2\varphi)$ gate is realized by a phase symmetry.
\end{proposition}
\begin{proof}
    Let our supposed logical $X$-gate be given by the symmetry
    \begin{equation}\label{eqn:logical_X}
        \bar{X}_i = I\otimes \cdots \otimes X \otimes \cdots \otimes I \otimes O_{k+1} \otimes \cdots \otimes O_{k+n}.
    \end{equation}
    Then $\ket{V}$ also has the symmetry
    \begin{align*}
        \bar{X}_i \mathbf{P} \bar{X}_i \mathbf{P}^\dagger &=  I\otimes \cdots \otimes P(-2\varphi_i) \otimes \cdots \otimes I\\
        &\qquad \otimes O_{k+1}P(\varphi_{k+1})O_{k+1}^\dagger P(-\varphi_{k+1}) \otimes \cdots \otimes O_{k+n}P(\varphi_{k+n})O_{k+n}^\dagger P(-\varphi_{k+n})
    \end{align*}
    which defines our addressable logical $P(\varphi_i)$-gate. Note that if leg $k+j \not \in \mathrm{supp}(\bar{X}_i)$ then $O_{k+j} = I$, and hence $k+j \not \in \mathrm{supp}(\bar{X}_i \mathbf{P} \bar{X}_i \mathbf{P}^\dagger)$ either. Finally, if $O_{k+j}$ is Pauli, then on this leg the resulting symmetry has $I$ (when $O_{k+j} = I,Z$) or $P(2\varphi_{k+j})$ (when $O_{k+j} = X,Y$).
\end{proof}

\begin{corollary}
    If a stabilizer code has a bitwise transversal $T$-gate, then it has transversal $S$-gates addressable on each logical qubit.
\end{corollary}

\begin{proposition}
    Let $\ket{V} \in \mathcal{H}_J \otimes \mathcal{H}_{J^c}$ be a tensor with (interblock) phase symmetry
    $$\mathbf{P} = CP(-\varphi_1)\otimes\cdots\otimes CP(-\varphi_k)\otimes CP(\varphi_{k+1}) \otimes \cdots \otimes  CP(\varphi_{k+n}).$$
    If $\ket{V}$ admits an addressable logical $X$-gate on leg $i\in J$ realized by a Pauli symmetry, then it admits (i) an $\{i\}$-addressable logical $P(\varphi_i)$-gate, and (ii) an $\{i\}$-addressable logical $CP(2\varphi_i)$-gate. In both cases, these are implemented with phase gates and the support of each is contained in that of the logical $X$-gate.
\end{proposition}
\begin{proof}
    As we did earlier, we view $\ket{V}\otimes\ket{V}$ as a $k+n$ qudit system with legs of bond dimension $4$ and $\mathbf{P}$ acting transversally as the tensor product of single-site operators. Let $\bar{X}_i$ be the supposed logical $X$-gate (\ref{eqn:logical_X}). Acting with $\bar{X}_i$ on the first copy of $\ket{V}$, we work analogously to the previous lemma and compute
    \begin{equation}\label{eqn:commutator1}
    (\bar{X}_i\otimes I) \mathbf{P} (\bar{X}_i\otimes I) \mathbf{P}^\dagger = I^{\otimes 2}\otimes \cdots \otimes (I\otimes P(\varphi_i))CP(-2\varphi_i) \otimes \cdots \otimes I^{\otimes 2} \otimes \mathbf{O}_{J^c},
    \end{equation}
    where the $k+j$ tensor factor of $\mathbf{O}_{J^c}$ is either $I\otimes I$ (when $O_{k+j} = I,Z$) or $(I\otimes P(\varphi_{k+j}))CP(-2\varphi_{k+j})$ (when $O_{k+j} = X,Y$).

    Now an identical computation that led to (\ref{eqn:commutator1}) gives
    $$(I\otimes \bar{X}_i) \mathbf{P} (I\otimes \bar{X}_i) \mathbf{P}^\dagger = I^{\otimes 2}\otimes \cdots \otimes (P(\varphi_i)\otimes I)CP(-2\varphi_i) \otimes \cdots \otimes I^{\otimes 2} \otimes \mathbf{O}'_{J^c},$$
    where the $k+j$ tensor factor of $\mathbf{O}'_{J^c}$ is either $I\otimes I$ (when $O_{k+j} = I,Z$) or $(P(\varphi_{k+j})\otimes I)CP(-2\varphi_{k+j})$ (when $O_{k+j} = X,Y$). Inverting this latter symmetry and composing with (\ref{eqn:commutator1}) produces the symmetry
    $$I^{\otimes 2}\otimes \cdots \otimes (P(-\varphi_i)\otimes P(\varphi_i))\otimes \cdots \otimes I^{\otimes 2} \otimes \mathbf{O}_{J^c}\mathbf{O}'_{J^c}.$$
    But the $k+j$ tensor factor of $\mathbf{O}_{J^c}$ is either $I\otimes I$ or $P(-\varphi_{k+j})\otimes P(\varphi_{k+j})$, and hence this operator is separable and defines a symmetry on each copy of $\ket{V}$. This precisely defines our $\{i\}$-addressable logical $P(\varphi_i)$-gate.

    By construction, this ${i}$-addressable $P(\varphi_i)$-gate consists precisely the single-site phase terms in (\ref{eqn:commutator1}), hence composing this equation with the inverse of the $P(\varphi_i)$-gate produces our $\{i\}$-addressable $CP(2\varphi_i)$-gate as desired.
\end{proof}

Note that this the lemma, the $\{i\}$-addressability of the $CP(2\varphi_i)$-gate refers to the thickened $i$th (logical) leg, and hence is $(i,i)$-addressable when viewed as an interblock gate. As we have only supposed a bitwise symmetry of our tensor, we should net expect to find $(i,j)$-addressable gates in general when $i\not= j$.

\begin{corollary}
    If a stabilizer code has a bitwise transversal $CS$-gate, then it has an addressable transversal $S$-gates on each logical qubit, and an addressable transversal $CZ$-gates between matching logical qubit on two blocks.
\end{corollary}
One can generalize this corollary to more controls and finer angles. 

While in the above results we have focused on logical gates, the exact same computations produce non-Pauli symmetries from stabilizers (with $X$-content) on a stabilizer code. We have seen above (Fig. \ref{fig:multicopy_tensor}c) a 3D color code that has a localized transversal $CZ$-gate as well as localized transversal $CZ$ symmetries. These are not special to the code per se but merely consequences of the fact it supports a bitwise transversal $CS$ gate.

\section{Generalized trace and gate propagation}
\label{sec:3}

\subsection{Transversal gates under Bell fusion}
As indicated by \cite{QL}, it is possible to learn about the symmetries of a tensor network created by gluing together smaller components through operator matching when the tensor legs are fused using Bell states. For completeness, we briefly review operator matching, tensor gluing, and some warm-up examples with concatenation. Familiar readers may skip ahead.

\begin{definition}[Tensor trace]
    A tensor trace of legs $j$ and $m$ refers to the operation where one sums over two of the tensor indices $a_j, a_m$ that correspond to those legs. Namely the trace is $V_{a_1,a_2,\dots}' = \sum_{a_j,a_m} V_{a_1,a_2,\dots,a_j,\dots, a_m,\dots}\delta_{a_ja_m}$. If $|V\rangle$ is the state representation of the tensor then this operation is equivalent to the Bell fusion $\langle \Phi^+|V\rangle$ where $\langle\Phi^+|=\langle00|+\langle 11|$ acts on qubits $j$ and $m$. In the case of stabilizer states (or Choi states of stabilizer encoding maps), these can be represented via their stabilizer tableau (or check matrix) $H_V$; the counterpart of this trace operation is called conjoining and is denoted as $\wedge_{j,m}H_V$ \cite{QL}.
\end{definition}

We will abuse terminology and use tensor trace, contraction, fusion, and conjoining interchangeably throughout the text as they all refer to the same operation.

Ref. \cite{QL} showed that symmetries on the dangling edges of the traced tensor network are preserved under Bell fusion if the operators $O_i, Q_j$ on the traced legs between traced tensors are compatible, i.e., $O_i=Q^*_j$. This implies that unitary transversal gates of any kind on the respective lego blocks will be preserved by the bigger quantum lego code if these compatibility conditions are satisfied. All operations with Bell fusions can be written as a self-trace operation where we contract two free indices in the same tensor network. When we glue two distinct blocks together, it can be understood as a self-trace on a network that has disconnected components. Consider a Bell fusion between any two sites $j,m$ of the same local dimension $q$, then the following holds.

\begin{lemma}[Matching Lemma\cite{QL,QL2}]
\label{lemma:operatormatching}
Suppose $U=U_{j}\otimes U_{m}\otimes U_{\{1,\dots,n\}\setminus\{j,m\}}$ is a unitary symmetry of a tensor (network) $V_{a_1a_2\dots a_j\dots a_m\dots}$. If $U_j=U_m^*$ where as before $(\cdot)^*$ denotes entrywise complex conjugation, then $U_{\{1,\dots,n\}\setminus\{j,m\}}$ is a unitary symmetry of the contracted tensor (network) with legs $j$ and $m$ contracted/fused.    
\end{lemma}

 This means that if we have a 1-qubit traversal gate on any lego, then if we trace two legs in a compatible way, then the gate is also transversal on the post-traced code.
Graphically, an example of such operator matching is shown in Figure~\ref{fig:qlego}c where $\bigotimes_i\mathcal O_i, \bigotimes_j\mathcal Q_j$ are local unitary symmetries of the blue and red tensors respectively and two legs are contracted. If $\mathcal O_3=\mathcal Q_3^*$ and $\mathcal O_2=\mathcal Q_2^*$, then $\mathcal O_4\otimes \mathcal O_5\otimes \mathcal O_1\otimes \mathcal Q_1\otimes \mathcal Q_4$ is a unitary symmetry of the tensor network. Operationally, Lemma \ref{lemma:operatormatching} gives us a precise way to propagate the unitary symmetry, and hence transversal gates, of the codes formed from legos. A graphically intuitive way to visualize this propagation is by operator pushing or operator flow \cite{Pastawski_2015,QL} in the tensor network. 

Note that we have placed no requirements on the operators making up a symmetries beyond being unitary! Therefore, the propagation of symmetries is much more general than the transversal Pauli gates in tensor network codes introduced in \cite{tnc}, the Pauli web in ZX-based analysis \cite{Bombin_2024,delaFuente:2024fyf}, and the Clifford gates one constructs from the Pauli stabilizer formalism. Importantly, none of the above statements relies on the atomic legos being Pauli stabilizer codes. As long as their transversal gates are known, these symmetries will carry over under lego gluing\cite{Shen:2023xmh}.

\begin{figure}
    \centering
    \includegraphics[width=0.8\linewidth]{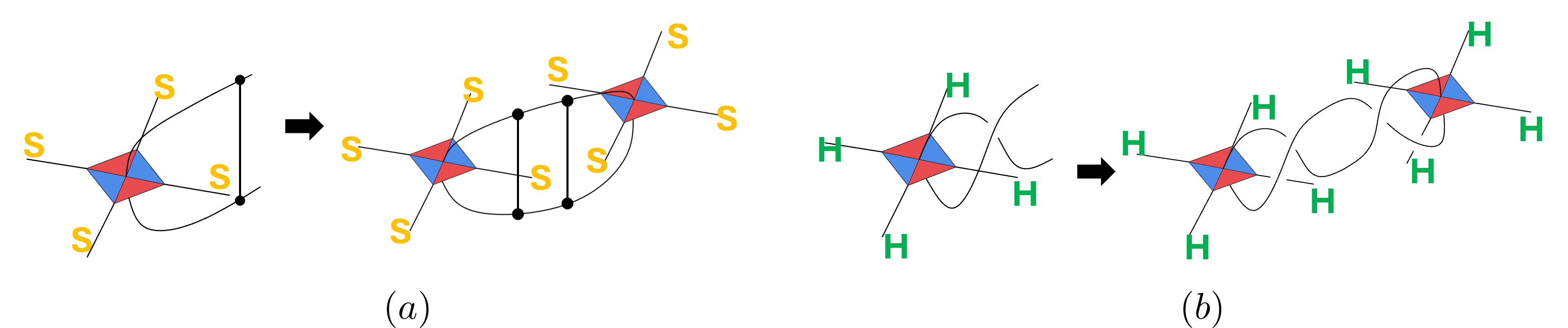}
    \caption{Transversal $S$ and $H$ gates emerge from unitary symmetries that are not unitary product such as $CZ$ and $SWAP$.}
    \label{fig:422swapHS}
\end{figure}

An obvious application of this is in the context of (generalized) code concatenation such as \cite{tnc} and holographic codes\cite{Pastawski_2015} where transversality follows trivially by gluing together suitable lego blocks with compatible symmetries. It can be especially helpful for understanding operator pushing of multi-qubit gates. Te be precise we define (partial) concatenation of codes of uniform local dimension as follows, but this could be generalized to nonuniform local dimension.

\begin{definition}[Partial Concatenation]
\label{def:partialconcat}
    Consider an $[[n,k]]_q$ code $C_1$ and an $[[n',k']]_q$ code $C_2$. A \emph{partial concatenation} is when a list of distinct physical qudits of $C_1$ indexed by $B\subset \{1, \dots, n\}$ is conjoined entry-by-entry to a list of distinct logical qudits of $C_2$ indexed by $B' \subset \{1,\dots, k'\}$  (where of course $|B| = |B'|$).
\end{definition}

It is clear that the special case $n = k'$ and $B = \{1,\dots,n\}$ reduces to the conventional code concatenation, and in particular the traditional case is when $C_2 = C^{\otimes k'}$ is $k'$ copies of a $[[n'',1]]$ code $C$. A key difference in the more general case is that logical degrees of freedom can be added during a partial concatenation, e.g. holographic codes. 

\begin{proposition}
\label{prop:isometrictrace_sym}
    Let $V_1, V_2$ be two tensors where their Choi states over $k+n$ and $k'+n'$ legs have symmetries 
    $$(L\otimes O_B\otimes Q_{\{1,\dots,n\}\setminus B})|V_1\rangle =|V_1\rangle \text{ and } (O^*_{B'}\otimes L'_{\{1,\dots k'\}\setminus B'}\otimes R)|V_2\rangle=|V_2\rangle,$$
    where as above $|B| = |B'|$. Provided that the subsystem on qubits $\{1,\dots,k\} \uplus (\{1,\dots,k'\}\setminus B')$ (the disjoint union of the untraced logical degrees of freedom) is maximally mixed, then $|V\rangle=\wedge_{B,B'}(|V_1\rangle\otimes |V_2\rangle)$ is the encoding tensor of a $[[n+n'-|B|,k+k'-|B|]]$ code where $\bar{L}^*\otimes\bar{L}'^*=Q\otimes R$. 
\end{proposition}

In this statement we have abused notation slightly: if $B = (b_1,\dots, b_r)$ and $B' = (b'_1,\dots, b'_r)$ then the tracing operator $\wedge_{B,B'} = \wedge_{b_1,b'_1} \cdots \wedge_{b_r,b'_r}$. The proof is left as an exercise for the reader as it follows simply from Lemma~\ref{lemma:operatormatching} and Prop.~\ref{prop:duallogical}. Note that if $|V_1\rangle, |V_2\rangle$ are both isometries when mapping from $k$ to $n$ qudits and $k'$ to $n'$ qudits respectively, then this is nothing but a partial concatenation. When tracing together tensors where the $m$ legs being traced on a tensor can be dualized to an isometry that uses said legs as input, then the trace is also said to be \emph{$m$-isometric}.

While it is clear how logical gates propagate under concatenation, here we emphasize that tensor contractions are more flexible as no isometric condition needs to apply. Such a generalization is needed for sparsification in \cite{Cao:2025oep}.

\begin{corollary}\label{coro:completetransversal}
    The arbitrary Bell fusion of seed codes with completely transversal gates $Q$ where $Q=Q^*$ is an $[[n,k]]_q$ stabilizer code with strongly transversal $Q^{\otimes n}$. If the $k$ qubit subsystem is also maximally mixed, then the same gate is completely transversal and implements $Q^{\otimes k}$ logically.
\end{corollary}
For example, the arbitrary contraction of encoding tensors of any CSS codes will produce a Clifford-stabilized state that supports completely transversal CNOT gates. This is of course consistent with the conclusion from \cite{QL} where it is known that the contraction of CSS tensors remain a CSS code. However, additional care needs to be taken if the subsystem chosen to be the logical legs is not maximally mixed. Such encoding maps are not isometries but can be made into an isometry by quotienting out the kernel of the map. While $Q^{\otimes n}$ remains a valid logical operator, it need not implement $\bar{Q}^{\otimes k}$. This is why the code need not be completely transversal in general. 

\subsection{Generalized trace and multi-qubit transversal gates}
Although all tensor contractions can be thought of as a Bell fusion between two of the existing legs by projecting these qubits in the Choi state onto the $|\Phi^+\rangle=\frac{1}{\sqrt{2}}(|00\rangle+|11\rangle)$ state, it is often convenient to also consider generalized traces in our context where states other than $|\Phi^+\rangle$ is used.

\subsubsection{Clifford-deformed edge states}
A slight generalization is to insert a single-site Pauli or Clifford operator into the traced edge. In this work, we restrict our attention to stabilizer codes and will not consider non-Clifford deformations of such edge states.
\begin{definition}[Local Clifford deformed traces]
    A \emph{Clifford-deformed fusion} is the projection of the fusion sites to $|\Phi_C\rangle\equiv I\otimes C|\Phi^+\rangle$, where $C$ is a single-site Clifford operator. 
\end{definition}

A special subclass of these traces are ones with single qubit sites and $C$ being a Pauli operator. When tracing together stabilizer codes, deforming the edges with Pauli operators have very little change to the overall structure of the stabilizer group, but they alter the signs of the generators. We can easily see this through operator pushing: the result is carried through each trace identically except for Paulis that anti-commute with the inserted operator, which then switch signs. Although such changes are inconsequential when pushing Pauli operators, they are often important for pushing other unitary gates, such as in identifying codes with strongly transversal gates \cite{Rall} for magic state distillation (MSD)\footnote{While strong transversality is not strictly necessary in MSD as one can simply make local unitary deformations, a uniform description does significantly simplify the analysis.}.

For example, to preserve strong transversality of some operator $O$, one has to glue together codes where one supports symmetry that involves $O$ and the other involves $O^*$ to satisfy the operator matching condition if we want to preserve some kind of transversality involving only physical $O$ gates while growing the codes. Generally, $O, O^*$ represent distinct operators that differ by more than just an overall phase. To overcome this issue, it is also convenient to consider local unitary deformation $U$ on one of the lego blocks to be traced such that $UOU^{\dagger}=\omega O^*$ for some phase factor $\omega$. For many commonly used gates, $U$ can often be identified as Pauli matrices. For example, to create a bigger code by tracing two $[[5,1,3]]$ codes (which support transversal $SH$ gates), one applies a local $U=Y$ deformation on one of the code blocks so that the implementation of the transversal $SH$ on this code becomes $\overline{SH}=(SH)^*\otimes (SH)^{\otimes 4}$. The tracing its first qubit with any leg of the other $[[5,1,3]]$ code produces an $[[8,2,3]]$ code that supports bitwise transversal $SH$ gate. This is of course equivalent to tracing together two undeformed $[[5,1,3]]$ codes with a $Y-$deformed edge $|\Phi_Y\rangle$ where $(SH)\otimes(SH)|\Phi_Y\rangle =\omega_Y|\Phi_Y\rangle$ (Fig.~\ref{fig:PauliDef}b). For a different example, $T\otimes T|\Phi_X\rangle =\omega_X|\Phi_X\rangle$ and hence $T$ gates are directly matched to $T$ gates over $X$-deformed edges, which are used in \cite{Shen:2023xmh} to produce codes with completely transversal $T$ gates. In fact, this trick works for all phase gates as seen in the following simple result.

\begin{lemma}\label{prop:3.2}
    Let $$P(\varphi)=\begin{pmatrix}
        1 & 0 \\
        0 & e^{i\varphi}
    \end{pmatrix}$$
    be any phase gate on a single qubit, then $P(\varphi)\otimes P(\varphi)|\Phi_X\rangle = e^{i\varphi}|\Phi_X\rangle$.
\end{lemma}

A useful application of the matching lemma and the deformed trace is to produce codes with transversal $P(\varphi)$ gates, like $T$.

\begin{corollary}
\label{coro:isotrace_Tcode}
Consider codes $[[n_1,k_1]]$, $[[n_2, k_2]]$  where $\bar{G}_1=P(\varphi)^{\otimes n_1}$ and $\bar{G}_2=P(\varphi)^{\otimes n_2}$. Then an $m$-isometric trace using deformed edge $|\Phi_C\rangle$ produces a $[[n_1+n_2-2\ell, k_1+k_2]]$ code that admits logical action $\bar{G}_1\otimes \bar{G}_2=P(\varphi)^{\otimes n_1+n_2-2m}$.
\end{corollary}
\begin{proof}
    Under deformed traces with $|\Phi_X\rangle$, $P(\varphi)\otimes P(\varphi)$ will stabilizer $|\Phi_X\rangle$ up to global phase. Hence the matching condition will be satisfied. Combined with Proposition~\ref{prop:isometrictrace_sym}, we know that $G_1^*\otimes G_2^*\otimes P(\varphi)^{n_1+n_2-2\ell}$ must be a global symmetry of the Choi state. 
\end{proof}

Note that the above lemma is a special case where all legs being traced have $P(\varphi)$s acting on them. However, for the more general codes that may have any power of $P(\varphi)$s acting on them, one can choose the undeformed Bell trace and deformed Bell traces based on whether the gate has support on the traced leg. To preserve the overall logical action $G_1, G_2$, we should perform an undeformed trace for places where no phase gates has support while reserving $|\Phi_X\rangle$ on legs where they do. An example of this to produce new classes of CSS-T codes is discussed in Sec.~\ref{sec:4}, but the Corollary also applies to non-CSS codes, and indeed, non-stabilizer codes.

\begin{figure}
    \centering
    \includegraphics[width=0.25\linewidth]{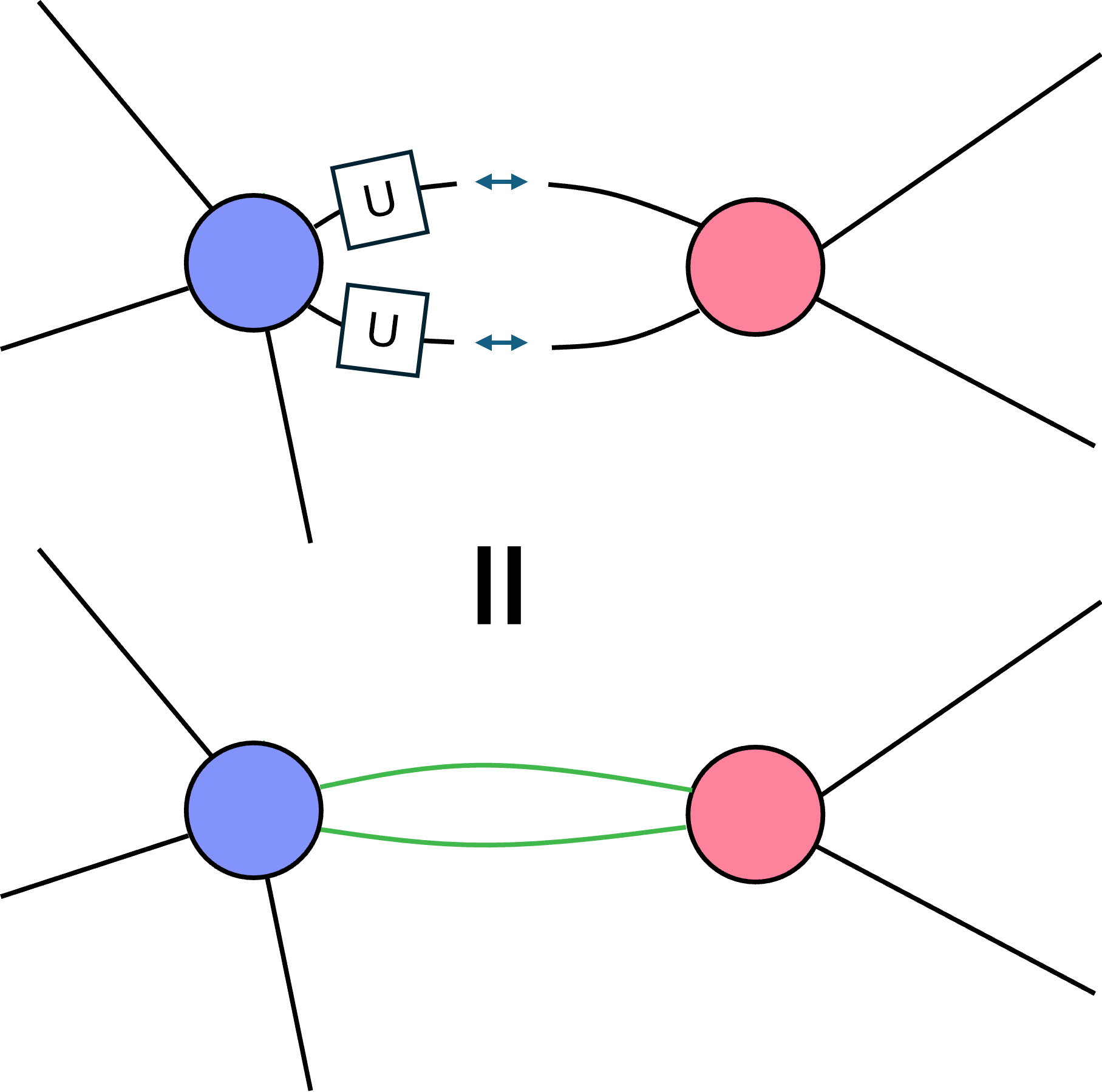}
    \caption{Local deformation followed by gluing is equivalent to projecting the relevant sites to unitary-deformed Bell states $U\otimes I|\Phi^+\rangle$ represented by the green edges.}
    \label{fig:LUdef}
\end{figure}

We can also utilize the deformed trace for multi-qubit gates. Again, consider the example of tracing together two $[[5,1,3]]$ codes. Each perfect code also supports an interblock transversal 3-qubit Clifford gate \cite{Gottesman_1998}
\begin{equation}
   K_3=\frac 1 2\begin{pmatrix}
        1& 0& i &0 &i &0& 1& 0\\
    0 &-1 &0 &i &0 &i &0 &-1\\
    0 &i &0 &1 &0 &-1 &0 &-i\\
    i &0 &-1 &0 &1 &0 &-i &0\\
    0 &i &0 &-1 &0 &1 &0 &-i\\
    i &0 &1 &0 &-1& 0 &-i &0\\
    -1 &0 &i &0 &i &0 &-1 &0\\
    0 &1 &0 &i &0 &i &0 &1\\
    \end{pmatrix}.
\end{equation}
However, $K_3^*$ is not a power of $K_3$. Therefore, straightforward Bell fusion does not preserve this gate and local Clifford deformations are needed to ensure the matching condition. This can be solved with multi-qubit $Y$-deformed Bell fusion
~(Fig.~\ref{fig:PauliDef}a). 
Incidentally, the $Y$-deformation also preserves the bitwise transversality of $SH$ gates from above. Hence codes produced with such deformed edges support both completely transversal gates.  
\begin{figure}
    \centering
    \includegraphics[width=0.4\linewidth]{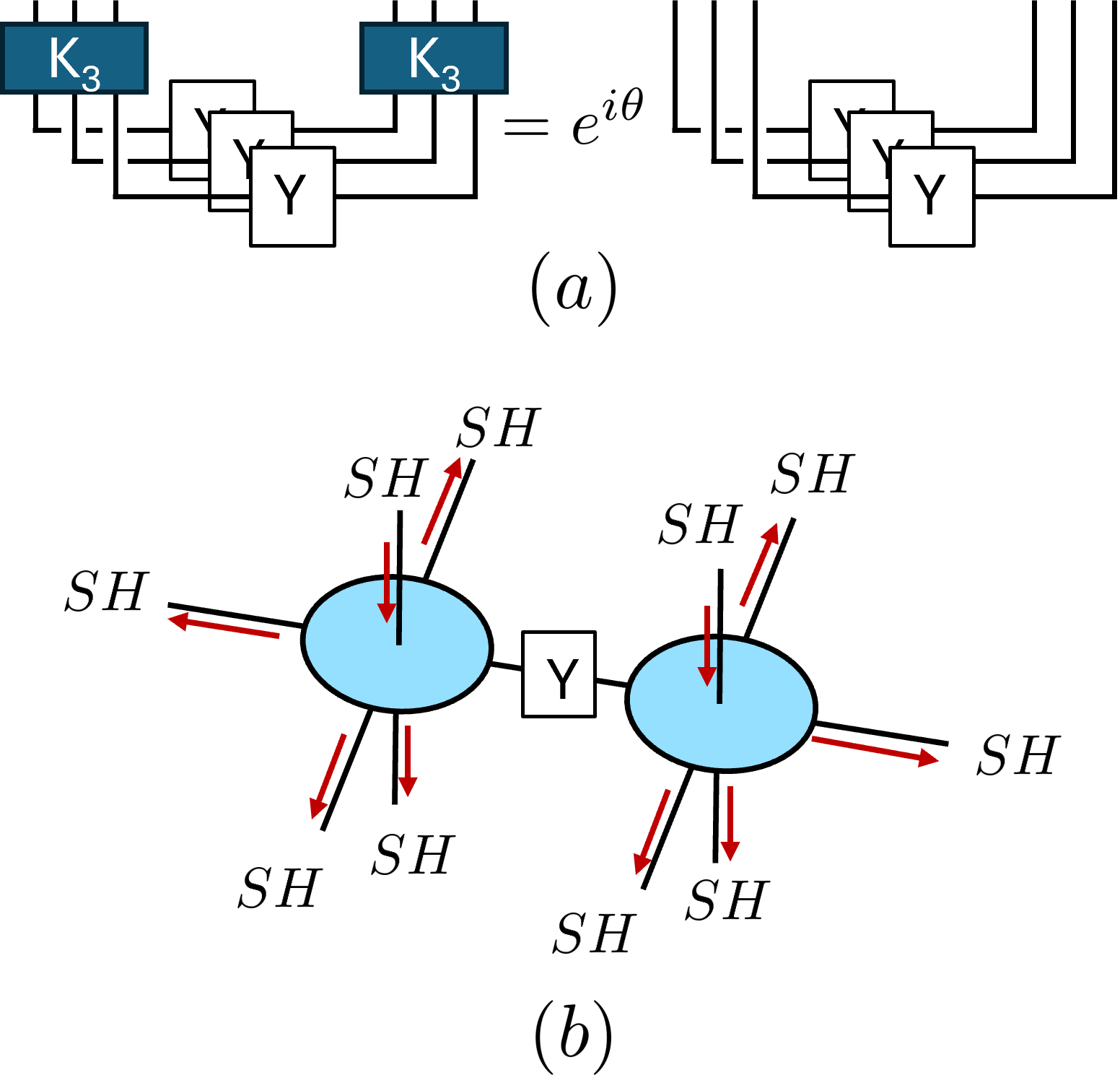}
    \caption{(a)  $K_3\otimes K_3|\Phi_Y\Phi_Y\Phi_Y\rangle\propto|\Phi_Y\Phi_Y\Phi_Y\rangle$. (b) Complete transversality of $SH$ is preserved.}
    \label{fig:PauliDef}
\end{figure}

 Other types of Clifford deformed edges have also been used in literature, such as in building the perfect code \cite{Su_2025}, deformed Steane code \cite{Shen:2023xmh}, graph states \cite{QL2},  twisted surface code~\cite{Yoder_2017,QL}, and hyperinvariant holographic codes \cite{Steinberg:2024ack}.

\subsubsection{Generalized trace with hyperedge states}

To motivate our case for more generalized edge (or hyperedge) states that can fuse more than two legs at once, we examine how intrablock multiqubit logical gates can remain transversal when tracing lego blocks together. For the sake of concreteness, we trace two code blocks each of which supports completely transversal $CZ$ gates, e.g. two Steane codes (Fig.~\ref{fig:CZ_steane}). From Lemma~\ref{lemma:operatormatching} and Corollary~\ref{coro:completetransversal}, we know that performing a single trace with the conventional Bell fusion produces a $[[12,2,3]]$ code with completely transversal $H$ and weakly transversal $S$ gates\footnote{Weakly transversal because $\overline{SS^\dagger}= S^{\otimes 6}\otimes (S^{\dagger})^{\otimes 6}$}. While it is clear that the code still supports bitwise transversal interblock CNOT and $CZ$ gates from the same Corollary, we can use operator pushing to see whether it also supports an intrablock $CZ$ gate on the two logical qubits. 

In Fig.~\ref{fig:CZ_steane}a, we use the symmetry of each Steane tensors to see the logical $CZ$ written as physical $CZ$, where one of these hits the traced edge $|\Phi^+\rangle$. As $CZ|\Phi^+\rangle= |00\rangle-|11\rangle = |\Phi_Z\rangle$, performing transversal $CZ$ as shown in Fig.~\ref{fig:CZ_steane}b deforms the edge. However, the induced $Z$-operator on this edge pushes out to the dangling edges by cleaning it with a stabilizer of one of the Steane tensors. This means that the intrablock logical $CZ$ remains depth-1 and is given by $((IZ)(CZ))^{\otimes 3}(CZ)^{\otimes 3}$ (Fig.~\ref{fig:CZ_steane}c). However, since the code does not support addressable logical $H$, we do not obtain an intrablock transversal CNOT gate for this $[[12,2,3]]$ code. 

Through a Clifford deformed trace using $|\Phi_H\rangle$, or equivalently a local Clifford deformations on one of the Steane codes, we can still construct a deformed version that does support intrablock transversal CNOTs (but not $CZ$s). Similar to Fig.~\ref{fig:CZ_steane}, we push a logical CNOT through the two logical legs to physical legs. This implements CNOTs on all pairs of dangling legs and a CNOT that acts on the connected edge that is $H$-deformed (Fig.~\ref{fig:CNOT_steane} top). Since the CNOT on $|\Phi_H\rangle$ can be commuted to a $CZ$, and from there reduced to a $ZH$ deformed edge, we can clean the $Z$ from of the connected edge again using a stabilizer of a Steane code on the control side, as illustrated by the red dots in Fig.~\ref{fig:CNOT_steane}. The control and target determine what additional Paulis should be pushed. In this case, we have chosen to clean the $Z$ operator, and it will always be pushed to the control side of the code block that executes the logical CNOT. Equivalently, one could push $X$ through the Hadamard and clean the $X$ to the target side.

\begin{figure}
    \centering
    \includegraphics[width=\linewidth]{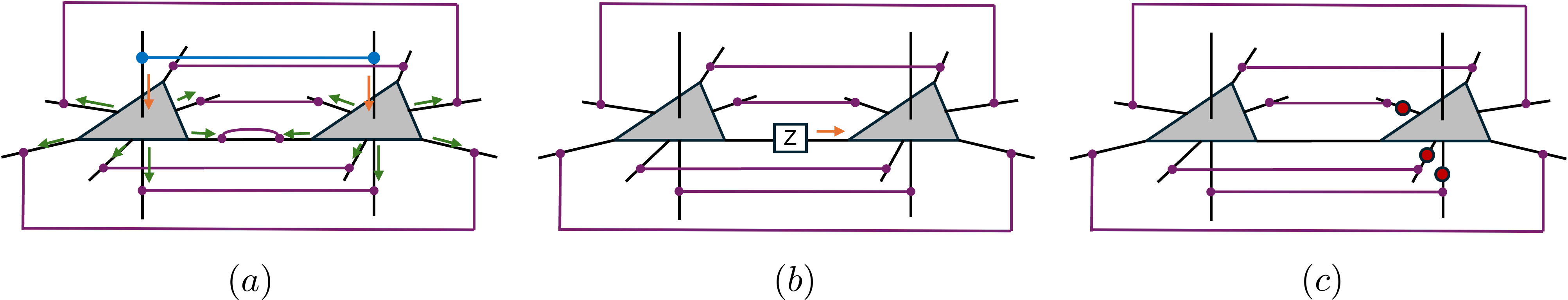}
    \caption{(a) Pushing a logical $CZ$ (blue) down pushes out 7 physical $CZ$s. (b) $CZ$ on the connected edge reduces to a $Z$-deformed edge. (c) $Z$-deformation can be cleaned using a $Z$ check on one Steane code, resulting in 3 additional $Z$ gates (red).}
    \label{fig:CZ_steane}
\end{figure}

\begin{figure}
    \centering
    \includegraphics[width=0.5\linewidth]{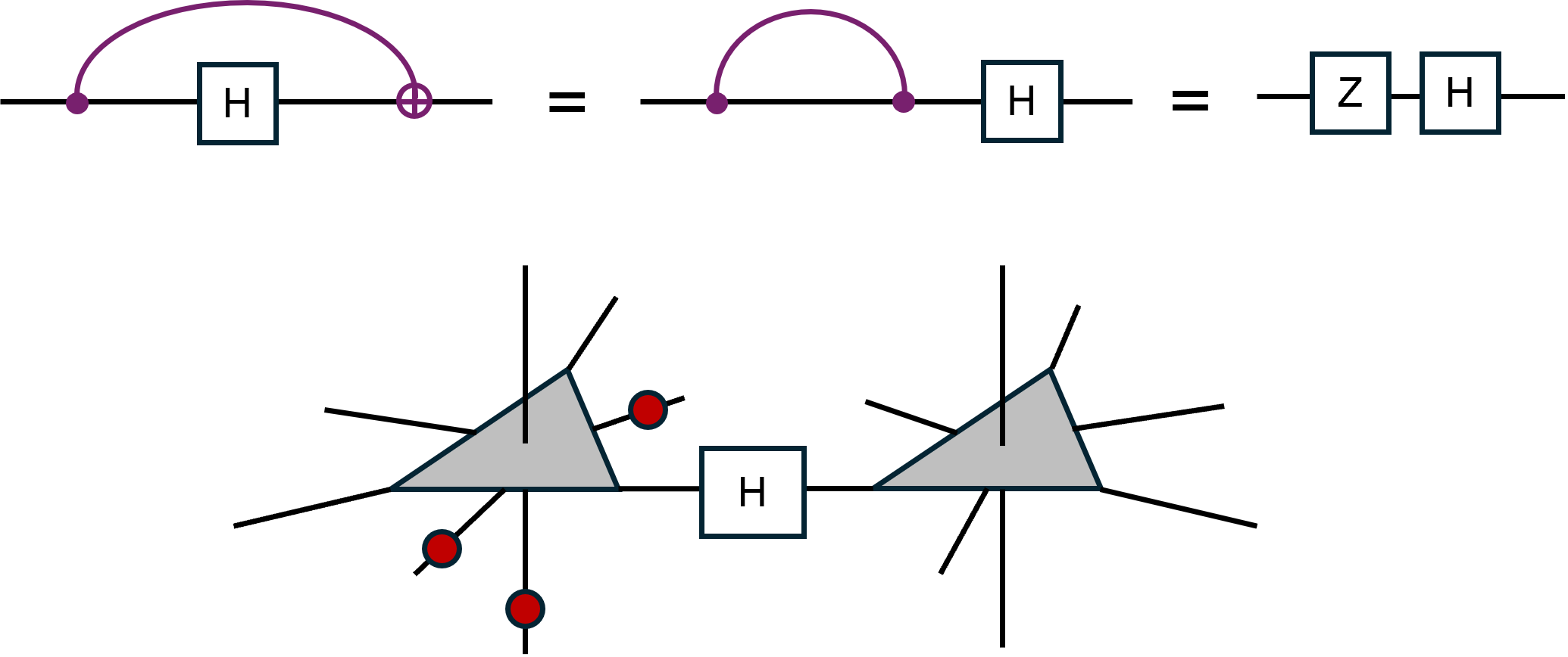}
    \caption{Transversal CNOT on the two logical qubits can be pushed to the dangling edges when using a $H$-deformed edge for fusion. Red dots indicate where Pauli $Z$ operators are needed to clean the induced $Z$ operator along the traced legs.) }
    \label{fig:CNOT_steane}
\end{figure}

As converting between $CZ$ and other locally Clifford-deformed gates simply amounts to using locally Clifford deformed edges or locally Clifford deformed codes, we will focus on $CZ$ and similar multicontrolled gates in the ensuing discussion. Nonetheless, it should be noted that codes supporting other intrablock transversal gates related by local Clifford conjugations can be easily constructed using the same technique. In particular, if the component codes support a transversal controlled-phase gates $CP$, where $P = P(\varphi)$ for some $\varphi$, then the gate acting on the connected edge deforms it as $CP\ket{\Phi^+} = (I\otimes P)\ket{\Phi^+} = |\Phi_P\rangle$. If $P$ can be cleaned by the adjacent tensors with local operators, then the resulting code also supports transversal $CP$ gate. Therefore, the exercise of constructing codes with intrablock $CP$ gates really reduces to subtasks that can identify suitable lego blocks on which $P$ is cleanable.

Unfortunately, the above technique is no longer enough to maintain intrablock $CZ$ targeted on any pair of logical qubits when $k>2$. Namely, if we trace together three or more such blocks together using deformed Bell edges there is no way to restrict the flow of $CZ$s only to two of the logical qubits targeted. Instead, we need a new type of edge state other than Bell states that can stem the flow of $CZ$s to other lego blocks while ensuring that $CZ$s hitting the edge state will leave the state invariant (up to single-site unitaries). In other words, the $CZ$s have to be cleanable on such edge states.

The obvious solution is an $r$-qubit GHZ state that we mentioned earlier. This state is stabilized, for any angle $\varphi$, by all $C^\ell P(\varphi)$ gates up to a local phase rotation $P(\varphi)$.\footnote{We take the convention that $C^0P = P$.} If the single-site phase gates $P = P(\varphi)$ can be cleaned using one of the lego blocks it joined, then the code will have addressable transversal intrablock $C^\ell P$s.  In this sense, the GHZ (hyper)edge state and higher degree generalizations that simultaneously connect three or more code blocks is exactly what we need to create codes with better addressability. For interblock transversal gates, we again note that $C^\ell P\otimes C^{\ell}P^{-1}|GHZ_r\rangle^{\otimes m} = |GHZ_r\rangle^{\otimes m}$ when each $C^\ell P$ acts on $\ell$ of the $m\geq \ell$ copies of the $r$-qubit GHZ states. This means that $C^\ell P$ operator flows can go through $\ell$ copies of the GHZ state without obstruction.

\begin{figure}
    \centering
    \includegraphics[width=0.9\linewidth]{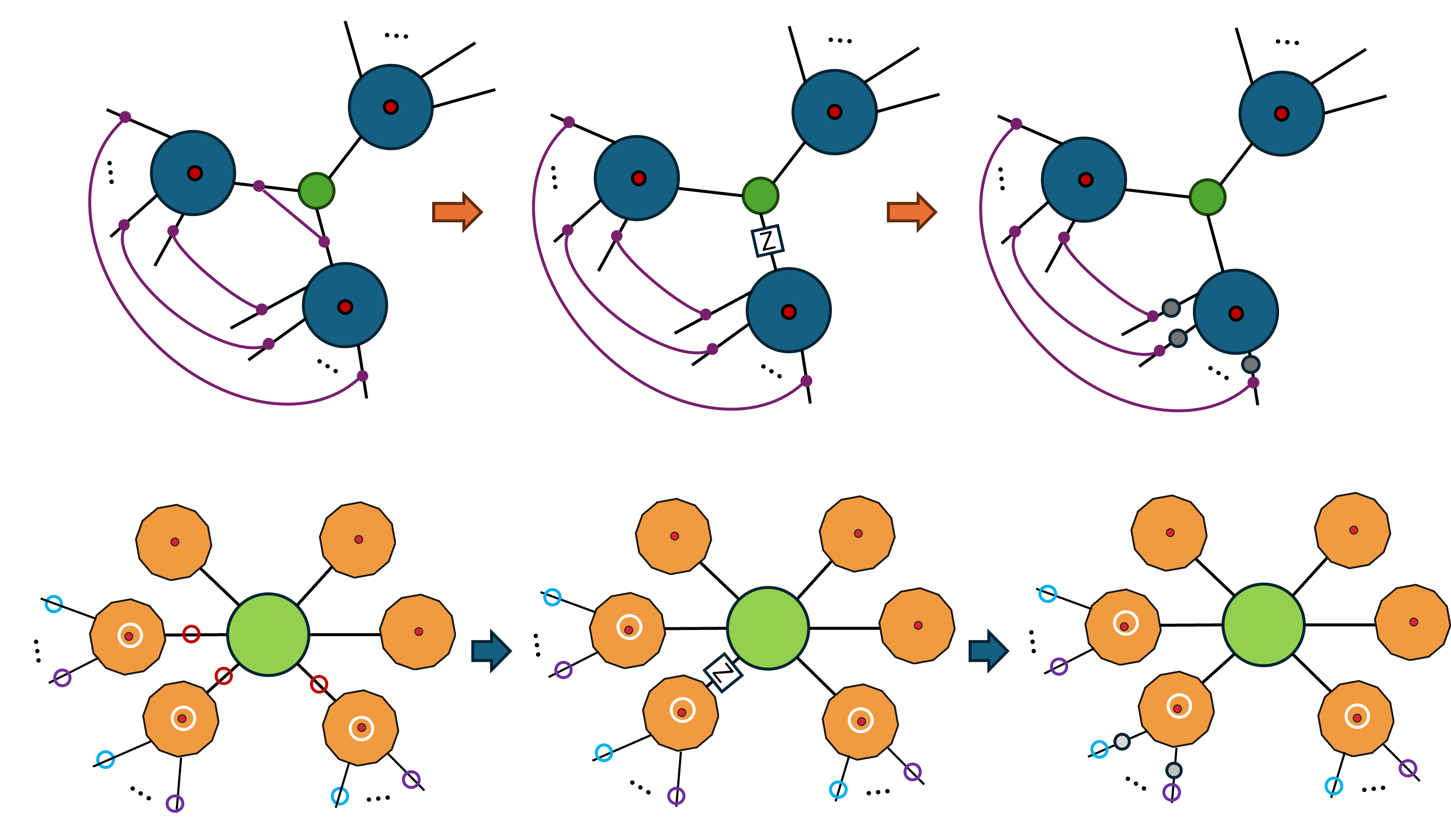}
    \caption{Top: Connecting 3 blocks (blue) each encoding a logical qubit (red) that support transversal $CZ$ gates with a GHZ hyperedge (green) supports a transversal intrablock $CZ$ on any two logical qubits. Bottom: When connecting codes (orange) that support transversal CCZ gates, then the resulting code supports intrablock transversal CCZ on any 3 logical qubits. The logical CCZ acts on the 3 logical qubits circled in white, and is implemented by first running single qubit gates (grey) followed by physical CCZs acting on legs with circles of the same color (blue, purple).}
    \label{fig:ghzedge}
\end{figure}

An example of this is shown in Fig.~\ref{fig:ghzedge} top for $CZ$ gates and $r=3$. One can think of each blue tensor as a Steane code connected through a GHZ hyperedge that connects 3 legs at once. Then for this $[[18,3,3]]$ code (or more generally the $[[6r,r,3]]$ code family), intrablock logical $CZ$s on any two qubits is transversal. Via operator pushing, $S$ gates are partially addressable, namely $\overline{S\otimes S^\dagger\otimes I}=S^{\otimes 6}\otimes (S^{\dagger})^{\otimes 6}\otimes I^{\otimes 6}$ acts on any of the two logical qubits. Logical Hadamards are no longer transversal for this code because $H\otimes U_1\otimes\dots\otimes U_{r-1}$ does not stabilize $|GHZ_r\rangle$ (for $r>2$) for any choice of single-qubit unitaries $\{U_i\}$. Instead, a single $H$ pushes to $H$ followed by a sequence of CNOTs on the remaining legs. 
Finally, the same code also admits transversal interblock $\overline{CZ}\otimes\overline{CZ}\otimes \bar{I}$ gates acting on two pairs of the logical qubits at once, which are partially addressable. 

There is no reason to restrict ourselves to examples of Clifford gates as the GHZ hyperedge works just as well with other controlled phase gates. We join together by a GHZ edge the $[[12,1,2]]$ code or its family in \cite{Vasmer_2019} that support transversal CCZ gates.  Intrablock CCZ on any triple of the logical qubits is transversal (Fig.~\ref{fig:ghzedge} bottom) where a family of $[[11r,r,2]]$ codes can be built by gluing together $[[12,1,2]]$ codes (orange) with a single $r$-qubit GHZ hyperedge (green). Then CCZ on any 3 logical qubits is transversal on a single code block. Interblock CCZs are also transversal but only have partial addressability where they need to act on all logical qubits on two different branches.

We summarize the above observations as the following 
formal statement.
\begin{proposition}\label{prop:GHZcomptran}
    Let $|V\rangle=\bigotimes_{i=1}^r|V_i\rangle$ where each $|V_i\rangle$ represents the Choi state of a non-trivial code (i.e. $d\geq 2$) that supports a completely 
    transversal $U = C^{\ell}P(\varphi)$ gate for some angle $\varphi$ and $\ell \geq 0$. Provied that $P(\varphi)$ can be cleaned off of any traced leg using logical identity operators, the generalized GHZ trace with $|GHZ_r\rangle$ where $r > \ell$ preserves intrablock transversality of $U$ on any subset of the logical qubits. 
\end{proposition}
\begin{proof}
    The proof is straightforward by pushing any such $U$. For $U=P(\varphi)$, the transversality and addressability of logical $P(\varphi)$ are trivially guaranteed by the condition that $P(\varphi)$ can be cleaned off of the traced leg. For $U=C^{\ell}P(\varphi)$, pushing $U$ to the GHZ edge produces a $P(\varphi)$ on a single edge, which can then be cleaned.
\end{proof}

More generally, we need to consider hyperedge states like the GHZ states that can glue together multiple legos all at once. 
\begin{definition}
Let $|V\rangle$ be the Choi state of a tensor network over $n>r$ qudits. The generalized trace with (hyper)edge state $|\Phi\rangle_{A_1A_2\dots A_r}$ over $r$ qudits is given by the tensor (network) of the state $\langle \Phi|V\rangle$. 
\end{definition}

Some forms of this fusion is present in tensor networks like PEPS or random tensor networks \cite{Hayden_2016} where $|\Phi\rangle$ is variationally tuned to minimize certain costs like energy or chosen from a Haar random ensemble. For applications to quantum codes, we need to look for edge states with symmetries that can support desirable transversal gates. For example, states with intrablock symmetries like those in Fig.~\ref{fig:multicopy_tensor}b can be used to glue together codes and retain intrablock transversal logical gates---a Steane tensor as a generalized edge that glues together 8 other Steane codes will permit intrablock $CZ$s on 4 pairs of logical qubits simultaneously.

For the rest of this work, we will mostly focus on the application of the GHZ tensor as generalized hyperedges that produce addressable $P$ and $C^{\ell}P$ gates through quantum lego and leave the complete classification of generalized edge states as well as their impact on transversality and addressablility to future work.

\subsection{Localized symmetry, cleanability, and addressability}

In the previous section we encountered codes that support transversal gates of different degrees of addressability ranging from completely transversal gates like $H$ in fused Steane codes, to partially addressable gates like phase gates or interblock $CZ$ gates in $[[6r,r,3]]$ codes, and finally to fully addressable intrablock transversal $CCZ$ gates in $[[11r,r,2]]$ codes. A key but simple observation is that addressability descends directly from the localization of symmetries in a tensor network. That is, the unitary symmetry of an encoding tensor only has support on a subset/subregion of the legs, where better addressability comes from more localized symmetries. In addition, we need a wide variety of such localized symmetries that do not have support on the logical legs of interest. For a code with $k>\ell$ logical qubits, intrablock addressable gates would imply that the tensors of interest have at least ${k\choose \ell}$ such localized symmetries for $C^\ell P$ gates. This imposes a large number of constraints which will limit the search space, but a first-principle identification of such tensors remains non-trivial.

So, at least for stabilizer codes, we can find codes with localized symmetries by instead searching for codes with global symmetries. However such symmetries always follow a preprogrammed support (matching the logical-$X$ operators), and in particular can not create ``off-diagonal'' interblock gates when $k > 1$. Instead we go a different direction: an easy heuristic for building tensors that have localized symmetries is by gluing together atomic legos that have such properties, e.g. GHZ states or repetition codes. 
For GHZ states or $Z$-spiders, we recall that these tensors support localized symmetries related to $C^\ell P$ gates, as illustrated in Fig.~\ref{fig:zxhcalc}. Since $C^\ell P$ can be broken down into the contraction of $Z$-spiders and phase tensors, we can push such gates from one leg to another where $PP^{\dagger}$ (or $C^\ell P\otimes C^\ell P^{\dagger}$) stabilizes ($\ell$ copies of) the GHZ tensor. Combining these units with other tensors with similar global symmetries will then synthesize other tensor gadgets that have localized symmetries. For example in Fig.~\ref{fig:tritensor}, a single tri-branch tensor is constructed by concatenating the repetition code with codes that have completely transversal $C^\ell P$ gates. Thanks to the localized symmetry of the bit-flip repetition code, the tri-branch tensor can now have $C^\ell P$ symmetries supported on no more than 2 branches. We can formalize such codes with localized symmetries as follows.

\begin{figure}
    \centering
    \includegraphics[width=0.9\linewidth]{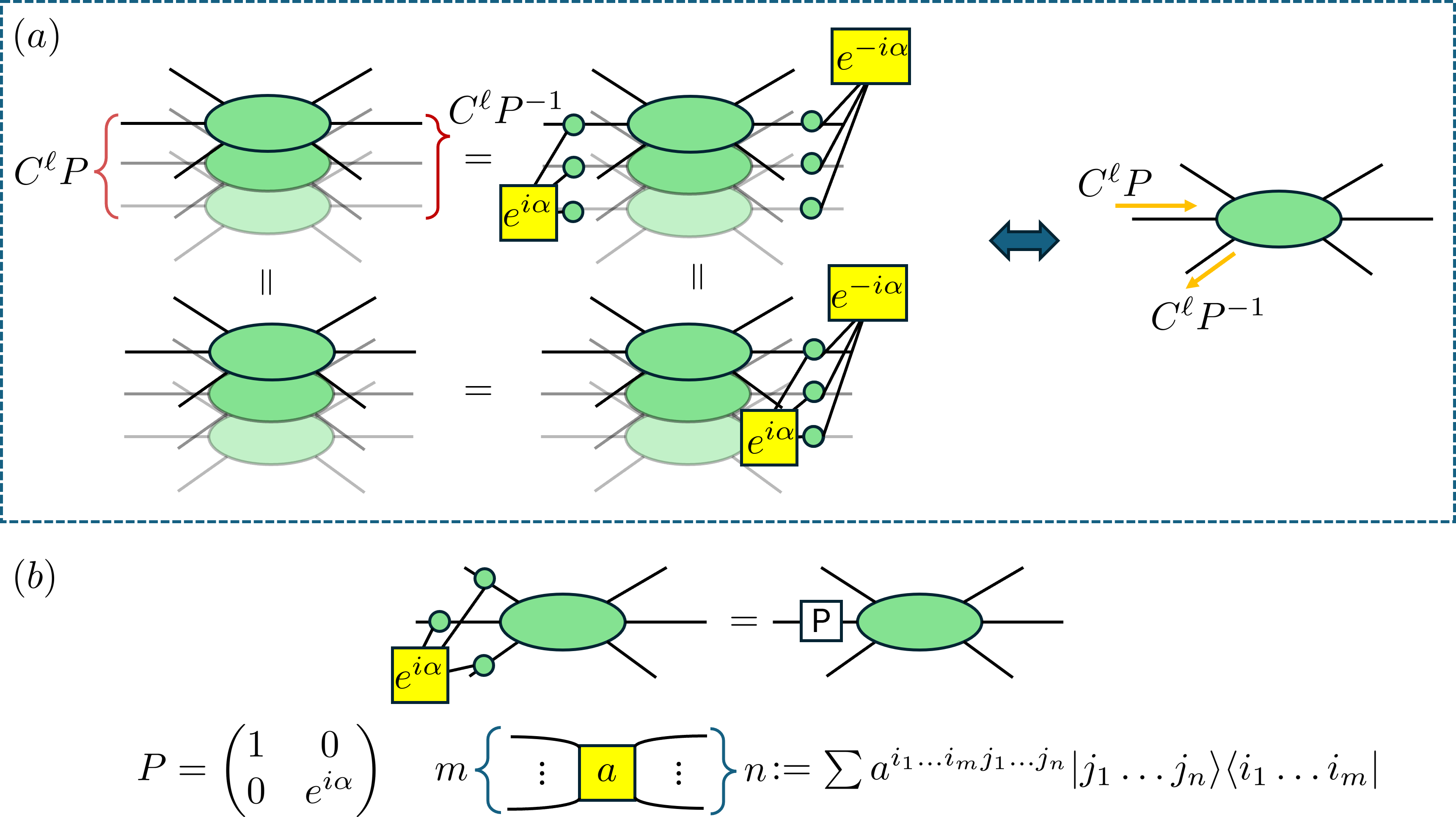}
    \caption{(a) Interblock transversal gates are given as localized symmetries over the GHZ tensor (green). For simplicity, even though operator pushing of the multi-qubit interblock gates is over $m$ copies of the code, we only show the operator pushing on one copy (right). (b) Intrablock gates will descend from the localized symmetry. Here $P = P(\alpha)$ is a phase gate and the yellow ``H-box'' is defined as the tensor of the map at the bottom.}
    \label{fig:zxhcalc}
\end{figure}

\begin{figure}
    \centering
    \includegraphics[width=0.4\linewidth]{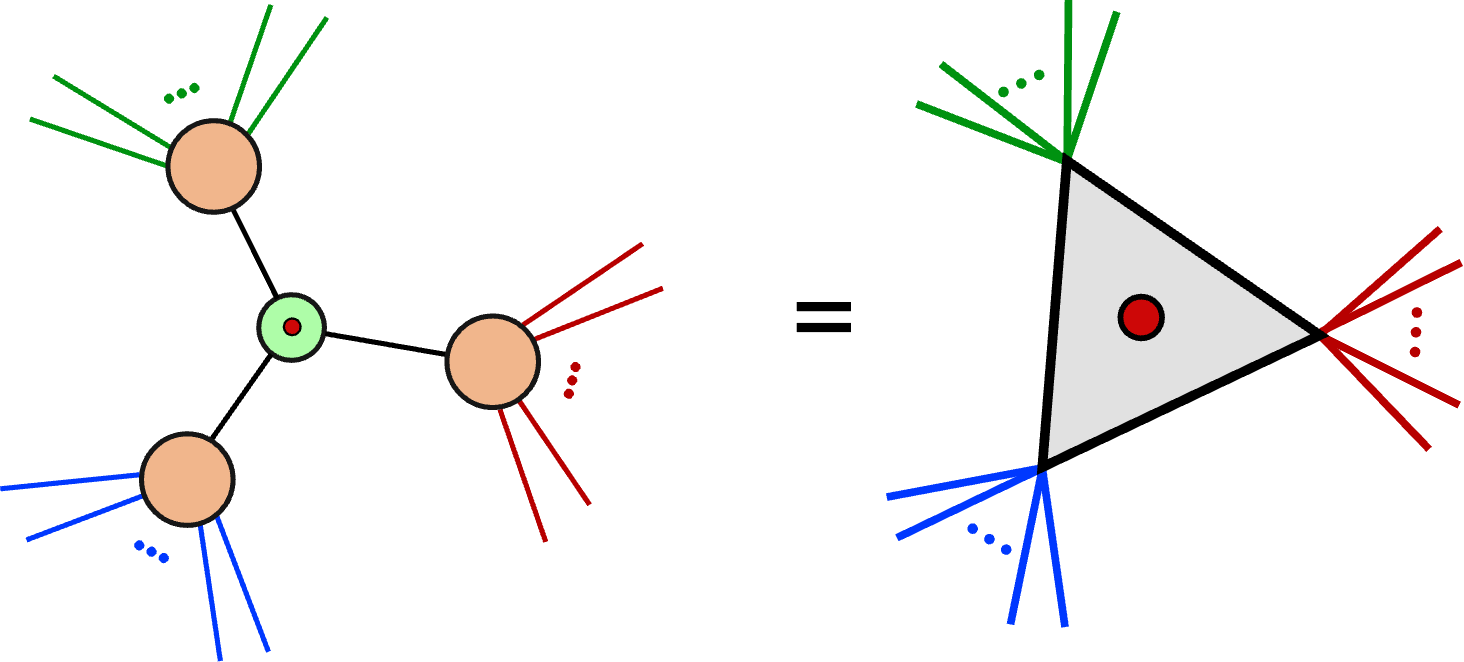}
    \caption{A single $Z$-spider of valence 4 is fused with three other components (orange) which support bitwise transversal $C^{\ell}Z$ gates. The central red dot represents the logical leg. Hence a single logical $C^\ell Z$ can be pushed to any of the three branches. In the same way, bitwise $C^\ell Z$ acting on any branch can be pushed to bitwise $C^\ell Z$ acting on another branch of a different color.}
    \label{fig:tritensor}
\end{figure}

\begin{definition}\label{def:branchcode}
    A \emph{($b$-)branch code} is constructed by concatenating a $[[b,1,d_x=b/d_z=1]]$ bit-flip repetition code with $b$ number of $[[n_0,1,d]]$ ``seed'' codes. In the resulting $[[bn_0,1,d]]$ branch code, $b$ distinct colors are assigned to each branch of the concatenated $[[bn_0,1]]$ code each with $n_0$ leaves. Suppose each $[[n_0,1,d]]$ seed code supports strongly transversal (controlled-)phase gates $\mathcal{U}=\{U\}$ where each $U\in\mathcal U$ has $U^N=I$ for some $N$, then we say it is a $(b,\mathcal U)$-branch code.
\end{definition}

Note that if a branch code supports $U\in\mathcal{U}$, then it also supports $U^r$ for any integer $r$ and hence it suffices to describe $\mathcal{U}$ by the finest phase angle supported at each arity of control. Yet, as all stabilizer codes admit transversal Pauli $Z$, we omit explicitly referencing it when relevant.
For example, gluing Steane tensors with a $b$-repetition code, we obtain a $(b,\{CZ,S\})$-branch code. In the same way, we can build $(b, \{CCZ\})$-branch code from contracting the $[[12,1,2]]$ blocks (where $Z\in \mathcal{U}$ implicitly), and $(b, \{T,CS,CCZ\})$-branch code from contracting 15-qubit quantum Reed-M\"uller (QRM) codes. In the last example, the $CZ$ is more localized compared to $T$ or $CS$ gates. 

\begin{lemma}
    Consider a $(b,\mathcal{U})$-branch code. Then every $C^\ell P\in \mathcal{U}$ can be localized to a single branch. That is, we may take the logical $\overline{C^\ell P}$ to be supported  on a single branch.
\end{lemma}
\begin{proof}
    Any $C^\ell P$ gates can be pushed from the logical leg of the repetition code to any one of its physical legs as a power of the $C^\ell P$ gate. Since each seed code is transversal, these can then be pushed into a single branch. 
\end{proof}

\begin{lemma}
    Consider a $b$-branch code with each branch seed having distance $d\geq 2$, and $A$ a set of legs containing at most one leg from each branch. Then any logical Pauli operator can be cleaned so as to have no support on $A$.
\end{lemma}
\begin{proof}
    A logical Pauli operator on the repetition code may be supported on all $b$ legs. However, from the usual cleaning lemma \cite[Lemma 1]{bravyi2009no}, any logical Pauli operator on one of the $[[n_0,1,d]]$ branch seeds can be cleaned off of any single physical leg using a stabilizer of the branch seed. Hence we can clean the logical Pauli operator branch-by-branch off the leg of $A$ on that branch where needed.
\end{proof}

For quantum codes, we focus only on the cleanability of the physical legs in the encoding tensor since the logical legs of an encoding isometry are always cleanable. Then $r$-cleanability ensures that there is not only localized symmetry, but enough localized symmetries in the tensor to push operations where needed. Cleanability also guarantees transversality and addressability. 

\begin{proposition}\label{prop:clean_transversal}
    If an $[[n,k]]$ code is $r$-cleanable with respect to $U$, then (interblock) logical $U$ is fully addressable and (interblock) transversal. 
\end{proposition}
\begin{proof}
    The definition of $r$-cleanability included that any product of $U$ acting on any subset of logical legs can be cleaned. Since the encoding tensor of a code is an isometry, they must correspond to physical operators that are tensor products of $U$ or its power. Hence such operators $U$ are both transversal and addressable. For multi-qubit logical operators $U$, they correspond to interblock gates.
\end{proof}

\begin{lemma}\label{lemma:branchcode}
    Consider a $(b,\mathcal U)$-branch code, with $b > 2$, and where each seed code has distance $d\geq 2$. Let $U \in \mathcal{U}$ with $U^N = I$ for an integer $N > 1$. Then any power of $U$ is cleanable off any two physical qubits in two distinct branches. Furthermore, the code admits a localized transversal logical $U$ gate.
\end{lemma}
\begin{proof}
    In the case $U = P$, the repetition code of length $b$ is stabilized by $P^{a_1}P^{a_2}\dots P^{a_b}$ whenever $N \:|\: \sum_i a_i$. Therefore, any two phase gates acting on two physical qubits $P^{a_i}P^{a_j}$ can be pushed to a third physical qubit that has $P^{a_k}$ where $a_k + a_i + a_j = 0\pmod N$. Hence any single-qubit phase gates in $\mathcal{U}$ are 2-cleanable.

   For multi-qubit gates, any $C^{\ell} P$ gate can be written as the contraction of a H-box with $Z$-spiders, for example see \cite{East_2022} and Fig.~\ref{fig:zxhcalc}. Since $Z$-spiders commute through the repetition code, which is also a $Z$-spider, any logical $C^\ell P$ gate can be realized as a single $C^\ell P$ gate on a group of $\ell$ qubits on the repetition code. Since $U$ is bitwise transversal on each branch, the logical operator can thus be written as $\overline{C^\ell P}=\bigotimes_{i\in \mathrm{branch}} C^{\ell}P_i$ where $i$ denotes a group of $\ell$ qubits when stacking $\ell$ copies of the branch code and $C^{\ell}P_i$ acts transversally on all groups that belong to the same branch.
   
   Now, just as above, any two $C^{\ell}P^{a_i}$, $C^{\ell}P^{a_j}$ gates acting on two groups of $\ell$ physical qubits can be pushed to $C^{\ell}P^{a_k}$ as long as $a_k+a_i+a_j=mN$ by commuting the Z spiders through the repetition code and use $C^{\ell}P^{a_i}C^{\ell}P^{a_j}C^{\ell}P^{a_k}=I$. For the branch code, for any incoming $C^{\ell}P^{a_i}, C^{\ell}P^{a_j}$ on two legs from distinct branches, they are pushed to tensor product of those same gates on the remaining qubits on those two respective branches plus the tensor product $(C^{\ell}P^{a_k})^{\otimes n_0}$ acting on a third branch.

\end{proof}

Using seed codes with localized symmetries like the ones above, we can now build up codes using the same technique in previous section but with more addressability.

\begin{lemma}\label{lemma:comptensorcode}
    Consider a $[[a,k,d_0]]$ code with $d_0 \geq 2$ that supports localized transversal $C^\ell P$ where $P^N=I$ for some integer $N$ and $C^\ell P$ is 1-cleanable. Suppose we join together $b\geq \ell$ of them using a GHZ tensor by one physical leg from each $a$-qubit code block, then this tensor network $\mathcal T_b$ is a $[[(a-1)b,kb,d_0]]$ code which supports addressable interblock, mesoblock, and intrablock transversal $C^\ell P$ gates on any $\ell$ logical qubits. 

    If $d_0\geq 3$, then the interblock $C^\ell P$ is fault-tolerant. If the shortened dual $[[a-1,k+1,d_0']]$ code formed by choosing the traced leg as a logical qubit further has $d_0'\geq 3$, then the intrablock and mesoblock $C^\ell P$ gates are also fault-tolerant.
\end{lemma}

\begin{proof}
    Since $d_0\geq 2$, we first make use of Lemma~\ref{lemma:isodist} where it is clear that the code has the proposed parameters. In particular, $d_0\geq 2$ ensures there there is no kernel to the encoding map and all logical legs indeed correspond to independent logical qubits. The tensor network is also invariant under permutation of any of the branches. For any interblock gate, since each code is 1-cleanable, there must exist a representation of the logical $C^\ell P$ operator that does not act on the leg being traced. That is, the physical $C^\ell P$ can be implemented transversally entirely on the physical dangling legs across $\ell$ blocks. Since the tensor network is invariant under branch permutation, the interblock gate can apply transversally for any $\ell$ logical qubits one on each code block. For intrablock gates, since the $C^\ell P$ operator can be pushed to the same boundary legs on each of the $b$ branches, and the branches are identical to each other, it follows that the $C^\ell P$ gate can be implemented transversally on any $\ell$ qubits by performing transversal $C^\ell Z^{\otimes m\geq d}$ on those $\ell$ branches. Since the operator flows on any intrablock gates do not intersect on any code block, it follows that mesoblock gates that mix interblock and intrablock operator pushing are also transversal by acting on the physical legs where the operator flow exits. 

    If $d_0\geq 3$, then it is clear that the interblock transversal gates are FT because each block can correct any single error. An $\ell$-qubit gate that is transversal spreads 1 error to at most $\ell$ errors but on distinct code blocks each correcting one error. Furthermore, if the shortened $[[a-1,k+1,d_0']]$ code obtained by converting one physical leg $m$ to a logical qubit corrects any single error, then the intrablock gates are also FT if $m$ is traced with the GHZ state. By construction, any single error can be corrected on each branch. As the $\ell$ qubit transversal gate spreads errors to distinct branches connected by the GHZ state, we can correct the error on the branches individually. By extension, the mesoblock gates are also FT.
\end{proof}

Note that requiring $d'_0\geq 3$ is strictly stronger than having $d_0\geq 3$ because $d_0\geq d_0'$ and the shortened code typically has lower distance (unless the code is non-degenerate) as the logical operators that have support on the new logical leg for the dual code come from stabilizer operators in the unshortened code. If the code is highly non-degenerate, then the dual code can have $d'_0\ll d_0$. For example, a surface code of size $L\times L$ can have $d_0=L$ but $d_0'=3$. For the construction above, it is sufficient to have the unshortened seed code be non-degenerate with $d_0\geq 3$ and all non-identity stabilizer weights 4 or higher.

\section{Applications: codes with global transversal gates}
\label{sec:4}
We now construct examples in increasing complexity to illustrate steps in designing codes that have transversal single or multi-qubit gates with varying degrees of addressability.
We begin with codes that have global non-Pauli transversal gates but are not addressable.

\subsection{Single qubit transversal gates}
\label{subsec:4.1}

By matching lego blocks that have (weakly) transversal single qubit gates, one can easily generate bigger codes with the same transversal property with possibly additional local unitary deformations on different legos and edges to ensure the uniformity of the construction.

For example, consider gluing two codes both of which have bitwise transversal Hadamard. Since $H^*=H$, it is clear that these transversal gates remain symmetries of the tensor (network) no matter how legs are contracted. Therefore, the Hadamard remains a bitwise transversal logical gate in resulting code also. This includes both case of generating bigger codes by gluing multiple code blocks as well as creating smaller codes via self-tracing. The $[[n,k,3]]$ holographic Steane code \cite{Harris_2018} constitutes one such an example, where each atomic lego is a Steane code that carries (strongly) transversal $H$, and therefore a global $\bar{H}^{\otimes k}=H^{\otimes n}$ is a logical gate of the code. However, to preserve the complete transversality of $S$ gates in the same code requires extra care because $S^*=S^\dagger$. However, we can solve this matching problem easily using Lemma~\ref{prop:3.2} where the seed codes are fused with $|\Phi_X\rangle$ instead. 

Applying the ``matching'' Lemma~\ref{lemma:operatormatching}, one can create codes with different classes of transversal gates. Here we focus on two that are of special interest. 

\subsubsection{Example: Transversal $SH$ gates} 
\label{subsubsec:SHgate}
While transversal Hadamards and $S$ are relatively easy to spot with self-dual CSS codes, codes with other transversal Clifford gates such as $SH$ are far less understood with the CSS construction or stabilizer formalism. There is an added interest for identifying such codes (also called $M_3$ codes) because the eigenstates of such gates are magic states\cite{Bravyi_2005}. Although they  correspond to linear self-orthogonal codes over GF(4)\cite{kalra2025invarianttheorymagicstate}, systematic classifications of these codes have been far less explored compared to other MSD protocols like triorthogonal codes. With quantum lego, $M_3$ codes can easily be generated by gluing $[[5,1,3]]$ codes that support transversal $SH$. However, since $(SH)^*\ne SH$, a Pauli-deformed edge $|\Phi_Y\rangle$ is needed as $SH\otimes SH|\Phi_Y\rangle\propto |\Phi_Y\rangle$. 

\begin{figure}
    \centering
    \includegraphics[width=0.6\linewidth]{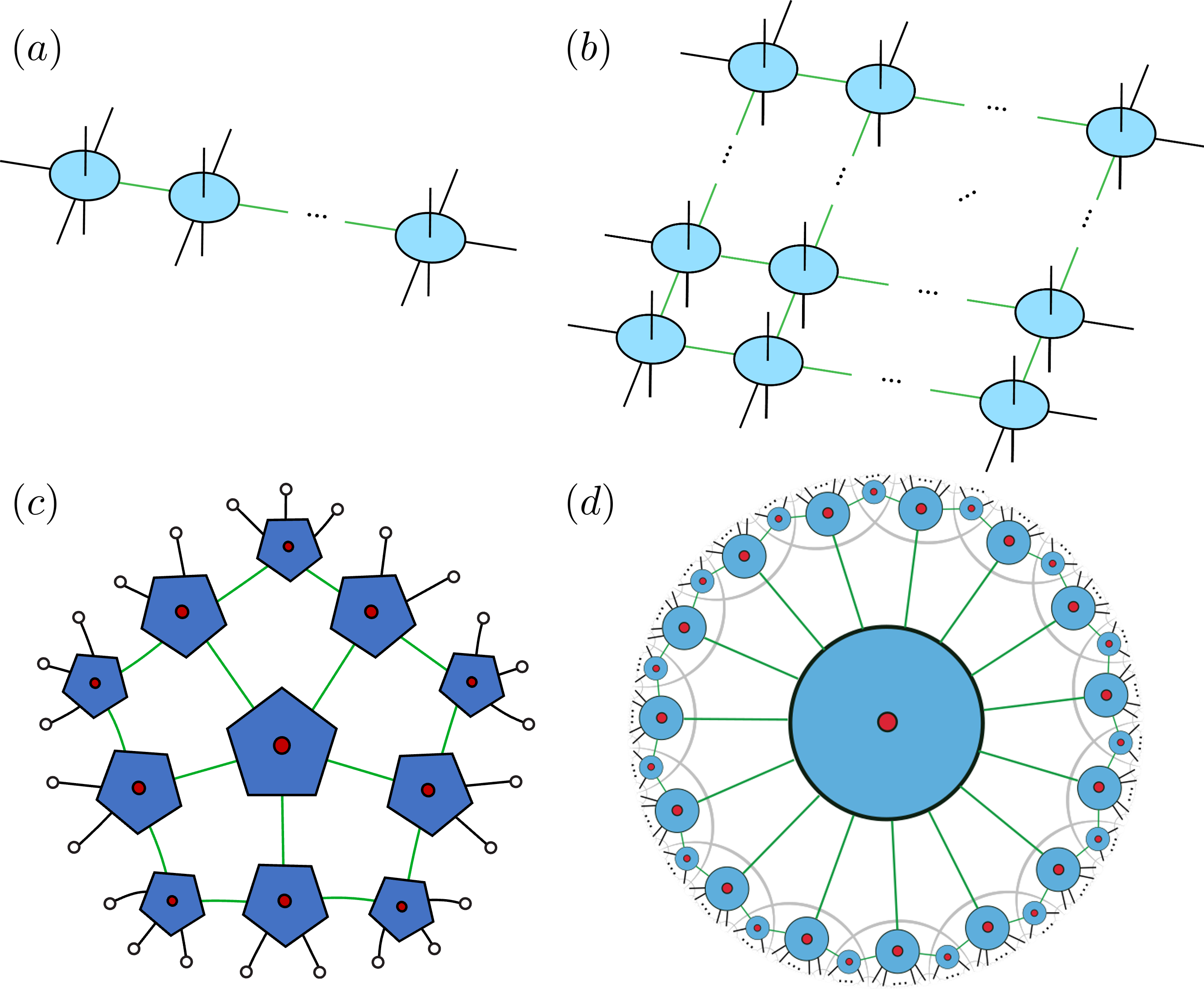}
    \caption{(a,b) A linear and planar code where each tensor represents a $[[5,1,3]]$ perfect tensor, having transversal $SH$ gate. Alternatively, each tensor could be a $[[15,1,3]]$ Reed-M\"uller tensor, having transversal $T$ gate, where each downward-pointing dangling edge contain $11$ qubits and hence has bond dimension $2^{11}$. (c) A $[[25,11,3]]$ holographic pentagon code with bit-wise transversal $SH$ gate. (d) An $[[n,k,3]]$ holographic Reed-M\"uller code with transversal $T$ gate. In all of the diagrams, red dots represent not shown logical legs/degrees of freedom and green edges represent a generalized trace with a Pauli $X$ or $Y$ insertion for transversal $T$ or $SH$ gates respectively. }
    \label{fig:1qbit_code}
\end{figure}

Arbitrary fusion of these seed codes will then generate a code with transversal $SH$, but we focus on a few special examples shown in Fig.~\ref{fig:1qbit_code} where each $Y$-deformed edge is marked in green. In Fig.~\ref{fig:1qbit_code}c, an edge-deformed holographic HaPPY pentagon code $[[n,k,3]]$ with transversal $(\overline{SH})^{\otimes k}=(SH)^{\otimes n}$ is given, which has asymptotic rate of $k/n\approx 0.447$ \cite{Pastawski_2015}. In Fig.~\ref{fig:1qbit_code}ab, we construct a linear $[[3L+2,L,3]]$ and a planar $[[L^2+4L,L^2,3]]$ code with Euclidean geometries. Notably, the latter has an asymptotic unit rate $k/n\rightarrow 1$. The encoding map of each code here can be written as a sequence of contractions of isometric tensors of bounded degree, which ensures that the overall encoding map is an isometry and that the tensor network is efficiently contractible. It then follows from \cite{QL2} that each code admits a polynomial time maximum likelihood tensor network decoder and their weight enumerators can be obtained efficiently.

A natural question is whether any of these codes can be used to extract magic states used in the BK protocol. In this procedure, one inputs $n$ noisy noisy magic states $\rho=(1-\epsilon)|T_0\rangle\langle T_0|+\epsilon|T_1\rangle\langle T_1|$ and perform syndrome measurements, postselecting on the trivial syndrome. Here $|T_{0,1}\rangle\langle T_{0,1}|=\frac 1 2 (I\pm \frac 1 {\sqrt{3}}(X+Y+Z))$. We might be tempted to conclude that the planar code with asymptotically unit rate already leads to a distillation protocol that has asymptotically constant overhead. However, unlike the Bravyi-Haah protocol, the code distance itself is insufficient to guarantee a stable fixed point at the desired magic state. We show in Appendix~\ref{app:msd} that the line code remains a valid Bravyi-Kitaev-like MSD protocol but the 2d code and the pentagon HaPPY code do not.

If one insists on using undeformed edges, then care must be taken to account for the flipping the $SH$ gates into $(SH)^*$ on adjacent seed codes. One such example was discussed in holographic codes \cite{Williamson}, which results in a weakly transversal gate. In such examples, things can be remedied by applying local Pauli $Y$-deformations on the physical qubits, hence the end result is not much different from what we constructed above. In general, however, tensor networks obtained from $Y$-deformed traces are not always equivalent to those obtained from undeformed traces up to local unitary rotations on the physical/logical legs. For example, if three such seeds are contracted into a loop using undeformed Bell fusion, then the matching condition to preserve even weakly transversal $SH$ cannot be satisfied and the symmetry will be lost altogether.

\subsubsection{Example: Transversal $T$ and other single-qubit phase gates}

It was shown in \cite{QL} that by tracing (on arbitary number of legs) two atomic legos with bitwise transversal $T$ gates, e.g. arising from triorthogonal codes, the resulting code also has a weakly transversal $T$ gate. Specifically, $T^{\otimes n_1}\otimes (T^{\dagger})^{\otimes n_2}$ is a transversal gate, where the first $n_1$ qubits belong to the first atomic lego and the next $n_2$ belong to the second. The inversion of $T$ is to satisfy the matching condition $T\leftrightarrow T^*=T^{\dagger}$ along the traced legs. This inversion is inconvenient if we want to preserve strong transversality of the logical gates and to allow for consistent self-traces on the same tensor (network). Here we apply Lemma~\ref{prop:3.2} by tracing seed codes with $|\Phi_X\rangle$ where $T\otimes T|\Phi_X\rangle\propto |\Phi_X\rangle$ so that $T$ can be matched to $T$ with deformed traces. 
With such deformed traces, smaller codes like the $[[7,1,1]]$ code have been obtained through tensor self traces \cite{Shen:2023xmh}. Similar to Sec.~\ref{subsubsec:SHgate}, codes with other geometries, such as a planar $[[4L+11 L^2,L^2,3]]$ (Figure~\ref{fig:1qbit_code}b) or linear $[[13L+2,L,3]]$ (Figure~\ref{fig:1qbit_code}a) constructions are also possible using the same kind of tensor network but with $X$-deformed traces. These codes similarly admit efficient optimal tensor network decoders and enumerator computations. Nevertheless, to maintain the isometric property of the encoding map, we achieve a much lower rate compared to the transversal $SH$ codes above and the class of triorthogonal codes by \cite{BravyiHaah}. 

Thanks to the generality of Lemma~\ref{prop:3.2}, the same technique can also be used to generate codes with other transversal non-Clifford phase gates higher level Clifford hierarchy~\cite{Webster_2023}, provided the right seed codes are used.

\subsection{Multiqubit transversal gates}
\label{subsec:4.2}

Recalling that multiqubit transversal gates correspond to the symmetries that preserve multiple copies of a tensor, we can use the same matching lemma for identifying transversal multiqubit gates. As we see earlier, since any (two-stack of a) CSS code supports completely transversal CNOT gates as a symmetry, operator matching tells us that tracing two CSS codes in any arbitrary manner retains this symmetry and hence support completely transversal CNOT gates also. Of course we already know this from a different perspective, where \cite{QL} showed that the (self)trace of any CSS is CSS, and hence must support completely transversal CNOTs. However, with operator matching, we can also examine other transversal gates, such as $CCZ$, or even more exotic multiqubit gates \cite{Huang_2023}.

\subsubsection{Example: Codes with Transversal Interblock $K_3$ Gates}
A slightly unusual example is the 3-qubit Clifford gate $K_3$ \cite{Gottesman_1998}.
This gate is strongly transversal on the $[[5,1,3]]$ code. Therefore, to generate classes of codes that have it as a transversal gate, we can use the perfect code as atomic legos. To satisfy the matching condition, one easily checks that $(YYY)K_3 (YYY) = K_3^*$. As $K_3\otimes K_3^*(|\Phi^+\rangle)^{\otimes 3}=(|\Phi^+\rangle)^{\otimes 3}$, this implies that a bitwise $K_3$ gate satisfies the operator matching condition in the codes in Figure~\ref{fig:1qbit_code}abc where $Y$-deformations are applied to the green edges. Therefore, the codes considered in Fig.~\ref{fig:1qbit_code} of Sec.~\ref{subsubsec:SHgate} admit not only transversal $SH$, but also bitwise transversal $K_3$. 

\subsubsection{Example: Mixed single and multi-qubit intrablock gates}
Using the same Lemma \ref{prop:3.2} and Corollary \ref{coro:isotrace_Tcode}, we can also build up codes sequentially through isometric traces by combining atomic codes with transversal gates of the form $\bar{G} = T^{\otimes n}$ where $\bar{G}$ is an intrablock multi-qubit gate. We need to use $|\Phi_X\rangle$ fusion instead of the usual Bell fusion to preserve strong transversality. A similar construction without deformed edge is also possible \cite{weirdgates}. 

For instance, using $[[8,3,2]]$ codes that has $\overline{CCZ}=T^{\otimes 8}$ as seed codes, we notice that a tree tensor network of the form in Fig.~\ref{fig:CCZ_transv} having degree $8$ in-plane legs instead of $12$ with $3$ bulk/logical qubits on each tile (represented by the red dot) will admit a $\overline{CCZ}^{\otimes k/3} = T^{\otimes n}$ code with parameters $[[n=8(7^L),k=4(7^L)-1,d=2]]$ where $L=0,1,\dots$ denotes the number of layers. The logical $CCZ$s will act on the triplet of logical qubits on each tile.

A similar exercise can be repeated using $[[12,2,2]]$ codes \cite{Webster_2023} that have $\overline{CS}=T^{\otimes 12}$ as building blocks. Contracting the network as in the same figure while placing two logical qubits on each tile (red dot), we arrive at a class of $[[n=12(11^L),k=\frac {12}{5}(11^L-1),d=2]]$ codes that have $\overline{CS}^{\otimes k/2}=T^{\otimes n}$ as transversal gates where the logical $CS$ act on the pairs of logical qubits on each time.

A 1D tensor network like Fig.~\ref{fig:1qbit_code}a can also be built from a sequence of ($1-$)isometric contractions as long as we contract only one leg for the tensors on the boundaries and two for those in the bulk. Again, using $[[8,3,2]]$ codes as building blocks, let each tensor node have $3$ logical legs while the rest as physical legs. Then performing $X$-deformed Bell fusion produces a class of $[[n=6(L+1)+2,k=3(L+1),2]]$ codes
where $\overline{CCZ}^{\otimes L+1}=T^{\otimes n}$. Each logical $CCZ$ acts on the 3 logical qubits associated with each tensor tile and $L$ is the length of the graph measured in the path length connecting the left and right boundary nodes.

Repeating with $[[12,2,2]]$ blocks in a 1D tensor network where we have two logical and $12$ physical legs on each tensor node, one obtains $[[n=10(L+1)+2,k=2(L+1),2]]$ codes with $\overline{CS}^{\otimes L+1}=T^{\otimes n}$. The logical CS gates act on the two logical qubits on each tile.

If the atomic codes have stronger isometric properties, which requires them to have higher code distances, then other network architectures are also possible without inducing kernels on the resulting encoding map.

\subsubsection{Example: Codes with transversal interblock CCZ gates}
A similar exercise can be done to generate codes with (interblock) completely transversal $CCZ$ gates. A convenient set of lego blocks with completely transversal CCZ gate is the class of rectified cubic codes \cite{Vasmer_2019}, where the smallest example is a $[[12,1,2]]$ code\footnote{Note that many codes are said to have transversal $CCZ$, e.g., the $[[8,3,2]]$ or $[[2^D,D,2]]$ color codes, where logical $CCZ$s implementations include additional single qubit rotations. In this example, we use atomic codes where interblock $CCZ$s are given strictly by the tensor product of physical $CCZ$s.}. 
Then it is easy to generate other higher rate error detection codes using such atomic legos such that the strongly transversal $\overline{CCZ}^{\otimes k}=CCZ^{\otimes n}$ remains a logical operator. For example, a tree tensor network generates a $[[132, 13, 2]]$ code with transversal $CCZ$ gate (Figure~\ref{fig:CCZ_transv}). Since $CCZ^*=CCZ$, no special Clifford deformation is needed when gluing together these blocks.
\begin{figure}
    \centering
    \includegraphics[width=0.3\linewidth]{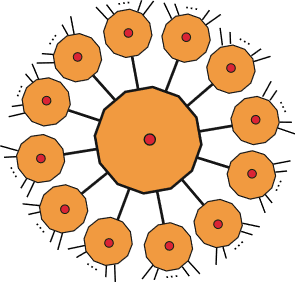}
    \caption{A code with completely transversal $CCZ$ gate built from rectified cubic code legos. Red dots represent logical qubits.}
    \label{fig:CCZ_transv}
\end{figure}

Connecting such tensors into a line with regular Bell fusions, as in Figure~\ref{fig:1qbit_code}a where each downward pointing leg contains 8 qubits, we obtain a class of $[[10L+2,L,2]]$ error detection codes with bitwise transversal $CCZ$ gates. Similar to codes with transversal $T$ gates, these are also of interest for magic state distillation. 

\subsubsection{Codes with asymptotic unit rate}

Given what we have just seen, one might wonder if it is possible to use combination rules such as Lemma~\ref{lemma:comptensorcode} to produce codes that have (asymptotically) constant overhead ($\gamma\rightarrow 0$) for magic state distillation. One way is to generate codes with asymptotically unit rate while preserving bitwise non-Clifford transversality. For $T$, using CSS codes on qubits, this is impossible: triorthogonal codes constitute the most general class of codes that realize logical $T$ with physical transversal $T$ gates \cite{Rengaswamy_2020}. However, such codes with $d_Z\geq 2$ must have $n\geq 2k$ \cite{Nezami_2022}, ruling out any asymptotic unit rate codes of this type. 

For other gates or higher dimensional qudits, this is less clear. For example, the tensor network in Figure~\ref{fig:1qbit_code}b clearly allows codes with such parameters as the number of logical qubits scale as the area of the 2 dimensional region whereas the physical qubits scale as its area plus the perimeter. Hence, when such kind of tensor networks are generalized to $D$-dimensional lattices, one might expect asymptotic unit rate codes to exist which can support completely transversal non-Clifford gates. 

\begin{proposition}\label{prop:MDSconstraint}
    Suppose there is a quantum maximum-distance-separable (MDS) $[[n,k,d]]_q$ qudit code that is non-degenerate, has $k>0$, and supports completely transversal $U$ and there exists single-site unitary $C$ such that $U\otimes U|\Phi_C\rangle=\omega |\Phi_C\rangle$ where $\omega\in \mathbb{C}$. Then a $d-1$ dimensional hypercubic lattice tensor network constructs an asymptotic unit rate code with constant distance $\geq d$ and completely transversal $U$.  
\end{proposition}

\begin{proof}

Given such an MDS code, the total number of legs is $n+k = 2k + 2(d-1)$; pick $k$ legs to be in the bulk at each vertex and $k$ physical legs at that site. The remaining $2(d-1)$ legs are contracted in the cardinal directions in the $D=d-1$ dimensional grid where by assumption we have $d\geq 2$ or $D\geq 1$. We assigning the deformed edge $\ket{\Phi_C}$ to each trace. 

As the code is non-degenerate, any $d-1$ physical qubits must be in a maximally mixed state, which implies that the tensor is an isometry from any $d-1$ physical legs to the remaining $n-d+1$ physical legs. Since the trace has been isometric, these logical legs correspond to independent qudit degrees of freedom. For stabilizer codes, we can also see this from the point of view of Pauli flows. The cleaning lemma guarantees that all incoming Pauli flows of weight $\leq d-1$ must be cleanable because they cannot be non-trivial logical operators. However, since stabilizers cannot have weight less than $d$, the incoming flows always push to non-trivial outgoing flows on the remaining physical legs. 

Now choose any logical leg to inject the initial logical flow, which flows to at most $D$ orthogonal directions. Without loss of generality, choose the coordinate such that the flow goes towards the positive direction relative to the tensor using the natural coordinate of the lattice. Then for each junction, we can clean the incoming flow with a stabilizer such that it always flows to the $D$ positive canonical directions. This will reach the boundary without ever returning to the starting point. 

Then using Lemma~\ref{lemma:isodist}, we see that the distance must not decrease under the isometric trace. By operator matching, we are guaranteed complete transversality of $U$ for a code where the total number of logical legs are given by the number of interior nodes $Vk$. The number of physical legs is given by $Vk+V_\partial$ where $V_\partial$ is the volume of the boundary. For a hypercubic lattice, we know that $V_\partial/V\xrightarrow {V\rightarrow \infty} 0$. Hence the rate $k/n\xrightarrow{n\rightarrow \infty} 1$.
\end{proof}

It is known that quantum maximum-distance-separable codes are already heavily constrained \cite{Huber_2020}. It is unclear whether any MDS codes with higher local dimension $q$ can support bitwise transversal non-Clifford gates, thus permitting asymptotically unit rate codes. 
We remark that the condition on non-degeneracy may not be necessary. It is still possible for a localized group of tensors to satisfy the isometric or operator pushing condition even though each individual tensors are not sufficiently isometric in the needed directions. For example, having 2-isometries in holographic codes are sufficient to make sure the global encoding tensor is isometric. However, in \cite{Steinberg:2024ack}, the holographic encoding map is still globally isometric even though each local tensor is only 1-isometric. This is because localized groups of tensors still satisfy a non-local isometric constraint.

\section{Applications with addressable transversal gates}

\label{sec:5}
\subsection{Holographic codes with QRM legos}\label{subsec:hqrm}
Holographic codes are first constructed as toy models \cite{Pastawski_2015} of the Anti-de Sitter/Conformal Field Theory (AdS/CFT) correspondence where these QECCs have a hyperbolic geometry when written in terms of tensor networks. In these codes, physical qudits live on the conformal boundary of the hyperbolic disk, analogous to the degrees of freedom associated with a boundary CFT. Logical qudits live in the bulk of the hyperbolic space of one higher dimension which emerges from the long range boundary entanglement. Different varieties of the code have been constructed, some with constant rate and above square-root distance scaling \cite{Harris_2018,Steinberg:2024ack}, while retaining high erasure and code capacity thresholds against (biased) Pauli errors that are competitve with leading topological or LDPC codes \cite{tnc,Fan:2024add}. A concerning drawback for holographic codes is the difficulty in implementing fault-tolerant non-Clifford gates as pointed out by \cite{Williamson}, but transversal non-Clifford gates are permitted in codes that satisfy complementary operator algebra quantum error correction but not complementary subsystem quantum error correction \cite{Harlow_2017,Pollack_2022,Cao_2021}. One such example is the holographic quantum Reed-Muller (QRM) code \cite{QL}. A recent breakthrough  also allows one to construct a universal gate set that is addressable and fault-tolerant \cite{Steinberg:2025kzb}. 
However, most of these multi-qubit addressable gates still come at an overhead that scales non-trivially with the length of the code.

We first summarize our main conclusion as follows.

\begin{theorem}\label{thm:holoQRM}
    There exists an explicit construction of a \emph{constant-rate} holographic QRM code of hyperbolic tiling built entirely with QRM tensors where any interblock and intrablock logical $CZ$s on acting on a pair of logical qubits can be implemented transversally if the location of their bulk tensors are related by a discrete rotation in the tensor network. Furthermore, the logical $CZ$s acting on disjoint pairs of logical qubits in the same layer have depth 1.
\end{theorem}

\begin{theorem}[Holographic black hole code]
\label{thm:holobhcode}
    Consider a QRM or Steane-QRM black hole code where all the logical qubits are localized on the horizon of an AdS black hole tensor network, intrablock and interblock logical $CZ$s on acting on any pair of logical qubits are transversal and fault-tolerant. Furthermore, logical $CZ$s acting on disjoint pairs of logical qubits have depth 1. 
\end{theorem}

\subsubsection{Addressable transversal $CZ$ in Holographic QRM codes}
The proofs are instructive in showing how operator pushing is used to identify transversal gates, hence we provide a sketch here. 

\begin{figure}
    \centering
    \includegraphics[width=0.8\linewidth]{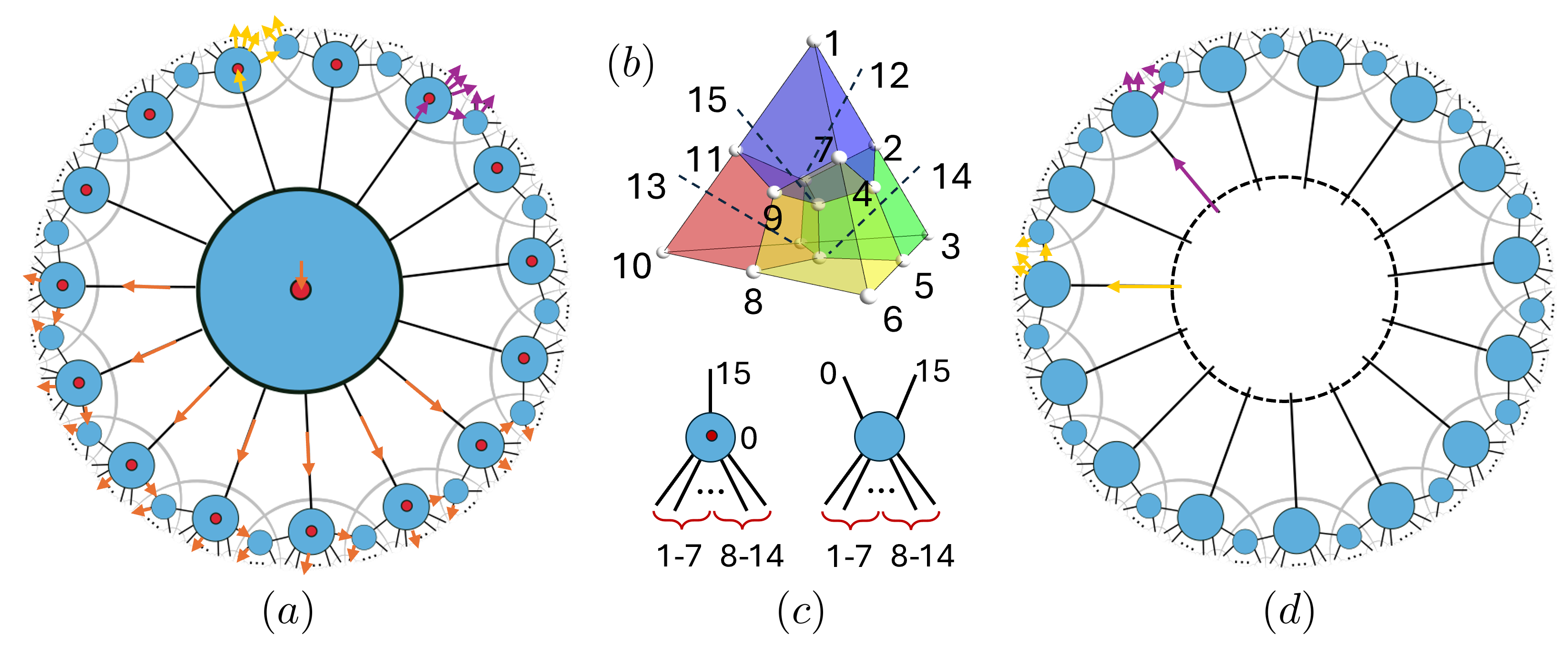}
    \caption{(a) A heterogeneous holographic Reed-Muller code that supports targeted $CZ$ gates. The support of an interblock $CZ$ acting on two bulk qubits in the center can be found following the orange flow while intrablock transversal $CZ$ on two bulk qubits in the first layer is supported on the purple and yellow flow. (b,c) Labelling of the tensors of each seed code. The logical leg of the encoding tensor is labelled 0. (c) left is $T_1$ while right is $T_2$. (d) A black hole QRM tensor network that encodes 15 logical qubits represented by the wires lying along the dashed circle in the middle of the bulk that represents the black hole horizon. The rest of the tensors have no bulk qubit, but consists only of $[[15,1,3]]$ seed codes and the shortened $[[14,2,2]]$ codes. }
    \label{fig:QRM_CZ}
\end{figure}

Consider a holographic code built entirely from $[[15,1,3]]$ QRM legos where the edges are traced without any Pauli insertion (Fig.~\ref{fig:QRM_CZ}a). We can build up such a code using edge inflation\cite{Jahn_2021,Steinberg:2025kzb}, where the tensor network is grown from a single $[[15,1,3]]$ seed in the center of the bulk. At layer 1, a seed with one leg facing inward is glued to the central tensor, while in layer 2, a $[[16,0,4]]$ tensor is added with two legs facing inwards that connects two tensor legs from two distinct branches that are adjacent to each other. At layer 3, each leg is again contracted with a seed with only one leg facing inward. This process is then repeated until some radial cut off. For convenience, we refer to all tensors with one leg facing inwards as $T_1$ and those with two facing in as $T_2$. 

We choose the orientation of the legs on the QRM lego as follows. For $T_1$ codes we use the centermost qubit (i.e. qubit 15 in Figure~\ref{fig:QRM_CZ}bc) as input and the logical qubit (qubit 0 in Fig.~\ref{fig:QRM_CZ}c) as bulk qubit. In $T_2$ tensors, we use the logical and centermost qubit legs as inputs and the rest as output. For conceptual simplicity, though this is unnecessary to ensure transversality, we also choose the placement of the output legs such that qubits $1-7$ in $T_1$ and $T_2$ are placed in the group on the left among output wires while the rest are placed the right (Figure~\ref{fig:QRM_CZ}c). 

From Lemma~\ref{lemma:isodist} and the fact that these seed tensors correspond to non-degenerate codes, we know that the max-rate holographic QRM code in Fig.~\ref{fig:1qbit_code}d has distance at least 3. The constant distance parameter is somewhat deceptive here as the distances of different bulk qubits are different. From AdS/CFT, we expect the qubits near the conformal boundary have lower distance whereas those closer to the center have higher distance. In a discrete model such as this one, to ensure that the bulk qubits in the interior have higher distance (and to ensure better parallelizability for the logical $CZ$ gates) we consider a finite but reduced-rate code where we do not add any bulk qubit on $T_2$s.  As a function between the bulk logical and the minimal radial distance $r$ (where $r$ is the number of layers) between it and the boundary, the logical qubit will have distance 
$d\geq 2^{r/2}$ because each inflation by $T_1$ doubles the distance in the worst case while inflation by $T_2$ does not add to the distance since it corresponds to concatenating some qubits with a $[[14,2,2]]$ code \cite{Steinberg:2025kzb,Cao_2021}. Suppose the radial coordinate of two logical qubit differ by $\delta r$, then their distance differ by at least $O(2^{\delta r/2})$ with the inner qubits having the higher distance. 
This holographic code has a slightly lower rate compared to the one in Fig.~\ref{fig:1qbit_code}d. It supports a transversal non-Clifford gate of $T$ type but is not bitwise transversal as $T$ and $T^{\dagger}$ need to alternate across adjacent seed codes thanks to the matching lemma. Pauli gates are clearly transversal and so are global bitwise CNOT gates. 

Recall from Fig.~\ref{fig:multicopy_tensor}c that the QRM lego admits localized symmetries of $CZ$ operators such that two copies of the 16-legged tensor is left invariant when $CZ$s act on 8 pairs of their legs. More precisely, $CZ^{\otimes 7}$ acting on any face of the tetrahedron implements a logical $CZ$ while $CZ^{\otimes 8}$ acting on the vertices of any colored cube leaves the code invariant. Using these facts, inter- and intra-block transversal gates can now be obtained via operator pushing (Fig.~\ref{fig:QRM_CZ}a). 

\paragraph{Transversal interblock gate.} In orange, an interblock $CZ$ between the central bulk qubits on two copies of the holographic code can be pushed to 7 other legs, which can then be cleaned off of leg 15 any $T_1$ tensors to another 7 legs that are facing outwards (e.g. to legs 4-9 and 14). As $CZ$ can always be cleaned off of the input legs of $T_1$ and $T_2$, we proceed similarly through subsequent layer, until one reaches the boundary. Then the interblock logical $CZ$ acts transversally on qubits where the orange flows are supported and identity elsewhere. The identical exercise holds for the other bulk qubits. It is clear that this operator pushing can always generate interblock $CZ$ that targets the two bulk qubits at the same location in the bulk between two copies of the holographic QRM code. But since different branches of this tensor network are actually identical, we see that logical $CZ$s between any two bulk qubits related by a discrete rotation symmetry (e.g. bulk qubits in layer 1) can also be pushed to the boundary as a transversal gate. 

\paragraph{Transversal intrablock gate.} The same argument also applies to intrablock operator pushing. Consider for instance a $CZ$ that acts on two bulk qubits in the same layer. A logical $CZ$ (yellow and purple arrow in Fig.~\ref{fig:QRM_CZ}a) can then be pushed to legs 1-7 on both QRM tensors, which then flow through leg 0 of adjacent $T_2$ tensors to legs 1-7 on the two $T_2$s. This pattern can be made identical across the two flows and  can be pushed to the boundary. Therefore, such intrablock $CZ$ gates are also transversal and are given by acting the tensor product of $CZ$s on where the yellow and purple flows are supported. Each $CZ$ straddles one yellow arrow with another purple arrow at the same position related by a translation on the boundary. 

As bulk qubits in different layers do not admit obvious identical subgraphs in the operator flow, it is most likely that this code does not permit transversal $CZ$ gates that connect bulk qubits in distinct layers. 
We leave a proof of this for future work on specific holographic codes. Nonetheless, each layer of the code constitutes a subblock within each code where any two bulk qubits can talk to each other transversally as long as there is a discrete rotation symmetry.

\paragraph{Black hole QRM code.} Similar to \cite{Pastawski_2015,Steinberg:2025kzb}, one can also consider a black hole code (Fig.~\ref{fig:QRM_CZ}d) where instead of having bulk qubits sitting at different radii from the center, we can cut out the central tensors of the holographic tensor network such that the bulk logical qubits all lie on the ``horizon'' of the black hole and not assign bulk qubits elsewhere. This amounts to changing the contraction of $T_1$ tensors slightly, where leg $0$ is used as the inward facing leg while the other fifteen legs are outward facing. Then the rest of the network is built up the same way as before using edge inflation. Then inter- and intra-block operator pushing follows similarly. Yellow and purple arrows in Fig.~\ref{fig:QRM_CZ}d  represent how intrablock $CZ$s acting on two bulk qubits on the horizon can be pushed to the boundary. In this case, any pair of bulk qubits are now equally addressable, and the same holds for interblock $CZ$ gates, since all logical qubits lying on the horizon are identical each other up to a discrete rotation. 

For both the reduced rate QRM and the black hole codes above, the $CZ$ gates are also parallelizable as long as they act on disjoint sets of logical qubits. These logical gates can be parallelized as long as the operator flow of distinct logical gates do not intersect. For $CZ$s acting on disjoint pairs of logical qubits in the same layer, by construction their outflows can be chosen to be identical and non-overlapping. A slightly more careful treatment is needed when multiple logical $CZ$s  act in parallel on bulk qubits in different layers in the reduced rate QRM code. With our current arrangement, the support of such $CZ$s do not overlap since operator flow from the in-plane legs and those from the logical legs have no overlapping support by construction (Fig.~\ref{fig:QRM_CZ}c). Explicitly, the former flow to legs $1-7$ while the latter to legs $8-14$ on $T_1$. Similarly, any $CZ$ flow through leg $0$ on $T_2$ will go through legs $1-7$ while that going into $15$ will come out out on legs $8-14$. It then follows that the resulting flows from two different layers are completely disjoint from each other. Therefore, $CZ$s from different layers in this code are completely parallelizable.

\paragraph{Fault tolerance.} While the interblock $CZ$ gates are clearly fault-tolerant as each code block can corrects single qubit errors, the intrablock fault-tolerance is more nuanced. However, thanks to the expansion property of the holographic code, they too can be fault-tolerant. Consider any single-qubit error on the boundary. An intrablock $CZ$ gate propagates it to two errors on the boundary. However, each of these errors is fed into a distinct QRM lego that is either an error correcting code $T_1$ or an error detecting code $T_2$. 

For the finite-rate QRM code, contraction with $T_1$ and $T_2$ are nothing but partial concatenation with $[[14,2,2]]$ codes. As such, these errors are detectable. However, correctable errors in this code do not always propagate to correctable errors, and the intrablock $CZ$ gates are not fault-tolerant.

For the black hole code, a single boundary error still propagates to two distinct errors on either two $T_1$ tensors or two $T_2$ tensors. Because $T_1$s here have no bulk qubits, their contractions correspond to the concatenation of $[[15,1,3]]$ codes. Hence errors here can be corrected with one round of syndrome measurement. If the errors propagated to $T_2$ blocks, then the $T_2$ blocks themselves can only detect these errors but not correct them. However, error detection allows us to locate these errors, and the worst case amounts to having two located errors from $T_2$ propagating into a single $T_1$ block in the inner layer. These located errors can be corrected on a $d=3$ code, hence a correctable error propagates to correctable errors and the gate is fault-tolerant.

\subsubsection{Addressable Transversal $CZ$ gate in Heterogeneous Steane-QRM code}
We now consider the heterogeneous Steane-QRM code constructed in \cite{Steinberg:2025kzb} which supports a universal set of fault-tolerant gates. Edge inflation is also used to produce a holographic code but with alternate layers of QRM and Steane codes. Ref. \cite{Steinberg:2025kzb} shows that these codes support addressable fault-tolerant but non-transversal gate sets consisting of $H,T$ and $CNOT$. Here we show that transversal $CZ$s, instead of fault-tolerant $CNOT$s, can be used to complete the universal set using operator flow we discussed in earlier sections. For the sake of concreteness, consider the construction of QRM-Steane iterations from the center to the boundary where Steane codes are only used for $T_1$ tensors while QRM are used for both $T_1$ and $T_2$ tensors. Both the zero rate code which encodes a single qubit in the bulk  (Fig. 1c of \cite{Steinberg:2025kzb}) and the black hole code (Fig. 6c of \cite{Steinberg:2025kzb}) which encodes multiple qubits living on the horizon can be constructed. 

\begin{figure}
    \centering
    \includegraphics[width=0.9\linewidth]{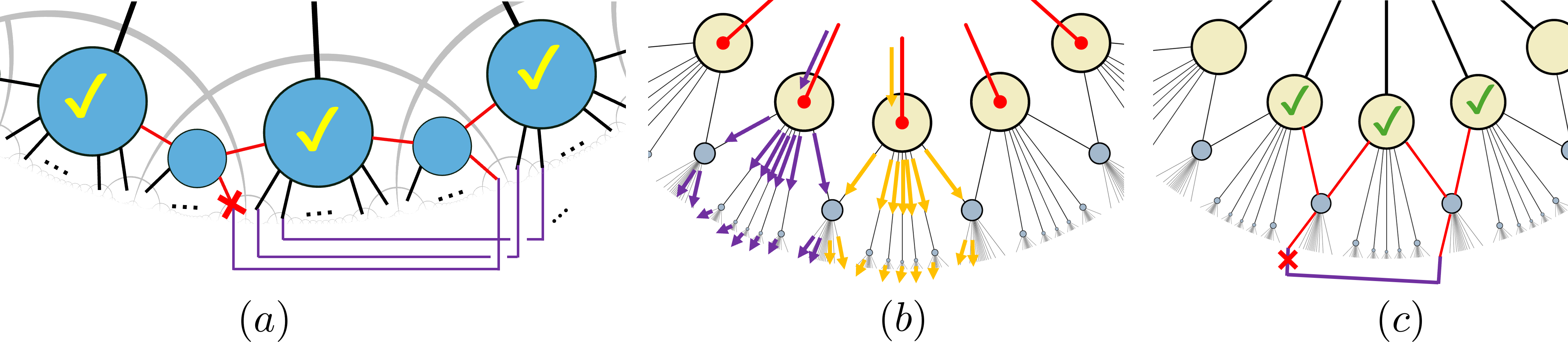}
    \caption{(a) In a black hole QRM code, a boundary error can propagate to at most 4 located errors (red) in the inner layer, which are then correctable. (b) Intrablock $CZ$ gates can be pushed from the bulk qubits (red) to the boundary following the operator flows. Here the beige tensors are Steane codes and blue tensors are QRM codes. Each purple and yellow arrow marks one of the qubits on which physical $CZ$ gates acts. (c) Intrablock gates are fault-tolerant in the Steane-QRM code where error propagates along red lines after acting physical $CZ$ (purple), but can be corrected at the Steane layer. }
    \label{fig:QRM_FT}
\end{figure}

For the zero rate code, it suffices to consider interblock transversal gates. The key difference between the holographic QRM code above and the heterogeneous QRM-Steane code is that $CZ$ is completely transversal on Steane tensors. However, this does not pose any obstruction to operator pushing because QRM layers can push any $CZ$ operators flowing through its input in both $T_1$ and $T_2$ tensors (Fig.~\ref{fig:QRM_CZ}c). Although $CNOT$ is also interblock transversal in this code, it is bitwise transversal where one has to perform CNOTs on all physical qubits. Logical $CZ$ on the other hand only has support on a small subregion where the operator flow has support. With 3 layers in the tensor network consisting of QRM, Steane, and QRM as we move from the center outwards, a logical $CZ$ only requires $49$ physical $CZ$s whereas the original transversal $CNOT$ in \cite{Steinberg:2025kzb} is $\overline{CNOT}=CNOT^{\otimes n}$ with $n=1335$. This reduction is even more significant as we increase the number of layers, where the overhead saving is exponential in the number of layers even though both types of gates are transversal. 

For the black hole code, we have both inter and intrablock $CZ$ gates. For the former, operator flow proceeds similarly as the zero rate code, because QRM codes are 2-cleanable with respect to $CZ$ gates. The bulk qubits on the horizon are again related to each other by a discrete rotation, hence interblock transversal $CZ$ gates are permitted between any two bulk qubits, same as the black hole QRM code. Intrablock $CZ$ gates are also transversal. This is again because QRM are 2-cleanable and thus any $CZ$ flows on the input legs can be pushed to the boundary without obstruction (Fig.~\ref{fig:QRM_FT}b). For the same reason, $CZ$ gates on disjoint pairs of logical qubits are parallelizable as in the reduced rate and black hole QRM code.

These transversal $CZ$ gates are also fault-tolerant in the zero rate and black hole codes. Interblock transversal gates are clearly fault-tolerant already. Intrablock gates in the black hole code are fault-tolerant in the same way as the black hole QRM code. Since any boundary error can propagate to two errors from the physical $CZ$s, the worst case is when it occurs on a $T_2$ block where errors can only be detected but not corrected. However, the error detection on $T_2$s herald the location of these errors and thus can be subsequently corrected at the inner Steane layer where each Steane tensor has distance 3 and corrects 2 located errors (Fig.~\ref{fig:QRM_FT}c). 

 Although a comprehensive analysis of the Steane-QRM fault-tolerant overhead has not yet been done in \cite{Steinberg:2025kzb}, it is clearly better to use transversal $CZ$s as the two-qubit gates instead. For black hole  codes of $k>1$, a naive implementation of addressable $CNOT$s for certain logical qubits can take a depth-20 circuit~\cite{Steinberg:2025kzb} even with $r=2$ layers where we do not include the central tensor as a layer. Generically, because of the global $CNOT$ symmetry on both the Steane and QRM code, addressable $CNOT$s are not transversal and the depth grows with the number of layers where the worst case scaling can result in a depth $O(10r)$ gate, or $O(\log(n))$ where $n$ is the length of the code. Furthermore, special arrangements on the gates are often needed to ensure fault tolerance.

\subsection{Generalized Holographic Code}
\label{subsec:genholo}

\subsubsection{Construction and transversal gates}
We now use the insight above to construct some examples of holographic codes that admit transversal addressable diagonal gates like $T$ and $C^\ell Z$. 
A large class of finite rate and black hole holographic codes can be build to support addressable inter/intrablock transversal diagonal gates using the $b>2$-branch codes we introduced in Definition~\ref{def:branchcode}. They are structurally similar to the holographic QRM and the heterogeneous QRM-Steane code \cite{Steinberg:2025kzb} in the previous section.

We first notice that transversal addressable gate in these codes arise from operator flows that start in any point in the bulk and end up on the boundary without crossing. In the above example, this rests on having localized transversal $CZ$ gates in a single atomic lego. Here we will use the $(b>2, \mathcal{U})$-branch codes we introduced earlier as seed codes which are built by concatenating a $b$-qubit bit-flip repetition code with $b$ lego blocks that have completely transversal $T$ or $C^\ell Z$ gates. For instance, concatenation by the $[[15,1,3]]$ code will produce 2-cleanable branch codes with $\mathcal{U}=\{T,CS,CCZ\}$.

Conjoining these $(b,\mathcal{U})$-branch codes with $d\geq 2$ in a hyperbolic geometry, we can create a holographic code using edge inflation. Again, we orient these seed codes as $T_1$ and $T_2$ tensors that have 1 or 2 legs respectively facing towards the center of AdS when contracted. Each leg of the central tensor is contracted with another $T_1$ atomic block at the first layer with the remaining legs facing outward. Then all outward facing legs are contracted with a $T_1$ block except the two edges that come out of two different atomic blocks --- they are contracted with a $T_2$ tensor in the next layer where the two contracted legs belong to two distinct branches (e.g. red and blue legs feeding into the tensors on the outermost layer in Figure~\ref{fig:transversal_holo}). The rest is identical to the holographic codes in Sec.~\ref{subsec:hqrm}. By construction, we see that operator pushing from any bulk qubit can flow to the boundary without obstruction because each seed code is 2-cleanable by Lemma~\ref{lemma:branchcode} and each $T_1$ or $T_2$ only need at most 2-cleanability to push the operator flows from their ingoing to outgoing legs. For simplicity, we also arrange the ordering of all legs across all $T_1$ (or $T_2$) tensors to be identical in the same layer. 

\begin{figure}
    \centering
    \includegraphics[width=0.9\linewidth]{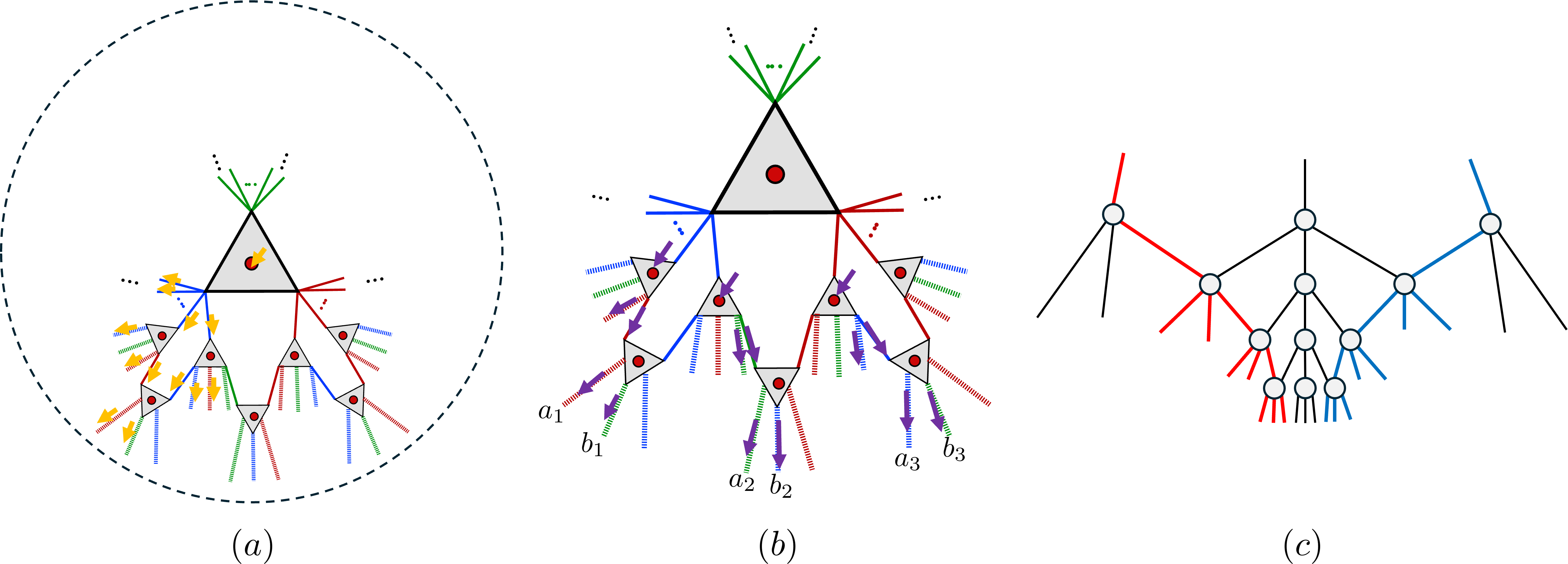}
    \caption{A holographic code created from the tri-color atomic blocks in a tensor network representation. Each colored dashed line represents more than one qubit whereas solid line represents one qubit degree of freedom, a red dot represents a logical qubit. (a) A interblock transversal gate represented as operator flow. (b) An intrablock 3-qubit operator acting on 3 qubits in the same layer of the bulk can be pushed. (c) Two operator flows supported on red and blue edges never merge due to expansion.}
    \label{fig:transversal_holo}
\end{figure}

\begin{remark}
    The holographic codes constructed above have finite rate.  This is because each tensor is cleanable on the legs that are contracted by Lemma~\ref{lemma:branchcode}. Using Theorem~\ref{thm:tnc}, we know that the number of logical qubits add under each tensor contraction. Since the number of boundary (physical) qubits scale the same way as volume in a hyperbolic tiling, the rate is finite for any fixed $b$.
\end{remark}

\begin{theorem}\label{prop:5.3}
    A holographic code constructed using edge inflation and $\{b>2,\mathcal U\}$-branch seed codes of $d>1$ support transversal addressable logical $U\in \mathcal{U}$ gate on any single logical qubit if $U$ is a single-qubit phase gate $P$ such that $P^N=I$. 
    If $U$ is a $C^\ell P$ interblock, intrablock, or mesoblock gate acting bulk qubits on the same type of tensor nodes (i.e. $T_1$ or $T_2$) in the same layer, then $U$ is transversal. If multiple $C^\ell P$s act on disjoint subset of logical qubits on the same layer and type, then they are also parallelizable.
\end{theorem}

\begin{proof}

We convert the tensor network into a graph where we assign a bulk vertex from the set $V_{B}$ to each logical leg and a boundary vertex $V_{\partial}$ to each physical leg. We denote the internal vertices that come from tensor nodes as $V_I$. The dangling edges then become graph edges that connect elements of $V_I$ with those of $V_B$ or $V_\partial$. Let such a graph be $G=(V_B\cup V_I\cup V_\partial,E)$. Operator flows that traverse through $G$ originate in $V_B$ and terminate in $V_\partial$. To show the (addressable) gates are transversal, it suffices to demonstrate that the logical operators can be pushed to the boundary without inconsistencies. A flow is inconsistent if two different operators can act on the same point in the tensor network.

We focus on the phase $P$ or $C^\ell P$ flows as the Pauli flows are unremarkable. For each $P$ flow that originates in any $v_s\in V_B$, we can push it from the logical node $v_s$ to edges $e=(v_s,v)$ where $v\in V_I$ is the tensor node connected to the bulk qubit $v_s$. This flow then goes through the neighbours of $v$ in the $i$th branch of the seed code --- the logical flow goes through all legs belonging to a single branch (i.e. same color) as a $P$ flow. Because $b>2$, we can always choose this branch to be the branch that faces towards the UV direction (or boundary) of the network. Each of the edges is then fed into another neighouring node $v'$. If $v'\in V_\partial$, then the flow line terminates. If $v'\in V_I$, then we change $P$ to $P^{\dagger}$ in the flow. 

There are two types of internal nodes, ones belonging to $T_1$, and hence receiving a single $P^a$ flow on a branch $i'$ for $a=1,\dots,N$, or ones that belong to $T_2$, which receives two flows $P^a, P^b$ on two of its branches $i',j'$. Both types of flows are cleanable by Lemma~\ref{lemma:branchcode}. For the first type, all outgoing legs point towards the boundary. Hence this $P^a$ flows to its downstream neighbours towards the boundary on the remaining legs on branch $i'$ as well as as a $P^{N-a}$ flow on another branch $k'\ne i'$. Note that since $b>3$, there is at least one more branch on which there is no flow. If a logical $P$ gate acts on the input node $v_s'$ adjacent to $v'$, then it can be pushed to the branch that has no flow, though this is not required by the consistency condition. 

For the second type, we have two incoming $P^a, P^b$ flows into the tensor for some integer $a, b=1,\dots, N$. Therefore, the remaining legs in branches $i',j'$, which face outward, now have $P^a,P^b$ flowing through them respectively. In addition, the total net flow at this tensor junction $v'$ has to be conserved so that there is an outflow of $P^c$ on branch $k'\ne i',j'$ for some integer $c$ such that $a+b+c=g$ if a logical $P^g$ operator acts on the logical node $v_s'$. As such, for a third branch with flow $T^{(a+b-g)}$ is needed. The same pattern above recurs until all flow lines terminate on vertices in $V_\partial$.
We first note that because flows always reach the sink, all logical $P$ gates are addressable and can be represented as the type of operator (tensor product of $P^a$ where generally a different value of $a=1,\dots,N$ is needed for each site) on the boundary.

Apply the above layer to layer until one reaches the boundary, we thus produce a valid physical representation of the logical $P$ operator via operator pushing.

The $C^\ell P$ gates act transversally on bulk qubits where operator can be pushed to the boundary following a flow of the same pattern; except the flow acts on $\ell$ blocks all at once. Let the support of the operator flow for phase operator $P$ from sources $S=\{v_s\}$ on a single code block to the boundary sinks $K\subset V_\partial$ is given by the minimal induced subgraph $F_{S}$ on which the flow traverses. Two flows are isomorphic if the induced subgraphs $F_S\cong F_{S'}$ with the constraint that the isomorphism maps $S$ to $S'$ and $K$ to $K'$. A logical $C^\ell P$ can flow from the sources to the sink if one can identify $\ell$ non-intersecting isomorphic flows. 
For interblock gates, the $\ell$ isomorphic flows will come from distinct code blocks, i.e., disconnected copies of the tensor networks. By construction, such isomorphic flows exist for all flows originate from sources with the same tensor type (i.e. $T_1$ or $T_2$) within the same layer. These flows clearly do not intersect as they act on different copies of the tensor network.

Intrablock $C^\ell P$ gates requires one to identify $\ell$ isomorphic flows within a single code block. We need to also demonstrate that these induced subgraphs can be identified on a single tensor network without overlap. 

Consider operator flows originate from the tensors of the same type in the same layer. 
For $b>2$, we notice that outgoing operator flows separated by at least one leg/branch never cross (Fig~\ref{fig:transversal_holo}c). 
Regardless of whether the sources are on $T_1$ or $T_2$, there is always at least one branch where all its legs are pointing toward the boundary. Because the legs and branches of the tensors are ordered identically in the same layer (Fig.~\ref{fig:transversal_holo}b), the outgoing flows will always be separated by at least one other branch on which there is no flow. Hence these flows will not cross downstream.

As long as there are no fewer than $\ell$ such sources on a single layer, we can then generate the required isomorphic flows $\{F_s^{\alpha}, \alpha =1,\dots, \ell\}$.  Let the sink $K_\alpha$ of each $F_s^{\alpha}$ to have size $|K|$, then the intrablock logical operator is given by
$\overline{C^\ell P}= \bigotimes_{i=1}^{|K|} C^\ell P^{a_i}$ where each $C^\ell P^{a_i}$ gate acts on exactly one qubit from from $K_\alpha$ with $a_i=1,\dots,N$. 

For the mesoblock $C^\ell P$ gates, we allow some isomorphic flows $F_s^{\alpha}$ to act on the same block while others across different blocks. This follows trivially from the analysis above. 

Finally, for parallelizability --- $P$s can be inserted on a mix of $T_1$ or $T_2$ blocks in the same layer. For such logical operators, the outgoing flow can indeed merge in some tensors. However, for $P$ gate does not not impact parallelizability because their physical implementation consists only of powers of $P$ gate, which are diagonal and hence commute with each other. If the operator flows overlap, we simply add the power of $P$s in the outgoing flows to implement logical T gates inserted at these locations at the same time.

For $C^\ell P$ gates, above analysis indicates that $F_s^{\alpha}$ from any two distinct bulk qubits on the same layer already do not intersect. Therefore, the operator flows generated by multiple $C^\ell P$s acting on disjoint subsets of bulk qubits on the same layer (and same tensor type) do not intersect either, and such gates are indeed parallelizable.

\end{proof}

Similar to Sec~\ref{subsec:hqrm}, we can construct a black hole code where we first define a holographic state with a rank-$m$ central tensor. Each physical leg of this tensor is contracted with $T_1$ tensors on each leg. Each $T_1$ tensor now is precisely a branch code where the logical leg faces towards the center. Then for outgoing legs, we again contract a $T_1$ tensor except two legs coming from adjacent tensor nodes, which we contract with a $T_2$ tensor. Each $T_2$ tensor has two legs facing towards the center, one of which is the logical leg of the branch code, while the other a physical leg from one of the branch seeds. Note that neither $T_1$ nor $T_2$ here introduces any additional bulk qubits. Finally, for each layer with $T_2$, the next layer is contracted with only $T_1$ tensors. This repeats until the boundary cut off. We then excise the central tensor of the tensor network, which creates a black hole-like code with $k=m$.

\begin{corollary}[$b$-branch black hole codes]
\label{coro:5.1}
    The black hole code constructed above with $\{b>2,\mathcal U\}$-branch codes with $d>1$ seed codes support addressable logical $U$ gate on any single qubit logical qubit if $U$ is a single-qubit phase gate. Furthermore, if $U=C^\ell P$, then it is interblock, mesoblock, and intrablock transversal on any $\ell$ of the bulk qubits. If multiple $C^\ell P$s act on disjoint subset of logical qubits, then they are also parallelizable.
\end{corollary}
\begin{proof}
    The proof is similar to that of Theorem~\ref{prop:5.3} except the $T_1,T_2$ tensors now no longer have bulk inputs. As such, the support of different isomorphic operator flow is even harder to intersect compared to those in Theorem~\ref{prop:5.3}. By construction, all bulk qubits are on the same layer with the same type of tile ($T_1$), hence the rest follows from the above proof. 
\end{proof}

\subsubsection{Decoder and Fault tolerance}
All holographic codes admit tensor network decoders \cite{tnc,QL2} which are maximum likelihood decoders that are exactly contractible in $poly(n)$ time. For large number of bulk qubits, then it requires more efficient schemes such as parallel decoders that compute the marginal probabilities\cite{QL2,Farrelly:2020xpp}. Such decoders will only be $O(k \mathrm{poly}(n))$ and can be parallelized to further reduce the run time in $k$. However, FT decoders which account for circuit and measurement noise are very much an open problem for holographic codes which we leave for future work.

The fault tolerance of the gates discussed above is simple for single qubit gates and multi-qubit interblock transversal gates as they are clearly FT. The intrablock gates are more complicated as it depends on the connectivity of the tensor network and cleanability of each seed tensor. FT of the mesoblock gates are somewhere in between, but since its flows on a single block cannot spread more than the intrablock gates, it is sufficient to prove intrablock fault tolerance from which mesoblock fault tolerance follows.

\begin{theorem}\label{thm:finite_rate_holo}
 If the $b$ tensors joined together by the repetition code in constructing the $\{b>2, \mathcal{U}\}$ branch code can each be written as a $[[n_0,2,d\geq 3]]$ code where $C^\ell P\in \mathcal{U}$, then a finite rate holographic code constructed using the $b$-branch codes with a bulk qubit on each tensor will admit fault-tolerant intrablock, interblock and mesoblock transversal $C^\ell P$ gates.
\end{theorem}

\begin{proof}
Recall that the holographic code of interest has the form in Fig.~\ref{fig:transversal_holo} and that each $b$-branch code is constructed with the form shown in Fig.~\ref{fig:tritensor}. We will refer to the tensors joined together by the repetition code as the branch seed (orange tensors in Fig.~\ref{fig:tritensor}). In a holographic code consisting of the $b$-branch codes (grey triangles in Figures), 
we again only have $T_1$ and $T_2$ types. For each such tensor, at most 1 leg from each branch seed faces inwards, i.e., towards the center of the hyperbolic disk by construction. 

    For any single physical error on the boundary, a logical intrablock $C^\ell P$ gate propagates it to at most $\ell$ errors. Since the operator flows are constructed such that they are isomorphic, these errors occur on the same type of tensors. When one such error occurs on either a $T_1$ or $T_2$ tensor, it corresponds to a single error acting on a branch seed code $[[n_0,2,d\geq 3]]$ that makes up the branch code. Since each branch seed has only up to 1 leg pointing inwards, as long as we orient the branch seed tensors such that one of its logical legs is connected to the repetition code and the other points towards the center, then any single error on this seed will be a correctable error as the code has distance $3$. These errors can then be corrected with one round of error correction.
\end{proof}

\begin{theorem}\label{thm:bh_holo}
    The intrablock, mesoblock, and interblock transversal $C^\ell P$-qubit gate on any $\ell$ logical qubits are FT on the generalized black hole codes built from $\{b>2, \mathcal{U}\}$-branch codes with branch seeds $[[n_0,1,d\geq 3]]$ and $C^\ell P\in \mathcal U$.
\end{theorem}

\begin{proof}
    In the black hole code, there are no bulk qubits on any other tiles other than those next to the horizon. As a result, this corresponds to a concatenated code where all $T_1$ tensors have distance $\geq 3$. Like the previous theorem, a single physical error can be propagated to at most $\ell$ physical errors on tensors of identical types, i.e. $T_1$ or $T_2$. For $T_1$ tensors, the errors correspond to physical errors on any of the $[[n_0,1,d\geq 3]]$ branch codes, and thus can be corrected with one round of syndrome measurement. 

    For errors occuring on $T_2$ tensors, one of the branch codes must have a leg that is pointing towards the interior of the network. Errors on this branch can propagate to 2 logical errors in the next layer. However, by construction, the parent connecting to a $T_2$ tensor must be two adjacent $T_1$ tensors, hence these logical errors each propagate to a single physical error on a distinct code block with $d\geq 3$, which can then be corrected. 
\end{proof}

\subsubsection{Example: Diluted Holographic Reed-Muller code with tricolor tensors}
The holographic QRM code we examined in Sec.~\ref{subsec:hqrm} supports addressable $CZ$ gates on bulk qubits on the same layer. However, most of its other interesting logical gates like $T, CCZ$ are global. It is easy to make them addressable also, if we replace each seed tensor with a branch code defined above.

    Consider the triangle tensor in Figure~\ref{fig:tritensor}, where we concatenate a 3-qubit bit-flip repetition code with with three $15$-qubit QRM codes to produce a $[[45,1,3]]$ code. By Lemma~\ref{lemma:branchcode}, it is 2-cleanable with respect to $CS, CCZ$, and any power of $T$ on input legs that belong to two distinct branches. From Theorem~\ref{prop:5.3}, these gates are indeed addressable and transversal. This diluted holographic QRM code now supports transversal $T$ on any single qubit. In addition, it supports transversal addressable interblock, (mesoblock) and intrablock $CS$/$CCZ$ gates on any pairs/triplets of bulk qubits on tensors of the same type on the same layer.

While the interblock $C^\ell Z$ gate and $T$ gates are all fault-tolerant, the intrablock gates are not because each $[[15,1,3]]$ code is only dualizable to a $[[14,2,2]]$ code. However, we can construct a branch seed that does satisfy such a condition by performing a single trace on the physical legs of two $[[15,1,3]]$ codes. By Lemma~\ref{lemma:isodist}, this is a $[[28,2,3]]$ code. Using it as the branch code with the right orientation, then Theorem~\ref{thm:finite_rate_holo} ensures fault tolerance. Although $T$ is no longer completely transversal on this branch seed tensor, 2-cleanability of $T$ gate in the tricolor branch code they construct remains valid. 

Following Corollary \ref{coro:5.1} and Theorem~\ref{thm:bh_holo}, we also produce a 
 diluted QRM black hole code $[[n,k,d]]$ code where $k$ is given and can be adjusted by the number of horizon qubits, $n$ the number of qubits/legs on the conformal boundary, and $d\sim \exp(O(r))$ where $r$ is the radial distance between the black hole horizon and the conformal boundary. In this case, the $[[45,1,3]]$ code is sufficient to ensure fault tolerance of the all $T, CS$ and $CCZ$ gates that are inter/intra/mesoblock addressable on any (subset of) bulk qubits.

\subsection{Iterated Fractal Code}
\subsubsection{Construction}

Building on these simple lego gadgets that support addressable multi-qubit gates through GHZ tensors, we can also grow a code iteratively via a different self-similar pattern.

By Lemma ~\ref{lemma:comptensorcode}, one can connect $b$ codes to increase the number of encoded qubits without sacrificing distance. However, we still need to increase distance. An obvious way is through code concatenation. However, to build up for the iteration step,  we also want our code to remain 1-cleanable with respect to some unitary $C^\ell P$ after concatenation. This can be done by concatenating with a seed code that has this property. 

\begin{lemma}\label{lemma:concatT}
    Let $U$ be a unitary where $U^N=I$, $T'$ be  a $[[m,1,d']]$ code that is 1-cleanable with respect to $U$ 
    and let $T_b$ be an $[[n,k,d]]$ code that supports interblock, mesoblock, and intrablock transversal $U$ addressable on any subset of logical qubits whose physical representation consists of the tensor product of $U$ and $I$ on the physical qubits. Then concatenating $T_b$ by $T'$ on all physical qubits produces a $[[nm,k,dd']]$ code that preserves the transversality, addressability, parallelizability, and 1-cleanability of the $U$ gates. If $d'\geq3$, then these transversal gates are also fault tolerant.
\end{lemma}

\begin{proof}
    The code parameters follow trivially from concatenation. The transversality and parallelizability of the $U$ gates follow from operator pushing. Finally, the cleanability is inherited from the $T'$ tensors ---  any logical operator of $T_b$ acts on its physical legs as $U$, each of which can be pushed through the logical leg of $T'$. By 1-cleanability of $T'$, each of it can be pushed to all but one physical legs of $T'$. That is, the encoding tensor of the concatenated code is $A$-cleanable where $A$ contains all logical legs and any one physical leg. 
    
    Since the gates are transversal, interblock gates spread 1 error to distinct code blocks whose errors can be corrected by the inner code with $d'\geq 3$. The original intrablock $U$ gates on $T_b$ get pushed to tensor product of $U$ gates acting on distinct inner $T'$ code blocks. Suppose $U$ acts on $\ell$ qubits, a single physical error spreads to $\ell$ distinct $T'$ blocks and are correctable. Mesoblock logicals follow from both of the above. Hence addressable $U$ gates remain fault tolerant.
\end{proof}

\begin{proposition}[Iterated Fractal Codes]
Let $U$ be a $C^\ell P$ gate with $\ell\geq 0$ such that $P^M=I$ for some integer $M$. Consider the following iterative construction:     \begin{enumerate}
    \item Start with a base code that implements transversal $U$ gates as tensor product of $U^a, a=1,\dots, M$ and $U$ is 1-cleanable
    \item Fuse together $f$ such codes using a $f$-qubit GHZ state based on Lemma~\ref{lemma:comptensorcode}.
    \item Concatenate each physical qubit by $[[m,1,d']]$ codes $T'$ in Lemma~\ref{lemma:concatT} where $U$ is 1-cleanable on $T'$.
    \item Iterate $r$ times.
\end{enumerate}

We then obtain a class of codes that support transversal $U$ gates that have asymptotic parameter $[[N, K=O(N^{\alpha}),D=O(N^{\beta})]]$ where $\alpha=\frac{\log f}{\log f+\log m}$, $\beta=\frac{\log d'}{\log f +\log m}$. Furthermore, if $\ell>0$, then $U$ is interblock transversal on any $\ell$ logical qubits from $\ell$ distinct blocks. It is also mesoblock (if $\ell \geq 3$) and intrablock transversal on any $\ell$-tuple of logical qubits. Furthermore, these gates are parallelizable if they act on disjoint subsets of the logical qubits. Finally, $U$ is interblock fault-tolerant if the distance of the base code has $d\geq 3$ and also meso/intrablock fault-tolerant if its shortening by choosing a physical leg as logical leg also has $d\geq 3$.
\end{proposition}
\begin{proof}
    Without loss of generality, let the base code be $[[n_0,1,d_0]]$, then branching constructs a $[[(n_0-1)f,f,d_0]]$ code. By Lemma~\ref{lemma:comptensorcode}, transversality is preserved. Then concatenation produces $[[n_1=(n_0-1)fm,k_1=f,d_1=d_0d']]$ code. Then by Lemma~\ref{lemma:concatT}, we maintain transversality of the gate $U$. Iterating $r$ times, we have 
    $n_r= mf(n_{r-1}-1),~k_r=fk_{r-1}=f^{r},~d_r=d'd_{r-1}=d_0(d')^r$. Writing $n_r$ through the iterative relation, we have $$n_r=(mf)^rn_0-\sum_{i=1}^r (mf)^i <n_0 (mf)^r.$$
Rewriting the exponent for $K=k_r, D=d_r$, we have $[[N,K=O(N^{\alpha}), D=O(N^{\beta})]]$.

   Finally, by Lemma~\ref{lemma:comptensorcode} and Lemma~\ref{lemma:concatT}, the distance of the seed code ensures that the $U$ gates are fault-tolerant under repeated concatenation of tensor contractions. 
\end{proof}

Note that for large branching ratio $f\gg1$ and fixed seeds with parameter $[[m,1,d']]$, $KD = O(N^{1-\epsilon})$ where \begin{equation}
    \epsilon = \frac{\log m-\log d'}{\log(mf)}.
\end{equation}

Tensor network decoders based on weight enumerators apply for this code, where we can obtain the marginal error probability by computing the mixed coset enumerator (Sec. III G.2 of \cite{QL2}). The enumerator, and by extension, the exact error probabilities are obtained through a tensor network contraction, which is efficient as the entire code is described by a tree tensor network where contraction will take only linear time in $n$. To compute all the marginals, we need to repeat this $k$ times, and the complexity is given by $O(nk)$ but the $k$ marginals can be parallelized. 

\subsubsection{Example: fully addressable $CZ$ with partially addressable $T$, $CS$, and $CCZ$}
    For example, using the $[[15,1,3]]$ code and GHZ states as seed tensors, we have $m=15,d'=3$. Taking $b=3$ as in Figure~\ref{fig:fractal_code}, we have $\alpha =\beta=0.29$. 

\begin{figure}
    \centering
    \includegraphics[width=0.7\linewidth]{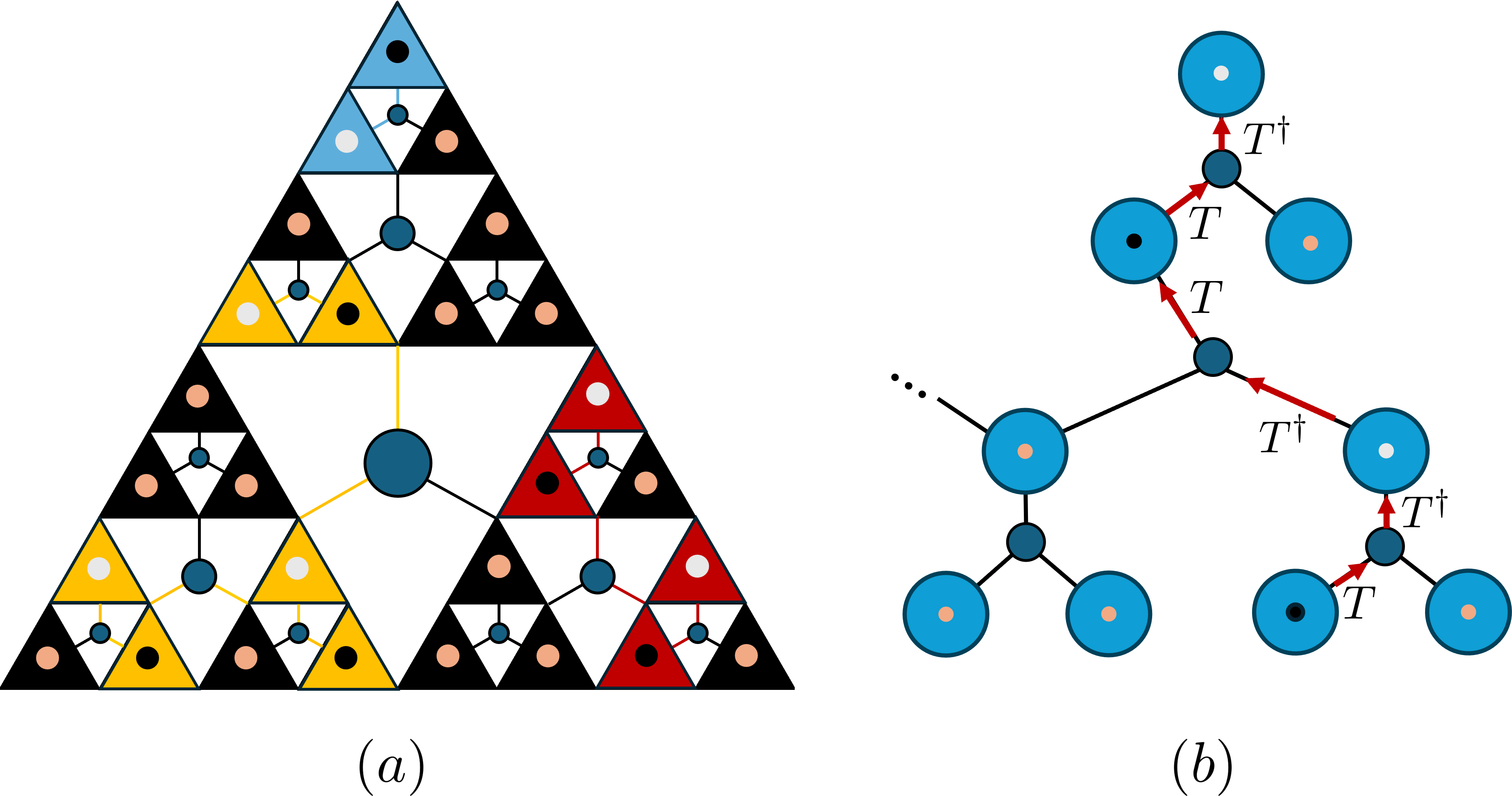}
    \caption{(a) A $b=3$ code produced with 3 iterations by conjoining GHZ states and 15-qubit QRM codes. Each black triangle represents the unit tensor at the base of the iteration where each contains a logical qubit. The dangling legs representing physical degrees of freedom and tensors from subsequent concatenations are suppressed for clarity. Logical $T$ and $T^{\dagger}$ are marked as black and white dots respectively where corresponding physical $T$ and $T^{\dagger}$ act on the suppressed legs associated with that triangle. Three different partially addressable $T$ gates are given each marked with a different color (blue, yellow, red) with their corresponding internal $T$ flows marked on the same colored edges. (b) Operator flow that produces the logical $TT^{\dagger}TT^{\dagger}$ gate on the red triangles. Physical legs and $T$s or $T^{\dagger}$s acting on them are suppressed.}
    \label{fig:fractal_code}
\end{figure}

Graphically, we can represent this codes as a fractal (a Sierpinski triangle at $b=3$) where for some finite cut off, we assign the smallest unit tensor to each triangle with edge length $\delta$, then a GHZ tensor to each white space. These tensors are glued together with one leg from each unit tensor. Then three such unit tensors are again joined together after concatenation, again glued to a GHZ tensor on a white triangle of length $2\delta$ and so forth. A logical leg is assigned to each unit triangle and as we have shown above that transversality is preserved under the iterative process. 

Since we have chosen the seed codes to be at least 1-cleanable with respect to $CZ$s, the resulting fractal code supports inter/intrablock transversal $CZ$ on any two logical qubits. However, since $T$ and $CCZ$ are bitwise transversal on the seed and are not 1-cleanable, these gates are not fully addressable as a result. However, we retain partially addressablility of logical $T$, interblock $CS$, and interblock $CCZ$ gates. This is because each unit tensor has transversal $T$, but a GHZ tensor is stabilized by $T\otimes T^{\dagger}$ (by $CS\otimes CS^\dagger$ and $CCZ\otimes CCZ$ acting on two and three blocks respectively). Hence the transversal $T$ (or $CS$ and $CCZ$) will always flow from one branch to another connected by the GHZ tensor in the middle. Depending on which branch is connected to the GHZ tensor on the coarser scale, this flow can activate $2\ell$ other logical qubits. 

Since the base seed code we use have distance 3, the interblock $CZ$ gates are also fault-tolerant. However, to ensure intrablock FT $CZ$ gates by the above Proposition, we need to build the tensor network with a base code that can be shortened to a $[[14,2, 3]]$ code. This is not possible with the QRM seed~\cite{Nezami_2022}. It is easy to build one up however, as we did previously by single tracing two 15 QRM codes into a single $[[28,2,3]]$ code. One can easily check that this tensor is still 1-cleanable with respect to $CZ$. By Lemma~\ref{lemma:branchcode}, the intrablock $CZ$ gates are also fault-tolerant if using this $28$-qubit code as seed.

\subsubsection{Example: A code with fully addressable $T, CS$ and $CCZ$}

It is also straightforward to construct codes that support other types of addressable gates by choosing the right kind of seeds. Let's first construct a $(b\geq2,\mathcal U =\{T,CS,CCZ\}$)-branch code with the punctured $[[15,1,3]]$ QRM code concatenated on the branches. By Lemma~\ref{lemma:branchcode}, we know that $\forall U\in \mathcal U$ is 1-cleanable. Hence the fractal code we have by using the branch codes as the base seeds now supports fully addressable $T, CS, CCZ$ gates in addition to the Pauli gates at the expense of slightly worse parameters. For $b=3$, we have $\alpha = \beta\approx 0.224$.

In the same way, the interblock transversal gates are fault-tolerant since the base seed codes have $d=3$. To make the intrablock $CS$, $CCZ$ gates (and mesoblock $CCZ$ gates) fault-tolerant also, we use the same trick by performing a single trace on two branch codes. This makes a $[[30b-2,2,3]]$ code whose tensor satisfies the requirements of Lemma~\ref{lemma:branchcode} where $\alpha=\beta\approx 0.197$.

To achieve better asymptotic scaling on rate and distance while ensuring fault-tolerant intrablock gates, a mix of different seed codes and more efficient conjoining scheme beyond naive tree-style contraction should be used.

\section{Discussion}
In this work, we have shown how operator pushing can be generalized to symmetries that correspond to multi-qubit gates. We also show how the localized symmetries give rise to addressable single and multi-qubit and multi-qudit logical gates under the quantum lego formalism. We have constructed local gadgets for which addressable (multi-controlled) gates can be built from existing lego blocks using generalized trace. 

We then applied this formalism to produce codes with global and addressable transversal gates. Using graphical methods and deformed trace, we can generate new families of codes that have finite rate with transversal logical $T$, $SH$, $CS$, $K_3$, or $CCZ$ gates, some of which can be used for magic state distillation. 
Using operator matching, we identified criteria for building (qudit) codes with completely transversal non-Clifford gates and asymptotically unit rate using hypercubic tensor networks. Although we cannot rule out their existence, such codes appear heavily constrained --- they can only be constructed if there exist MDS codes that support transversal non-Clifford gates.

As a proof of principle, we produced two types of codes, holographic codes and iterated fractal codes, with (multi-qubit) addressable gates that are efficiently decodable using tensor weight enumerator decoders. From operator pushing, we showed that addressable inter, meso, and intrablock transversal $C^\ell P$ gates can be easily identified using our extended formalism. While single qubit gates and interblock multiqubit gates are fault-tolerant by default, addressable intrablock transversal gates can also be made fault-tolerant provided the local seed codes can be shortened to a $k=2, d\geq 3$ code. 

For finite rate holographic codes, the transversal gates are addressable to logical qubits in the same layer in the bulk whereas for black hole codes, an $\ell$-qubit transversal gate can address bulk qubits of any subset with size $\ell$. Logical gates acting on disjoint subsets of logical qubits are also parallelizable. A novel result for holographic codes, we showed that the holographic QRM codes and the heterogeneous holographic Steane-QRM code both admit transversal addressable $CZ$ gates. In particular, in the black hole code studied by \cite{Steinberg:2025kzb}, any targeted inter/intrablock logical $CZ$ is fault-tolerant and has only depth $1$, drastically reducing its overhead. 

We also built up a class of fractal codes that admit fully addressable transversal inter/meso/intrablock $C^\ell P$ gates using a simple iterative algorithm whose fault-tolerant properties and parameters are determined by the local seed codes. 

Much future exploration is needed. The current work focuses on isometric traces where no kernel emerges in the encoding map under trace. However, the QL constructions of many interesting codes, such as qLDPC codes, do have kernels. In this case, symmetries of the Choi state need not translate into the same logical actions when that state is converted into an encoding isometry. It is important, therefore, to determine the true logical actions after the kernel is removed. On the level of understanding individual lego blocks, we still lack a general theory to identify atomic codes with localized symmetries or addressability. For instance, are there codes that support localized $CZ$ or CNOT gates with $n<14$ \cite{Koh:2026vsg}? What level of transversality or addressability are allowed given the no-go results \cite{Guyot:2025lyj,BKthm,Fu:2025lbb}? 
As for lego fusion, we also lack a systematic characterization for what generalized trace can preserve or generate localized symmetries. 
Combining the recent technique for growing sparse subsystem codes \cite{Cao:2025oep} and this work will also lead to a new method for producing codes with low fault-tolerant overhead and addressable transversal gates.

Coding theoretically, tensor network methods now can be used to complement conventional techniques in generating and characterizing code families following a graph-theoretic classification. Immediate work will include, for example, codes with  transversal non-Clifford or exotic Clifford gates like $SH$, $CS$, and $K_3$.

Another open direction is to use tensor network-based approaches to constrain the existence of codes with given parameters and logical gates. For instance, it is still unclear whether qudit MDS codes with transversal non-Clifford gates exist;  Prop.~\ref{prop:MDSconstraint} may be used in conjunction with resource-theoretic bounds to constrain the existence of gates with transversal non-Clifford gate. It is also interesting to study whether these tensor network methods can be used to produce tighter bounds on triorthogonal codes, as study of small sized codes suggest that the current bound $n\geq 2k$ is rather loose. 

There are other forms of unitary symmetries that this work has not yet studied, such as those that enable code switching and code automorphisms that correspond to fold transversal gates. Both are powerful tools for generating fault-tolerant (universal) addressable gate sets. It is key to study how these symmetries may transform under conjoining \cite{Koh:2026vsg}.

Finally, a more efficient algorithm is needed for practical implementations of Quantum Lego such as \cite{pato_2025_16761072}, which requires the efficient tracking of these additional symmetries that give rise to transversal (addressable) gates in the higher levels of Clifford hierarchy. One possible extension is through the XP stabilizer formalism ~\cite{Shen:2023xmh,Webster:2022kdn,Webster_2023}. Efficient representations for addressable multi-qubit gates and more general unitary product stabilizers are also vital. On the practical front, future extensions of machine learning-based frameworks \cite{Su_2025,Mauron:2023wnl} that incorporate the recent progress in small codes \cite{Chen_2022,Koh:2026vsg}, sparse code synthesis \cite{Cao:2025oep}, tensor enumerator \cite{TNenum,QL2,Pato:2025zqj}, and quantum anticode \cite{Cao:2025iec} may be used to design device or task-tailored codes of intermediate scales that can support exotic fault-tolerant gates.  

\section*{Acknowledgement}
We thank Victor Albert, Shayan Majidy, Balint Pato, Matt Steinberg, and Hayata Yamasaki for helpful comments and discussions.

\appendix

\section{Useful statements for code fusion}

Here we develop some helpful lemmas for how the code parameters transform when fusing lego blocks. For stabilizer tensors or states, the isometric property can also be checked through Pauli operator cleaning.

\begin{definition}
    Let $\ket{V}$ be a $n$-qudit tensor and $J \subset \{1,\dots, n\}$. An operator $P$ can be \emph{cleaned} (off of $J$) if there exists unitary symmetry $U$ of $\ket{V}$ such that $ PU|V\rangle = P'|V\rangle$ where $P'$ only has support on $J^c$.
\end{definition}

 \begin{lemma}
     Let $|V\rangle$ be a stabilizer state such that no stabilizer is supported entirely on $J$. Then the map $V:\mathcal{H}_J\rightarrow \mathcal{H}_{J^c}$ which uses $J$ as input/logical legs and $J^c$ as output physical legs of the tensor is an isometry.
 \end{lemma}
\begin{proof}
    By \cite{Fattal:2004frh}, we know that $J$ must be maximally entangled with $J^c$. Hence by Lemma~\ref{lemma:isometry} $V$ is an isometry.
\end{proof}

\begin{lemma}\label{lemma:isodist}
    Let $C_1, C_2$ be $[[n_i,k_i,d_i]],~i=1,2$, stabilizer codes with parameters where each $d_i\geq 2$. Then a single trace of these two codes on any two legs $a$ and $b$ produces a code $C=C_1\wedge_{a,b} C_2$ with parameter $[[n'=n_1+n_2-2, k'=k_1+k_2, d'\geq \min\{d_1,d_2\}]]$. Furthermore, the $k_i$ logical qubits that only act non-trivially on the pre-trace block $C_i$ have post-trace distance $d'_i\geq d_i$ and post-trace logical operators that act on both $C, C'$ must have $d_{12}\geq d_1+d_2-2$.
\end{lemma}

\begin{proof}
    Without loss of generality, consider any logical Pauli operator on $C_1$. Since we have $d_1\geq 2$, any weight one error on $a$ can be cleaned as there is no logical operator supported on it. In particular, all logical operators can be pushed to the complement of $a$ because there are no weight 1 stabilizers in the code. The same argument holds for $C_2$. Therefore, after trace, all $k_1+k_2$ logical Pauli operators can be pushed to the dangling legs without ever passing through the fused $a$ and $b$. Hence $k'=k_1+k_2$ and fusion reduces the number of physical legs by two, hence $n'=n_1+n_2-2$.

    Finally, for any logical operator supported on $C_1$ that does flow from $a$ to $b$, can be cleaned off of $b$ and pushes to the rest of the dangling legs on $C_2$. Since this amounts to multiplying a Pauli error on $b$ with a stabilizer, and all stabilizers in a $d>1$ code must have weight $\geq 2$, hence this operator must have weight at least $d_1$. The same argument holds and logical operators that act non-trivially on $C_2$ will have distance at least $d_2$. For logical operators that act on both blocks, any matching operator must have weight at least $d_{12}\geq d_1+d_2-2$. Since both $d_1, d_2\geq 2$, we must have $d_{12}$ must be lower bounded by $d_1$ and $d_2$. 

\end{proof}

A simpler version of the above can be found in \cite{tnc} but it does not discuss how the code distance transforms. A more general form of the above can also be made, as long as the fusion is performed on qubits where located errors are correctable.

\begin{proposition}
\label{prop:gen_fusion}
    Consider two stabilizer codes $C$ and $C'$ parameters $[[n,k,d]]$ and $[[n',k',d']]$  with stabilizer and normalizers $S,\mathcal{N}(S)$ and $S', \mathcal{N}(S')$ respectively. Let $A\subset [n]$ and $A'\subset [n']$ be two subsets of physical legs where $|A|=|A'|$ and a fusion of $A$ and $A'$ such that each qubit in $A$ is identify with one in $A'$ with a bijective map $\phi:A\rightarrow A'$.
    
    Suppose $A,A'$ are correctable erasures and that $$m\equiv \max_{L_C\in\mathcal{N}(S), L_{C'}\in\mathcal{N}(S')} |\phi(supp(L_C))\cap supp(L_C')|,$$ then the fused code has $[[n+n'-2|A|,k+k',d'']]$ where $\min\{d,d'\}\geq d''\geq \min\{d, d'\}-m$ where
    $$
     d''_C\geq d-m,\quad
     d''_{C'}\geq d'-m,\quad
     d''_{CC'}\geq d+d'-2m.$$
    Here $d''_C$ (resp. $d''_{C'}$ and $d''_{CC'}$) is the distance of the post-trace logical operators only acting non-trivially on the logical legs of $C$ (resp. $C'$ and nontrivially split between $C$ and $C'$). 
\end{proposition}
\begin{proof}
    This follows directly from operator pushing, where the normalizers of the new code must arise from tracing together normalizers of the individual codes. Since $A,A'$ are correctable erasures, such that $\Pi_C O_A\Pi_C \propto \bar{I}_C$ for all $O_A$ supported only on $A$ (and $\Pi_C' O_A'\Pi_C' \propto \bar{I}_{C'}$ for all $O_{A'}$ supported only on $A'$), it must follow that no logical operators can be completely supported on $A$ (or $A'$). By the cleaning lemma (or complementary recovery of stabilizer codes), any Pauli operator acting on such regions can be pushed to their respective complements or annihilated by the logical identity operators.\footnote{Note that this need not imply $|A|<d$ (or $|A'|<d'$) because the intersecting support of logical operators and $A$ (or $A'$) need not coincide with $A$ (or $A'$).}

    Now consider operator pushing. Since $A$, $A'$ are correctable erasures, one can always push the logical operators from the two different code blocks out to the boundary legs without ever going through $A$ and $A'$. Hence $k''=k+k'$ as Pauli operators acting on those logical degrees of freedom can always be acted upon independently. It is clear that the number of post-trace physical legs if $n+n'-|A|-|A'|=n+n'-2|A|$.

    The logical operators that do not flow through $A$ and $A'$ do not have their weights changed. For those that do, the logical operators acting only on the pre-traced code block $C$ can have their weight reduced by at most $m$, hence retaining a weight lowered bounded by $d-m$. The same holds for the operators from $C'$ and they are lower bounded by $d'-m$. Note that these lower bounds are generally not tight, because the stabilizers used to clean them typically have support outside of $A$ and $A'$ and form only a subset of the normalizers.  In the same way, logical operators from both $C$ and $C'$ flowing through $A$ and $A'$ can have their weight reduced by at at most $2m$, obtaining a minimal weight lower bound $d+d'-2m$. 
\end{proof}

For CSS codes, it suffices to check this for the $X$ and $Z$ operators separately.
\begin{corollary}
    If fusing two CSS codes $[[n,k,d_X/d_Z]]$ and $[[n',k',d'_X/d'_Z]]$ on correctable erasures $A, A'$ where $L_X, L_Z\in\mathcal{N}(S)$ and $L'_X, L'_Z\in\mathcal{N}(S')$ are $X$-type and $Z$-type normalizers of $C$ and $C'$ respectively such that 
    
    \begin{align*}
        m_X\equiv \max_{L_X, L'_{X}} |\phi(supp(L_X))\cap supp(L'_X)|\quad\mathrm{and}\quad m_Z\equiv \max_{L_Z, L'_{Z}} |\phi(supp(L_Z))\cap supp(L'_Z)|   
    \end{align*}

    then the fused code has parameter $[[n''=n+n'-2|A|,k''=k+k',d''_X/d''_Z]]$ where $\min\{d_X,d'_X\}\geq d''_X\geq \min\{d'_X-m_X,d_X-m_X\}$ and $\min\{d_Z,d'_Z\}\geq d''_Z\geq \min\{d'_Z-m_Z,d_Z-m_Z\}$.
\end{corollary}
Since the proof follows the same argument as above except we differentiate the $X$ and $Z$ operator flows, it is left as an exercise for the reader. Note that it suffices to consider $X$ and $Z$ flows separately because $Y$ flows are simply the combination of these two.

\begin{example}
    Consider tracing $8$ legs associated with the $8$ physical qubits on the vertices of the $3$-cube on the $[[15,1,3]]$ Quantum Reed-M\"uller code code. Those $8$ qubits are correctable erasures because the minimal weight logical $Z$ and $X$ are be supported on a single face. Because the maximal overlap of the normalizers on the traced regions are $m_X=4$ for logical-$X$ and $m_Z=2$ for logical-$Z$. The bounds from above has $d''_{C,X}=d''_{C',X}\geq 3$ while $d''_{C,Z}=d''_{C',Z}\geq 1$ whereas $d''_{CC',X} \geq 6$ and $d''_{CC',Z}\geq 2$. In reality, this produces a $[[14,2,7/2]]$ code where $d_X=7, d_Z=2$. One can also arrive at the same code by dualizing one of the physical qubits on the corner of the $[[15,1,3]]$ tetrahedron into a logical qubit. The resulting code have physical qubits that can be arranged on a bilayer where there are three weight-8 $X$-type stabilizers acting on the 3-cubes while the weight-4 $Z$-stabilizers act on the the faces of the 3-cubes (Fig.~\ref{fig:1422_CZ}). It is clear that the lower bound here is not tight because we did not include the support of the stabilizers that do not overlap with the regions that are traced. 
\end{example}

\begin{remark}
    Note that requiring a set of qubits $A$ to be a correctable erasure is a weaker condition than requiring all errors on $A$ to be detectable, i.e., have non-trivial syndrome means no Pauli operators there can live in the normalizer. The latter also means that all such Pauli operators can be cleaned (or pushed) using stabilizer multiplication such that the resulting representation on $A^c$ is non-trivial. Therefore, the tensor when used as a map from $A$ to $A^c$ is an isometry for any fixed bulk/logical input in $[k]$. Such an encoding tensor for $C$ is partially maximally entangled (PME) \footnote{Not to be confused with planar maximally entangled, which refers to having all $\floor*{n/2}$ on adjacent legs be maximally entangled when the tensor and legs are arrange to lay on a plane.}, that is, the qubits that correspond to $A$ in the Choi state is maximally entangled with the rest of the system.
\end{remark}

    The statements we examined so far do not require the encoding tensor to be PME. This can be seen from the above example that traces two QRM codes on 8 legs. If we know that one of the tensors is PME, then one only needs one tensor to be PME to ensure that the total number of logical operators add under fusion. That is, there is no kernel in the encoding map after trace. Such operations have been used widely in holographic code constructions where the tensor contractions are isometric. In this case, a refined version of Theorem 1 in \cite{tnc} can be proved below, which also follows as a Corollary of Prop.~\ref{prop:gen_fusion}.

\begin{theorem}\label{thm:tnc}
    Let $C$ and $C'$ be stabilizer codes with parameters $[[n,k,d]]$ and $[[n',k',d']]$ respectively, and consider a fusion of $A\subset \{1,\dots,n\}$ and $A'\subset\{1,\dots,n'\}$ with $|A|=|A'|$. If all Pauli errors on one of $A$ and $A'$ are detectable, then the resulting code has parameter $[[n+n'-2|A|, k+k']]$. 
\end{theorem}

\begin{proof}
    Without loss of generality, suppose all Pauli errors on $A'$ are detectable. In particular, $A'$ contains no nontrivial logical operator and hence is cleanable. That is, any Pauli operator can be pushed from $A'$ to non-identity elements on $A'^c$ and the resulting tensor is an isometry going from $A'$ to the rest of the legs (see also \cite{tnc}). In particular, then fusion reduces to a generalized concatenation where the encoding map $V$ has support $A'\uplus \{1,\dots,k'\}\rightarrow A'^c$, depending on the number of logical inputs $k'$. 
\end{proof}

We can say a bit more about the distance, even though it depends on the cleanability of each code block, thus making a general distance bound less useful.

\begin{proposition}
    Continuing the notation of Theorem~\ref{thm:tnc} and let the distance of $C$, the Choi state associated with the encoding tensor induces an isometry $V$ from $A'\uplus \{1,\dots,k'\}$, which defines a $[[|A'^c|,|A'|+k']]$ stabilizer code. If this code has distance $d_V$, then the minimal weight logical operator that only act non-trivially on the $k$ logical qubits in $C$ has word distance $d_{C}''\geq d-\min\{|A|,d\}+d_V$ and $d_{C'}''\leq d'$.
\end{proposition}
\begin{proof}
    In the proof of the theorem we showed $V:A'\cup [k']\rightarrow A'^c$ was an isometry. Now, note that the support of any minimal weight logical operator from $C$ pushing through $A$ has weight at most $\min\{A,d\}$ because the flow is constricted. By definition, the word distance \cite{Harris} can include operators that act also non-trivially on $[k']$. That support is then mapped to an operator has weight at least $d_V$. Hence the final logical operator acting on $C$ has weight at least $d-\min\{|A|,d\}+d_V$.  For the logical operators on $C'$, since all errors on $A$ are detectable, it follows that none of them are in the normalizer and the logical operators must admit representations that only have support in $A'^c$ by complementarity.  Therefore, logical operators acting non-trivially on $C'$ have distance at most $d'$, which was the untraced distance. For operator flows that go through the sites of fusion, this will depend on the property of $C$. 
\end{proof}

\section{Codes with transversal SH gates}
\label{app:msd}
Since operators on the logical legs in both the line and 2d codes can be pushed to the physical legs without obstruction, the tensor network is an isometry from the logical to physical degrees of freedom. Alternatively, in the line code, we first recall that a perfect tensor with 6 legs is an isometry when any $k\leq 3$ legs are chosen as input and the rest as outputs. The line code can be written as a sequence of contractions where WLOG a tensor on left endpoint of the line is contracted with an isometry with 2 inputs next to it, then followed by another isometry of the same type until we reach the right boundary. As this is a sequence of isometric maps, the final map remains isometric. A similar proof holds for the 2d code, where we first contract the perfect tensors into a line. For each such line tensor network that is not a boundary of the 2d network, it is built by a sequence of isometric maps with 2 inputs on one endpoint and followed by isometries with 3 inputs in the ones that follow. For a size $L\times L$ 2d network, we build up such 1d networks of length $L$, which are now isometries with $2L$ inputs and $2L+2$ outputs. Contracting such 1d networks $L$ times, we generate a 2d network that's also an isometry as it is produced by a sequence of isometric maps. 

For MSD, we need to prove that inputting products of $|T_0\rangle$ (or $T_1\rangle$) can output purer copies of them, i.e. these states are stable fixed points. This corresponds to studying the behaviour of $\Pi |T_x\rangle$ where $\Pi$ is the projection onto the code subspace and $|T_x\rangle = |T_{i_1}\rangle|T_{i_2}\rangle\dots|T_{i_n}\rangle$ where $i_a=0,1$ and $x$ is the bit string generated by $i_a$. One important requirement for Bravyi-Kitaev (BK) type  distillation to work in the 5-qubit code is that in addition to having transversal SH gates, if the Hamming weight $|x|=1,4$, then $\Pi|T_x\rangle=0$ (Fig.~\ref{fig:MSD_line}a). This allows us to suppress the errors with post-selection onto the code subspace. Note that this is different from the detectability condition and the T-gate distillation like in Bravyi-Haah. As a result, we need to show this by contracting the tensors with $T_0\rangle$ or $|T_1\rangle$s on the dangling legs. 

\begin{figure}
    \centering
    \includegraphics[width=\linewidth]{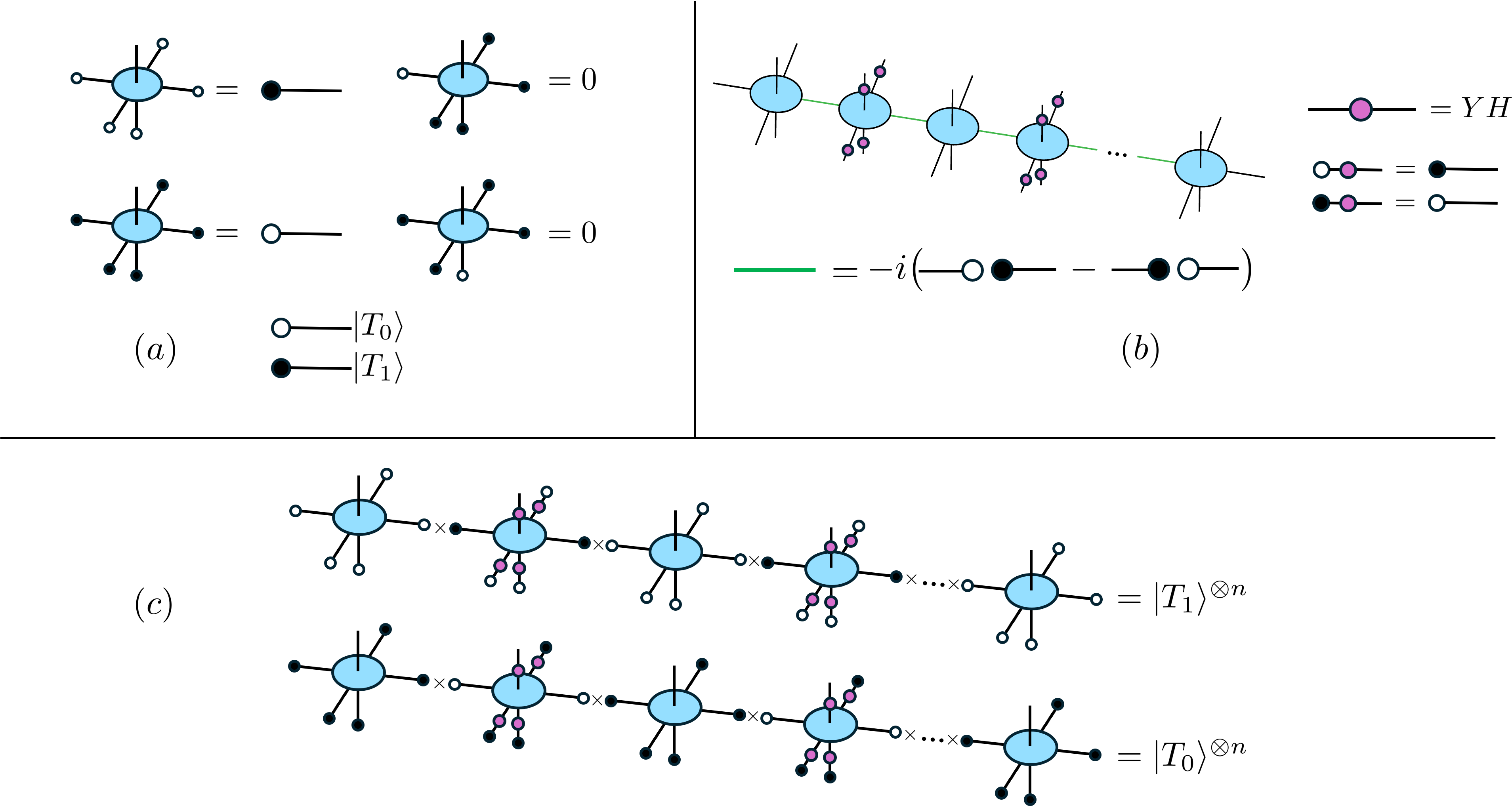}
    \caption{(a) Tensor identity of the 5 qubit code when contracting with $|T_{0,1}\rangle$. (b) $[[3L+2,L,3]]$ line code with local Clifford deformations. (c) Non-trivial contribution when $|x|=0$ or $|x|=3L+2$.}
    \label{fig:MSD_line}
\end{figure}

First note that $|\Phi_Y\rangle$ is stabilized by $-U\otimes U$ where $U|0\rangle=|T_0\rangle$, $U|1\rangle=|T_1\rangle$. Using this, we can rewrite it as $|\Phi_Y\rangle = -i(|T_0\rangle|T_1\rangle - |T_1\rangle|T_0\rangle)$ (Fig.~\ref{fig:MSD_line}b). Although the line code itself can already be used for BK style MSD, we have added local Clifford deformations to ensure that the input and outputs are made uniform. Using this to project the edges on the line code in the figure, it is clear that $|T_0\rangle^{\otimes n}, |T_1\rangle^{\otimes n}$ are stable fixed-points where the only non-trivial contributions to these inputs are given in Fig.~\ref{fig:MSD_line}c. By the identity in (a), the only nonzero contraction of the tensors must have nodes of the same color on the contracted edges in (c). Hence the black and white nodes must alternate throughout the chain by the identity in (b). 

If one of the physical nodes are fed a state with the opposite color in either the top or bottom line in (c), then that tensor must vanish by the identity in (a) and hence $\Pi |T_{|x|=1}\rangle=\Pi |T_{|x|=n-1}\rangle=0$. This implies that all terms in $\rho^{\otimes n}$ that have one $YH$ error will be filtered out after one round of distillation. Note that a tree tensor network with a bulk qubit on every node also satisfy this condition, but does not lead to higher rate. Therefore, if any single error probability is of order $p$, then one round of distillation will suppress it to $O(p^2)$.  For two or more such errors, the code does not provide this suppression. We will leave its prospect for MSD to future work.

Using the same argument, we can see that the 2d code or the max rate HaPPY codes do not satisfy this condition where $\Pi|T_{|x|=1}\rangle=\Pi|T_{|x|=n-1}\rangle=0$. Indeed, it is easy to see that $|T_{0,1}\rangle^{\otimes n}$ is not even a stable fixed point for these codes when two or more legs are contracted in every tensor. Therefore, they cannot simply be used for MSD with a BK-like protocol.

\section{Codes with weakly transversal $CZ$ gates}
\label{app:cz}

\begin{lemma}\label{lemma:weaklytransvCZ}
    Consider a CSS code, and let $A = CZ^{\otimes r}$ act on two code blocks. Then $A$ is a logical operator if and only if for every $X$-type stabilizer $s_X$ there exists a $Z$-type stabilizer with $\mathrm{supp}(A)\cap \mathrm{supp}(s_X) = \mathrm{supp}(s_Z)$.
\end{lemma}
\begin{proof}
    The Gottesman rules for $CZ$ are $XI \mapsto XZ$ and $IX\mapsto ZX$, while operators $ZI$ and $IZ$ are left fixed. Consequently, for any $Z$-type stabilizer $s_Z\otimes s_Z'$, we have $A(s_Z\otimes s_Z')A = s_Z\otimes s_Z'$ is again a stabilizer. Now suppose $s_X$ is an $X$-type stabilizer. Then $A(s_X\otimes I)A = s_X \otimes s_Z$ where $s_Z$ is the $Z$-type operator with support $\mathrm{supp}(s_Z) = \mathrm{supp}(A)\cap \mathrm{supp}(s_X)$. If $A$ is logical, then $s_X \otimes s_Z$ is a stabilizer of the product code, and hence so is $(s_X \otimes I)(s_X \otimes s_Z) = I \otimes s_Z$ and therefore $s_Z$ is a $Z$-type stabilizer of the code. Conversely, if given a stabilizer $s_X$ one has a stabilizer $s_Z$ as prescribed, then we have $A(s_X\otimes I)A = s_X \otimes s_Z$ just as above (and similarly $A(I\otimes s_X)A = s_Z \otimes s_X$). But as $s_X\otimes I$, $I\otimes s_X$, and stabilizers of the form $s_Z\otimes s_Z'$ generate the stabilizer group, we see $A$ preserves the stabilizer group and hence is a logical operator.
\end{proof}

It is well-known that the $15$-qubit QRM code has a bitwise transversal logical $T$ gate and transversal CNOT gate in addition to the usual Pauli gates. However, we note here that it also supports weakly transversal $\overline{CZ}=CZ^{\otimes 7}$ where equivalent representations can be chosen such that the support on each code block is on one of the faces from each tetrahedron. We also note that the two-block tensor has a symmetry where $CZ^{\otimes 8}$ connecting the 8 vertices of two subcubes within each code block leaves the tensor invariant. This can be easily checked using the above lemma and by conjugating the appropriate (logical) Pauli operators with logical $CZ$s. 

There are also alternative but equivalent forms of this tensor that can be worth examining with a bit more detail. As a 16-qubit tensor of a stabilizer state $[[16,0,4]]$, this implies that two copies of this stabilizer state is stabilized by $CZ^{\otimes 8}$ acting on across any of the vertices of the 3 dimensional subcubes. (See Fig. 4 right of \cite{raussendorf2012qec}.) 

Another interesting form is the $[[14,2,2]]$ code one obtains by dualizing two legs of the tensor as logical degrees of freedom. This can be accomplished by choosing 2 outer corners of the hypercube as the logical qubits. In the tetrahedral arrangement, we can choose any corner qubit in addition to the original logical qubit (or the qubit in the interior of the tetrahedron). These two choices are equivalent up to shuffling of the qubit locations, but the code can be written almost as two layers of the Steane code (Fig.~\ref{fig:1422_CZ}). In this case, the logical $CZ$ between two encoded qubits is nothing but the pairwise $CZ^{\otimes 7}$ across two layers. Note that it preserves all stabilizers because $CZ$ maps $XX\rightarrow -YY, YY\rightarrow -XX, IZ\rightarrow IZ, ZI\rightarrow ZI$. This preserves all stabilizers acting on the 8-vertex subcubes. It also preserves the Z type stabilizers by default. As for logical operators, let's label the qubits as in the figure where the top layer has qubits $Q_1=\{1,\dots,7\}$ while the bottom layer $Q_2=\{8,\dots, 14\}$. Notice that $\bar{P}_j=\bigotimes_{i\in Q_j} P_i$, where $P=I, X,Y,Z$ and $j=1,2$.

These statements can also be derived by noting that the QRM 15 encoding tensor (up to local Pauli deformation) admits a globally transversal interblock $CS^{\otimes 16}$ symmetry. Using our Lemma~\ref{lemma:phasegate_lemma} and the fact that logical $X$ acts on a face in the 15-qubit code, we can easily recover that $CZ$s are locally transversal.

The Steane code has a bitwise transversal $CZ$. It doesn't have weakly transversal $CZ$ or similar stabilizers like the QRM tensor above when applying Lemma~\ref{lemma:weaklytransvCZ}. However, when dualized to the $[[6,2,2]]$ code, it admits a depth-1 transversal $CZ$ gate which implements the intrablock logical $CZ$ on the two encoded qubits. To see this, pick any of the corner qubit in the $[[7,1,3]]$ code as a logical qubit, then we convert the code into a bilayer code where the first layer has qubits $1-3$ and the second layer has qubits $4-6$. Then $\bar{X}_1=X_1X_2X_3, \bar{Z}_1=Z_1Z_2Z_3$, $\bar{X}_2=X_4X_5X_6, \bar{Z}_2=Z_4Z_5Z_6$. The stabilizer group $S=\langle X_1X_2X_4X_5, X_2X_3X_5X_6, Z_1Z_2Z_4Z_5, Z_2Z_3Z_5Z_6\rangle$ which come from the stabilizer of the two remaining plaquettes on the 2d color code. Then $\overline{CZ}=CZ^{\otimes 3}$ where the physical $CZ$s can act on the pairs $(1,4), (2,5), (3,6)$ though some other configurations are also allowed by the permutation symmetries that preserves the code. We see that it preserves the stabilizers because $XXXX,ZZZZ,YYYY$ are stabilizers on any weight 4 plaquette. Conjugation on the weight 3 logical operators also produce consistent Pauli operators (Fig.~\ref{fig:multicopy_tensor}d). 

\section{Example of single-qubit symmetries in small legos}

The existence of transversal single qubit operations is most easily stated using the encoding tensor of the stabilizer code. As the encoding tensor is simply a stabilizer state, the transversal operation now takes the form of a symmetry of this state. Note that the existence of such a transversal symmetry is LOCC-invariant: given such a symmetry, permuting the labeling of the legs, or conjugation by a local unitary, preserves transversality. So to classify the possible symmetries, we can focus on the LOCC equivalence classes of stabilizer states. For tractability we will instead focus on logically Clifford equivalence classes, namely the graph states \cite{hein2006entanglement}. As weight enumerators are also LOCC-invariant \cite{rains2002quantum, TNenum}, we expect the enumerator of a graph states to play a role in determining its symmetries and indeed this is the case as we will see shortly.

\begin{table}
\begin{tabular}{|c|c|c||c|c|c|}\hline
No. & Example & Enumerator & No. & Example & Enumerator \\\hline
1 &
\begin{tikzpicture}[scale=0.2]
    \draw (0,0) -- (1,0);
    \draw[fill=black] (0,0) circle(0.2);
    \draw[fill=black] (1,0) circle(0.2);
\end{tikzpicture}
& $1+3z^2$ & 
5 & 
\begin{tikzpicture}[scale=0.2]
    \draw (1,0.5) -- (0,0) -- (1,-0.5);
    \draw (-1,0.5) -- (0,0) -- (-1,-0.5); 
    \draw[fill=black] (-1,-0.5) circle(0.2);
    \draw[fill=black] (-1,0.5) circle(0.2);
    \draw[fill=black] (0,0) circle(0.2);
    \draw[fill=black] (1,0.5) circle(0.2);
    \draw[fill=black] (1,-0.5) circle(0.2);
\end{tikzpicture}
& $1 + 10z^2 + 5z^4 + 16z^5$ \\\hline
2 & 
\begin{tikzpicture}[scale=0.2]
    \draw (1,0.5) -- (0,0) -- (1,-0.5);
    \draw[fill=black] (0,0) circle(0.2);
    \draw[fill=black] (1,0.5) circle(0.2);
    \draw[fill=black] (1,-0.5) circle(0.2);
\end{tikzpicture}
& $1 + 3z^2 + 4z^3$ & 
6 & 
\begin{tikzpicture}[scale=0.2]
    \draw (0,0) -- (1,0) -- (2,0);
    \draw (-1,0.5) -- (0,0) -- (-1,-0.5); 
    \draw[fill=black] (-1,-0.5) circle(0.2);
    \draw[fill=black] (-1,0.5) circle(0.2);
    \draw[fill=black] (0,0) circle(0.2);
    \draw[fill=black] (1,0) circle(0.2);
    \draw[fill=black] (2,0) circle(0.2);
\end{tikzpicture}
& $1 + 4z^2 + 6z^3 + 11z^4 + 10z^5$ \\\hline
3 & 
\begin{tikzpicture}[scale=0.2]
    \draw (1,0.5) -- (0,0) -- (1,-0.5);
    \draw (-1,0) -- (0,0); 
    \draw[fill=black] (-1,0) circle(0.2);
    \draw[fill=black] (0,0) circle(0.2);
    \draw[fill=black] (1,0.5) circle(0.2);
    \draw[fill=black] (1,-0.5) circle(0.2);
\end{tikzpicture}
& $1 + 6z^2 + 9z^4$ & 
7 & 
\begin{tikzpicture}[scale=0.2]
    \draw (0,0) -- (1,0.5) -- (2,0.5);
    \draw (0,0) -- (1,-0.5) -- (2,-0.5);
    \draw[fill=black] (0,0) circle(0.2);
    \draw[fill=black] (1,0.5) circle(0.2);
    \draw[fill=black] (2,0.5) circle(0.2);
    \draw[fill=black] (1,-0.5) circle(0.2);
    \draw[fill=black] (2,-0.5) circle(0.2);
\end{tikzpicture}
& $1 + 2z^2 + 8z^3 + 13z^4 + 8z^5$ \\\hline
4 & 
\begin{tikzpicture}[scale=0.2]
    \draw (1,0.5) -- (0,0.5) -- (0,-0.5) -- (1,-0.5);
    \draw[fill=black] (0,-0.5) circle(0.2);
    \draw[fill=black] (0,0.5) circle(0.2);
    \draw[fill=black] (1,0.5) circle(0.2);
    \draw[fill=black] (1,-0.5) circle(0.2);
\end{tikzpicture}
& $1 + 2z^2 + 8z^3 + 5z^4$ & 
8 & 
\begin{tikzpicture}[scale=0.2]
    \draw (0,0) -- (1,0.5) -- (2,0.5);
    \draw (0,0) -- (1,-0.5) -- (2,-0.5) -- (2,0.5);
    \draw[fill=black] (0,0) circle(0.2);
    \draw[fill=black] (1,0.5) circle(0.2);
    \draw[fill=black] (2,0.5) circle(0.2);
    \draw[fill=black] (1,-0.5) circle(0.2);
    \draw[fill=black] (2,-0.5) circle(0.2);
\end{tikzpicture}
&  $1 + 10z^3 + 15z^4 + 6z^5$ \\\hline
\end{tabular}
\caption{LC-classes of 2-qubit through 5-qubit graph states.}\label{table:graph_states}
\end{table}

To illustrate this, consider the simple quantum repetition code in the computational basis; its encoding tensor is the GHZ state $\ket{GHZ} = \frac{1}{\sqrt{2}}(\ket{000} + \ket{111})$. As presented this is the stabilizer state of $\mathcal{S} = \langle XXX, ZZI, ZIZ \rangle$. We convert this into graph state \#2 on Table~\ref{table:graph_states} by conjugating by $I\otimes H\otimes H$ having stabilizer $\mathcal{S}_{\#2} = \langle XZZ, ZXI, ZIX\rangle$. A transversal symmetry of this state is a $U = U_1 \otimes U_2 \otimes U_3$ such that $U^\dagger \Pi_{\#2} U = \Pi_{\#2}$, where $\Pi_{\#2} = \frac{1}{8} \sum_{S \in \mathcal{S}_{\#2}} S$.

Note that for a general $U$ we have $U^\dagger S U$ is a linear combination of Pauli operators. Yet each  Pauli in this sum will have precisely the same support (in the sense of Hamming weight) as $S$. Hence we decompose the projection $\Pi_{\#2}$ according to support sets,
\begin{equation}\label{equation:projection2}
    \Pi_{\#2} = \tfrac{1}{8}\left([I] + [XZI] + [XIX] + [IZX] + [ZXZ + YYZ + ZYY + YXY]\right).
\end{equation}
Each of these bracketed sums of operators must be preserved under conjugation by $U$. Clearly $U^\dagger I U = I$, so turning to the next term in this sum:
\begin{equation}
    U^\dagger (X\otimes Z \otimes I) U = U_1^\dagger X U \otimes U_2^\dagger Z U \otimes I = X\otimes Z \otimes I.
\end{equation}
Hence we must have $U_1^\dagger X U_1 = \lambda X$ and $U_2^\dagger Z U_2 = \mu Z$, where $\lambda\mu = 1$. As $X^2 = Z^2 = I$ we must have $\lambda = \mu = \pm 1$. Therefore $U_1 = R_X(\theta_1) P_1$ and $U_2 = R_Z(\theta_2)P_2$ where $\omega(X,P_1) = \omega(Z,P_2)$. Note that this last equality can be restated as $P_1\otimes P_2$ commutes with $X \otimes Z$.

By exactly the same reasoning, from the next two terms in sum (\ref{equation:projection2}) we discover $U_3 = R_X(\theta_3)P_3$ where $P_1\otimes P_3$ commutes with $X\otimes X$ and $P_2\otimes P_3$ commutes with $Z\otimes X$. That is, we must have
\begin{equation}
    U = (R_X(\theta_1)\otimes R_Z(\theta_2)\otimes R_X(\theta_3))S \text{ where $S\in \mathcal{S}_{\#2}$.}
\end{equation}
Finally we need that $U^\dagger (ZXZ + YYZ + ZYY + YXY) U = ZXZ + YYZ + ZYY + YXY$. This can done by direct computation: we leave it to the reader to verify that this holds precisely when $\sin(\theta_1 + \theta_2 + \theta_3) = 0$. Therefore Figure~\ref{fig:graph_state2_transversal}b captures transversal symmetries of the graph state.

\begin{figure}\centering
    \begin{tikzpicture}
        \draw (0,1.55) node {a)};
        \draw (2,1) -- (0,0) -- (2,-1);
        \draw (0,0) node[circle, draw, fill=white] {I};
        \draw (2,1) node[circle, draw, fill=white] {H};
        \draw (2,-1) node[circle, draw, fill=white] {H};
        \draw (0.2,0.6) node {\small $1$};
        \draw (2.2,0.4) node {\small $2$};
        \draw (1.4,-1) node {\small $3$};
        \draw (0,-2) node {};
    \end{tikzpicture}
    \qquad\qquad
    \begin{tikzpicture}
        \draw (-2,1.4) node {b)};
        \draw (0,0) -- (0,1);
        \draw (0,0) -- (-1,-1);
        \draw (0,0) -- (1,-1);
        \draw (0,0) node[circle, draw, fill=white] {\#2};
        \draw (0.2,0.6) node {\small $1$};
        \draw (-0.6,-0.2) node {\small $2$};
        \draw (0.6,-0.2) node {\small $3$};
        \draw (0,1.3) node {$R_Z(\theta_1)$};
        \draw (-1.2,-1.3) node {$R_X(\theta_2)$};
        \draw (1.2,-1.3) node {$R_X(\theta_3)$};
        \draw (0,-2) node {$\theta_1 + \theta_2 + \theta_3 = 0 \pmod{\pi}$};
    \end{tikzpicture}
    \qquad\qquad
    \begin{tikzpicture}
        \draw (-2,1.4) node {c)};
        \draw (0,0) -- (0,1);
        \draw (0,0) -- (-1,-1);
        \draw (0,0) -- (1,-1);
        \draw (0,0) node[circle, draw, fill=white] {GHZ};
        \draw (0.3,0.7) node {\small $1$};
        \draw (-0.7,-0.3) node {\small $2$};
        \draw (0.7,-0.3) node {\small $3$};
        \draw (0,1.3) node {$R_Z(\theta_1)$};
        \draw (-1.2,-1.3) node {$R_Z(\theta_2)$};
        \draw (1.2,-1.3) node {$R_Z(\theta_3)$};
        \draw (0,-2) node {$\theta_1 + \theta_2 + \theta_3 = 0 \pmod{\pi}$};
    \end{tikzpicture}
    \caption{ (a) We label the nodes of graph state \#2 to illustrate how we transform $\ket{GHZ}$ to this state. (b) Transversal symmetries of the \#2 graph state (modulo stabilizers). (c) Transversal symmetries of the GHZ state obtained from the \#2 graph state (modulo stabilizers).} \label{fig:graph_state2_transversal}
\end{figure}
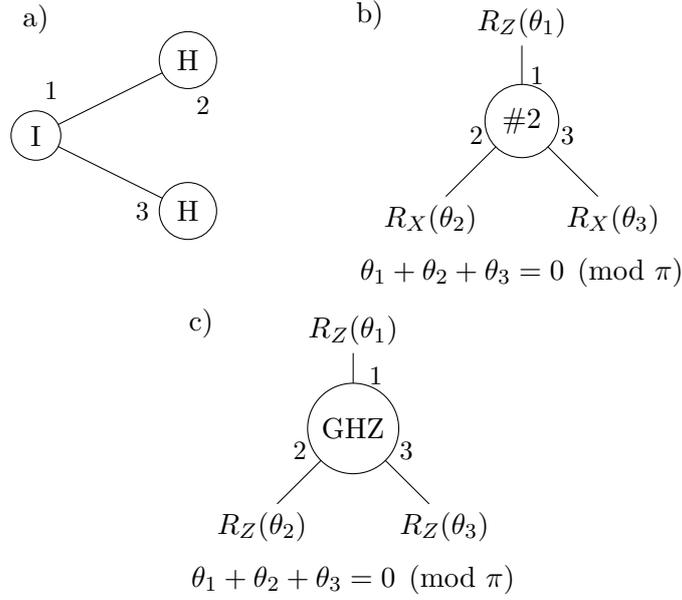

Turning to swap symmetries, it is clearly that swapping qubits 2 and 3 is a trivial symmetry as it preserves that stabilizer as presented. However swapping qubits 1 and 2 does not: it produces the stabilizer $\mathcal{S} = \langle ZXZ, XZI, IZX \rangle$, which is in the same LC equivalence class. To recover the specific graph state \#2, we apply $H \otimes H \otimes I$ which maps $\mathcal{S}$ to $\mathcal{S}_{\#2}$.

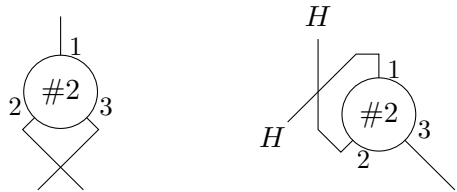
\begin{figure}[b]\centering
    \begin{tikzpicture}
        \draw (0,0) -- (0,1);
        \draw (0,0) -- (-0.5,-0.5) -- (0.3,-1.3);
        \draw (0,0) -- (0.5,-0.5) -- (-0.3,-1.3);
        \draw (0,0) node[circle, draw, fill=white] {\#2};
        \draw (0.2,0.6) node {\small $1$};
        \draw (-0.6,-0.2) node {\small $2$};
        \draw (0.6,-0.2) node {\small $3$};
    \end{tikzpicture}
    \qquad\qquad
    \begin{tikzpicture}
        \draw (0,0) -- (0,0.8) -- (-0.3,0.8) -- (-1.2,-0.1);
        \draw (0,0) -- (-0.5,-0.5) -- (-0.8,-0.2) -- (-0.8,1);
        \draw (0,0) -- (1,-1);
        \draw (0,0) node[circle, draw, fill=white] {\#2};
        \draw (0.2,0.6) node {\small $1$};
        \draw (-0.2,-0.6) node {\small $2$};
        \draw (0.6,-0.2) node {\small $3$};
        \draw (-0.8,1.3) node {$H$};
        \draw (-1.4,-0.3) node {$H$};
    \end{tikzpicture}    
    \caption{Swap symmetries of the graph state \#2.}
    \label{fig:enter-label}
\end{figure}

\section{Marked graph states and graph moves}

In our construction of targeted transversal one-qubit Clifford gates, we rely heavily on the formalism of graph states and their local Clifford transformations. While locally-Clifford equivalent graph states are considered the same, fairly so, for our purposes we would like to track the impact of a local Clifford transformation on our original stabilizer code. Hence we extend our formalism to ``marked'' graph states, which additionally carry information about the local Clifford transformation from our original encoding tensor into the graph state stabilizer. Hence with this additional information we can pull-back local Clifford transformations of the graph states to local Clifford operations on our codes.

By way of example, let us consider the simple 2-qubit quantum repetition code. The stabilizer of the code is $\mathcal{S} = \left\langle ZZ \right\rangle$ with representative logical operations $\bar{X} = XX$ and $\bar{Z} = ZI$. The encoding tensor of this code is a stabilizer state on a three-qubit Hilbert space, with stabilizer $\mathcal{S}_{\text{enc}} = \langle ZZI, XXX, ZIZ \rangle$. In this case we identify the third qubit of our larger Hilbert space with the logical qubit of the code. Every stabilizer state is locally Clifford (LC) equivalent to a graph state, however not uniquely so. Each LC-equivalence class of graph states typically contains states for many different graphs. However it is unique if we specify the LC-transformation that takes our encoding tensor state into a graph state form. For example, consider reordering our stabilizer generators as $\mathcal{S}_{\text{enc}} = \langle XXX, ZZI, ZIZ \rangle$. Then applying Hadamard gates to the second and third qubits gives a graph state stabilizer $\langle XZZ, ZXI, ZIX \rangle$, which corresponds to graph state $\ket{\:\begin{tikzpicture}[scale=0.2]
    \draw (1,0.5) -- (0,0) -- (1,-0.5);
    \draw[fill=black] (0,0) circle(0.2);
    \draw[fill=black] (1,0.5) circle(0.2);
    \draw[fill=black] (1,-0.5) circle(0.2);
\end{tikzpicture}\:}$. Hence our marked graph state consists of a labeling of the nodes of graph $\begin{tikzpicture}[scale=0.2]
    \draw (1,0.5) -- (0,0) -- (1,-0.5);
    \draw[fill=black] (0,0) circle(0.2);
    \draw[fill=black] (1,0.5) circle(0.2);
    \draw[fill=black] (1,-0.5) circle(0.2);
\end{tikzpicture}$
according to which qubits it represents and what Clifford operation we used to reduce our stabilizer state to the graph's normal form. In this case, we see the marked graph state of our quantum repetition code in Figure~\ref{fig:211-overview}c.

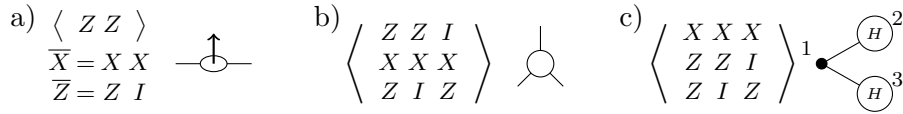
\begin{figure}[h]\centering
    \begin{tikzpicture}
    \draw (-3,0.6) node {a)};
        \draw (-2,0.5) node {\footnotesize $\left\langle\begin{array}{c@{\;}c} Z & Z\end{array}
        \right\rangle$};
        \draw (-2,-0.2) node {\footnotesize $\begin{array}{r@{\:}c@{\:}c@{\:}c} \overline{X} &=& X&X\\ \overline{Z} &=& Z&I\end{array}$};
        \draw (-1,0) -- (0,0);
        \draw[fill=white] (-0.5,0) ellipse (5pt and 3pt);
        \draw[->,thick] (-0.5,0) -- (-0.5,0.4);
    \draw (1,0.6) node {b)};
        \draw (2.2,0) node {\footnotesize $\left\langle\begin{array}{c@{\;}c@{\;}c} Z & Z & I\\ X & X & X\\ Z & I & Z\end{array}\right\rangle$};
        \draw (3.8,0) -- (3.5,-0.3);
        \draw (3.8,0) -- (4.1,-0.3);
        \draw (3.8,0) -- (3.8,0.5);
        \draw[fill=white] (3.8,0) circle (5pt);
    \draw (5,0.6) node {c)};
        \draw (6.2,0) node {\footnotesize $\left\langle\begin{array}{c@{\;}c@{\;}c} X & X & X\\ Z & Z & I\\ Z & I & Z\end{array}\right\rangle$};
        \draw (8.2,0.4) -- (7.5,0) -- (8.2,-0.4);
        \draw[fill=black] (7.5,0) circle (2pt);
        \draw (7.3,0.2) node {\scriptsize $1$};
        \draw (8.2,0.4) node[draw,circle,fill=white,inner sep=2pt] {\tiny $H$};
        \draw (8.5,0.6) node {\scriptsize $2$};
        \draw (8.2,-0.4) node[draw,circle,fill=white,inner sep=2pt] {\tiny $H$};
        \draw (8.5,-0.2) node {\scriptsize $3$};
    \end{tikzpicture}
    \caption{Views of the $[[2,1,1]]$ quantum repetition code. (a) The stabilizer and logical operations as a quantum code. (b) The stabilizer of the code's encoding tensor as a stabilizer state. (c) The same stabilizer (echelonized) as marked graph state.}
    \label{fig:211-overview}
\end{figure}

\bibliography{ref}
\bibliographystyle{unsrt}

\end{document}